\documentclass[final,3p,times]{elsarticle}
\usepackage{amssymb}
\usepackage{algorithm}
\usepackage{algpseudocode}
\usepackage{longtable}
\usepackage{array}
\usepackage{booktabs}
\usepackage{longtable}
\usepackage{amsmath}
\usepackage{enumitem}
\usepackage{graphicx}
\usepackage{subcaption}
\usepackage{float}
\usepackage{caption}
\usepackage{amsthm}
\newtheorem{axiom}{Axiom}[section]
\newtheorem{assumption}{Assumption}[section]
\newtheorem{theorem}{Theorem}[section]
\newtheorem{lemma}{Lemma}[section]
\newtheorem{proposition}{Proposition}[section]
\newtheorem{corollary}{Corollary}[section]
\newtheorem{definition}{Definition}[section]
\newtheorem{remark}{Remark}[section]
\newtheorem{example}{Example}[section]

\newcommand{\Atom}{\text{Atom}}
\newcommand{\Atoms}{\text{Atoms}}

\newcommand{\enc}{\text{enc}}

\newcommand{\poly}{\operatorname{poly}}

\newcommand{\E}{\mathbb{E}}
\newcommand{\Ess}{\operatorname{Ess}}

\newcommand{\timeindex}{\mathsf{time}}

\journal{Nuclear Physics B}

\begin{document}

\begin{frontmatter}



\title{Derivation Depth as an Information Metric:\\
Axioms, Coding Theorems, and Storage--Computation Tradeoffs}


\author[label1]{Siyuan Qiu} %% Author name
\ead{qsyqsyqsy@sjtu.edu.cn}

\author[label2]{Rui Wang\corref{cor1}} %% Author name
\ead{wangrui12@sjtu.edu.cn}
\cortext[cor1]{Corresponding author.}

\author[label3]{Jianfeng Xu\corref{cor1}}
\ead{xujf@sjtu.edu.cn}

\affiliation[label1]{organization={School of Computer Science},%Department and Organization
            addressline={Shanghai Jiao Tong University}, 
            city={Shanghai},
            postcode={200030}, 
            country={China}}
\affiliation[label2]{organization={School of Computer Science},%Department and Organization
            addressline={Shanghai Jiao Tong University}, 
            city={Shanghai},
            postcode={200030}, 
            country={China}}
\affiliation[label3]{organization={Koguan School of Law, China Institute for Smart Justice, School of Computer Science},%Department and Organization
            addressline={Shanghai Jiao Tong University}, 
            city={Shanghai},
            postcode={200030}, 
            country={China}}

%% Abstract
\begin{abstract}
We study online inferential cost in logical query answering over finite knowledge bases. Our central object is a premise-relative derivation depth, which measures the residual derivational effort of answering a query once a chosen premise base is treated as directly available. This notion is intended as an operational cost metric, not as a worst-case proof-complexity measure.
Within a finite effective logical setting, we distinguish intrinsic semantic content from operationally stored shortcuts by comparing derivation depth relative to a canonical irredundant semantic core and relative to the full stored base. We then develop a coding-theoretic account of derivation cost by encoding derivation traces. Under explicit structural conditions—effective predecessor computability, a serializability relation between depth and shortest traces, and a richness/incompressibility regime—we obtain two-sided bounds relating generic derivation cost to conditional Kolmogorov complexity, up to a logarithmic addressing term. For Boolean conjunctive workloads, we also give a step-accurate trace-based instantiation.
This depth--information correspondence yields an information-theoretic view of storage--computation tradeoffs: storing query results pays description length to reduce future derivation cost. We derive a frequency-weighted break-even rule for caching, formulate budgeted cache allocation as maximizing expected depth reduction, and extend the framework to noisy premise bases with loss and pollution. The contribution is a common framework linking online inference cost, conditional description length, and storage decisions in logical knowledge systems.
\end{abstract}



%% Keywords

\end{frontmatter}



%  I.  INTRODUCTION
%=============================================================
\section{Introduction}
\label{sec:introduction}

Modern information systems increasingly operate under a coupled constraint:
they must answer rich logical queries over large knowledge bases with low latency, while facing tight budgets on
persistent storage, bandwidth, and online computation~\cite{nguyen2025representation,cao2025saqe}.
In practice, this constraint is managed via a continuum of mechanisms---from rule compilation and
materialized views to memoization and cache hierarchies---all of which can be summarized as a single design question:
\emph{which information should be stored (and in what form), and which information should be derived on demand,
so as to minimize total resource cost under the system's access pattern?}
This question arises in database view selection and physical design~\cite{chaudhuri1995optimizing,gupta1997materialized},
in caching and paging systems~\cite{belady1966study,sleator1985amortized,megiddo2003arc},
and more recently in retrieval-augmented architectures for large models~\cite{lewis2020retrieval,karpukhin2020dense}.
It also appears, in a broader form, in neuro-symbolic AI as a recurring tension between storing explicit symbolic
structures and recomputing (or approximating) them via learned computation~\cite{garcez2023neurosymbolic,lamb2020graph}. Our goal is to give this tension a common formal language.

\subsection{Problem setting and motivating tension}
\label{subsec:intro-problem}

A query answering system implicitly maintains two kinds of ``knowledge'':
(i) \emph{semantic premises} that constitute an irredundant core of what is known, and
(ii) \emph{operational shortcuts} that are stored to reduce online inference.
When the query workload changes, shortcuts may become wasteful; when the base grows, recomputation becomes expensive.
A satisfactory theory should therefore separate at least three notions:
\begin{enumerate}[label=\textup{(\roman*)}, leftmargin=2.0em]
\item \emph{Intrinsic content:} what the system logically entails from a minimal semantic core;
\item \emph{Operational cost:} how expensive it is to answer queries under the currently stored shortcuts;
\item \emph{Economic tradeoff:} when it is beneficial to pay storage to reduce online inference, given query frequencies.
\end{enumerate}

The central technical idea of this paper is that online inferential cost can be captured by a premise-relative notion of derivational depth. To formalize this
distinction, we model knowledge bases and queries inside a fixed effective logical substrate
(Section~\ref{sec:model} and Section~\ref{subsec:logic}), define a canonical irredundant semantic core
\(\mathrm{Atom}(S_O)\) (Definition~\ref{def:atom-so}), and measure online inference by a base-relative derivation depth
\(\mathrm{Dd}(\cdot\mid B)\) (Definition~\ref{def:derivation-depth}).
This supports two baselines: the \emph{intrinsic} baseline \(\mathrm{Atom}(S_O)\) and the
\emph{operational} baseline \(S_O\), yielding intrinsic vs.\ operational depths
(Definition~\ref{def:int-op-depth} and Remark~\ref{rem:int-op-monotone}).

Our development follows a "semantic-to-operational" pipeline. Inspired by Objective Information Theory (OIT) \cite{xu2014objective,xu2024research,xu2025general}, we treat inference cost not as an absolute property of the query, but as a relative measure dependent on the available knowledge base.

\subsection{Related work and conceptual positioning}
\label{subsec:intro-related}

\paragraph{Logical and semantic foundations}
We work in a fixed effective logical substrate, instantiated here as FO(LFP) \cite{immerman1982relational,bruse2025characterizing}, which over ordered finite structures captures polynomial-time computation via the Immerman--Vardi theorem (see Section \ref{subsec:logic} and also \cite{immerman1999descriptive,vardi1982complexity}). This provides a machine-independent semantic setting in which satisfaction on finite structures is decidable and proof checking is effective, supporting our definitions of base-relative derivation depth, derivation traces, and their encodings.

\paragraph{Information measures and algorithmic information}
Shannon entropy \cite{shannon1948mathematical,cover2006elements,perrone2023markov} quantifies average information under a distribution, whereas algorithmic information theory measures individual description length via Kolmogorov complexity \cite{kolmogorov1965three,solomonoff1964formal,chaitin1969simplicity}. MDL \cite{rissanen1978modeling,grunwald2007minimum} develops description-length principles for modeling, but does not directly address the online inferential cost of deriving consequences from a premise base. Bennett's logical depth \cite{bennett1988logical} is closer in spirit: it connects description length and computation time. Our notion differs in being explicitly \emph{base-relative} and query-answering oriented: we relate \(\mathrm{Dd}(\cdot \mid B)\) to conditional Kolmogorov complexity through derivation-trace encodings and incompressibility arguments (Section \ref{sec:derivation-entropy-ideal}), and compare it with Bennett-style depth under explicit simulation assumptions (Section \ref{subsec:bennett}).

\paragraph{Storage--computation tradeoffs}
Time--space tradeoffs are classical in complexity theory \cite{yao1985should,beame1991general,borodin2023time}. In databases and systems, materialized views, physical design, and caching policies trade storage for query latency ~\cite{chaudhuri1995optimizing,gupta1997materialized,belady1966study,sleator1985amortized,megiddo2003arc,selinger1979access,choi2022learning}. Our perspective is complementary: rather than proposing a new cache policy, we give an information-theoretic account of when storing consequences is preferable to deriving them online. The framework distinguishes intrinsic semantic content from operational shortcuts via \(\mathrm{Atom}(S_O)\), links online derivation cost to conditional description length for generic information-rich queries, and yields a frequency-weighted break-even principle for caching (Section \ref{sec:tradeoff-ideal}).

\paragraph{System-wide allocation}
Budgeted cache selection naturally connects to submodular optimization, which provides standard approximation guarantees for many resource-allocation problems \cite{nemhauser1978analysis,krause2014submodular,ikeuchi2024network}. In our setting, cache allocation is formulated as maximizing expected reduction in derivation depth under a knapsack budget, so standard greedy methods yield \((1-1/e)\)-type guarantees under a diminishing-returns assumption (Section \ref{subsec:system-allocation}; \cite{sviridenko2004note}). For robust or non-monotone variants, we rely on existing non-monotone submodular optimization methods \cite{buchbinder2014submodular}.

\paragraph{Large-scale semantic systems and retrieval augmentation (motivation)}
Large knowledge graphs and semantic resources~\cite{bollacker2008freebase,carlson2010toward,miller1995wordnet} motivate
settings where storing everything explicitly is expensive.
Representation-learning and embedding methods~\cite{bordes2013translating,huang2024deep} compress knowledge at the cost of exactness.
Retrieval-augmented systems~\cite{lewis2020retrieval,karpukhin2020dense} and large language models
(e.g., Transformers~\cite{vaswani2017attention} and large-scale language modeling~\cite{brown2020language,madden2025next})
further emphasize architectural choices between parametric storage and non-parametric retrieval/derivation.\newline

In this sense, our contribution is not primarily a new caching policy, but a principled information-theoretic account of when storage should substitute for online inference.

\subsection{Main contributions}
The paper makes four contributions.

\begin{itemize}
    \item \textbf{Premise-relative inferential cost.}
    We introduce derivation depth relative to an available premise base as a measure of residual online inferential effort in logical query answering. The same query may have very different costs depending on which premises are treated as directly available; this premise relativity is essential for modeling materialization, caching, and stored shortcuts.

    \item \textbf{Depth--information correspondence under explicit conditions.}
    We develop a coding-theoretic account of derivation cost by encoding derivation traces, and show that under explicit structural assumptions—effective predecessor computability, a serializability relation between depth and shortest traces, and a suitable richness/incompressibility regime—generic derivation cost is tightly connected to conditional Kolmogorov complexity up to a logarithmic addressing term.

    \item \textbf{Intrinsic semantics versus operational shortcuts.}
    We formalize a canonical irredundant semantic core for a stored knowledge base and use it to distinguish intrinsic semantic content from operationally materialized shortcuts. This supports comparisons between depth relative to the core and depth relative to the full stored base.

    \item \textbf{Systems and robustness consequences.}
    Building on the depth--information correspondence, we derive a frequency-weighted storage--computation tradeoff, formulate budgeted cache allocation as maximizing expected reduction in derivation depth, and extend the same base-relative framework to noisy premise bases with loss and pollution.
\end{itemize}
The paper does not propose a universal proof metric, nor does it claim that derivation depth generally coincides with proof length, execution time, or algorithmic information. Rather, it identifies a practically and theoretically meaningful regime in which residual derivational effort can be related to conditional description length. The optimization and noise results should be read as consequences of this formalism, not as stand-alone algorithmic novelties.

\subsection{Organization}
\label{subsec:intro-organization}

Section~\ref{sec:model} introduces the system model, axioms, and the intrinsic/operational separation via
\(\mathrm{Atom}(S_O)\) and \(\mathrm{Dd}(\cdot\mid B)\).
Section~\ref{sec:derivation-entropy-ideal} develops derivation depth as an information metric and proves coding-style bounds
connecting depth to conditional algorithmic information.
Section~\ref{sec:tradeoff-ideal} derives frequency-weighted storage--computation tradeoffs and system-wide allocation
principles, including greedy approximation guarantees under submodularity assumptions.
Section~\ref{sec:noisy-extensions} extends the framework to noisy premise bases.
Finally, Section~\ref{sec:conclusion} concludes and discusses limitations and directions.

%=============================================================
%  II. SYSTEM MODEL AND AXIOMATIC FOUNDATIONS
%=============================================================
\section{System Model and Axiomatic Foundations}
\label{sec:model}

This section fixes the formal substrate used throughout the paper. It serves three purposes:
(i) to formalize a two-domain view of information via semantic and carrier state spaces;
(ii) to separate intrinsic semantic content from operationally stored shortcuts through the canonical core $\mathrm{Atom}(S_O)$; and
(iii) to define the base-relative derivation depth $\mathrm{Dd}(\cdot \mid B)$ that underlies the coding-theoretic and optimization results in later sections.

Concretely, the model consists of:
(i) a fixed logical language $L$ and an effective proof system;
(ii) a two-domain state description $(S_O,S_C)$ with computable realization;
(iii) a notion of irredundant semantic core $\mathrm{Atom}(S_O)$; and
(iv) a computable predecessor operator $P_O$ used solely to define and compute base-relative derivation depth.

%-------------------------------------------------------------
%  II.A Underlying logical system and effective state sets
%-------------------------------------------------------------
\subsection{Underlying Logical System \(\mathcal{L}\) and Effective State Sets}
\label{subsec:logic}

Throughout, we work in a fixed effective logical substrate $\mathcal{L}$, instantiated in this paper as first-order logic with least fixed-point operators $\mathrm{FO}(\mathrm{LFP})$ \cite{immerman1982relational,immerman1999descriptive,vardi1982complexity}, extended with multiple sorts including at least $\mathsf{Obj}$ (entities), $\mathsf{Time}$ (time points), and $\mathsf{Carrier}$ (carriers). We restrict attention to finite structures over a discrete, bounded time domain $T=\{t_0,\dots,t_n\}$. Under this restriction, satisfaction $M \models \varphi$ is decidable for every finite structure $M$ and formula $\varphi \in \mathcal{L}$, and---by the Immerman--Vardi theorem---$\mathrm{FO}(\mathrm{LFP})$ captures exactly the polynomial-time computable properties over ordered finite structures.

\begin{definition}[Expressible and effectively representable state sets]
\label{def:expressible-effective}
A state set $S(X,T)$ is \emph{expressible} in \(\mathcal{L}\) if there exist a finite \(\mathcal{L}\)-signature $\Sigma$ and a finite $L$-structure over $\Sigma$ whose domain contains representations of $X$ and $T$, such that membership $s \in S(X,T)$ and all relations used to describe the state dynamics are definable by formulas of \(\mathcal{L}\) over $\Sigma$. It is \emph{effectively representable} if each $s \in S(X,T)$ admits a finite encoding from which the relevant \(\mathcal{L}\)-predicates can be evaluated effectively~\cite{qiu2025research}.
\end{definition}

Henceforth, all state sets are assumed expressible and effectively representable in \(\mathcal{L}\). Intuitively, these conditions ensure that semantic states can be finitely represented and that proof objects over them can be effectively checked.

\begin{assumption}[Fixed effective proof system and deductive closure]
\label{assump:proof-system}
We fix an effective proof system $(\mathsf{PS})$ for \(\mathcal L\) such that proof checking is decidable.
We write \(\Gamma \vdash_{\mathcal L} \varphi\) to denote derivability of a well-formed formula \(\varphi\)
from a \emph{finite} set of formulas \(\Gamma\) in \(\mathsf{PS}\), and define the deductive closure operator
\[
\mathrm{Cn}(\Gamma)\;:=\;\{\varphi:\ \Gamma\vdash_{\mathcal L}\varphi\}.
\]
\end{assumption}

\begin{remark}[On decidability of \(\mathrm{Cn}(\cdot)\)]
\label{rem:cn-decidability}
Assumption~\ref{assump:proof-system} guarantees only that \emph{verifying} a purported proof is decidable.
In general, the membership predicate \(\varphi \in \mathrm{Cn}(\Gamma)\) may be undecidable (and is typically only
semi-decidable). Whenever we require decidability of specific instances of redundancy/entailment, we state it
explicitly (e.g., Assumption~\ref{ass:decidability}) and work in fragments where it is satisfied.
\end{remark}

%-------------------------------------------------------------
%  II.B Computable Information Axioms
%-------------------------------------------------------------
\subsection{Computable Information Axioms}
\label{subsec:axioms}
In the present paper, these axioms serve only to formalize a two-domain notion of information and to justify the distinction between semantic content and its carrier-level realization; the technical development that follows is driven primarily by the semantic side.

An \emph{information system} consists of a semantic (ontological) state set \(S_O\) and a carrier
(physical) state set \(S_C\), each equipped with its own time domain and representation.
A computable \emph{enabling} (realization) mechanism links semantic states to carrier states~\cite{xu2024research}.

\begin{axiom}[Binary attribute (two-domain information)]
\label{ax:binary-attribute}
Information is modeled by two state domains \((S_O,S_C)\), where \(S_O\) is the semantic/ontological
state set and \(S_C\) is the carrier/physical state set.
\end{axiom}

\begin{axiom}[Existence duration (domain-wise time and precedence)]
\label{ax:existence-duration}
There exist time domains \((T_O,\prec_O)\) and \((T_C,\prec_C)\) and time-index maps
\(\timeindex_O:S_O\to T_O\) and \(\timeindex_C:S_C\to T_C\).
Moreover, there is a fixed cross-domain precedence relation \(\prec\) such that whenever a carrier
state \(s_c\) is an enabled realization of a semantic state \(s_o\), we have
\[
  \timeindex_O(s_o) \prec \timeindex_C(s_c).
\]
\end{axiom}

\begin{axiom}[State representation (effective encodability)]
\label{ax:state-representation}
There exist finite binary representations \(\mathrm{enc}_O:S_O\to\{0,1\}^*\) and \(\mathrm{enc}_C:S_C\to\{0,1\}^*\) such that:
\begin{enumerate}[label=\textup{(R\arabic*)}]
  \item \emph{Decidable identity:} equality of encoded states is decidable.
  \item \emph{Effective predicate evaluation:} all relations/predicates used to describe membership and
  structure on \(S_O\) and \(S_C\) in the fixed logic \(\mathcal L\) are effectively evaluable from the encodings.
\end{enumerate}
\end{axiom}

\begin{axiom}[Enabling mapping (computable totality and surjective coverage)]
\label{ax:enabling-mapping}
There exists a relation \(R_{\mathcal E}\subseteq S_O\times S_C\) and the induced set-valued map
\(\mathcal E:S_O\Rightarrow S_C\), \(\mathcal E(s_o):=\{s_c\in S_C:(s_o,s_c)\in R_{\mathcal E}\}\), such that:
\begin{enumerate}[label=\textup{(R\arabic*)}]
  \item \emph{Totality on semantic states:} for every \(s_o\in S_O\), \(\mathcal E(s_o)\neq\varnothing\).
  \item \emph{Surjective coverage of carrier reality:} for every \(s_c\in S_C\), there exists \(s_o\in S_O\)
  with \(s_c\in \mathcal E(s_o)\). Equivalently, \(\bigcup_{s_o\in S_O}\mathcal E(s_o)=S_C\).
  \item \emph{Computable enabling selector:} there exists a partial computable function
  \(e:\{0,1\}^*\to\{0,1\}^*\) such that for every \(s_o\in S_O\), \(e(\mathrm{enc}_O(s_o))\) halts and equals
  \(\mathrm{enc}_C(s_c)\) for some \(s_c\in\mathcal E(s_o)\).
  \item \emph{Temporal precedence of realization:} if \(s_c\in \mathcal E(s_o)\), then
  \(\timeindex_O(s_o) \prec \timeindex_C(s_c)\).
\end{enumerate}
\end{axiom}

Together, Axioms 2.1–2.4 separate semantic content from physical realization while retaining effective access to both.

%-------------------------------------------------------------
%  II.C Realization mappings and information instances
%-------------------------------------------------------------
\subsection{Realization Mappings and Information Instances}
\label{subsec:enabling}
We now package the preceding axioms into the object used throughout the paper.
By Axiom~\ref{ax:enabling-mapping}, there exists a relation
\(R_{\mathcal E}\subseteq S_O\times S_C\).
We write the induced set-valued enabling (realization) map
\begin{equation}\label{eq:enabling-map}
  \mathcal E:S_O\Rightarrow S_C,
  \;
  \mathcal E(s_o)
  :=
  \{\,s_c\in S_C:\ (s_o,s_c)\in R_{\mathcal E}\,\}.
\end{equation}

\begin{definition}[Computable enabling (realization) mechanism]
\label{def:enabling}
A relation \(R_{\mathcal E}\subseteq S_O\times S_C\) (equivalently, the induced \(\mathcal E\) in
\eqref{eq:enabling-map}) is called a \emph{computable enabling mechanism} if it satisfies all clauses
\textup{(E1)}--\textup{(E4)} of Axiom~\ref{ax:enabling-mapping}.
\end{definition}

\begin{definition}[Information instance]
\label{def:info-instance}
An \emph{information instance} is a tuple
\[
  \mathcal I
  \;=\;
  \langle
    S_O, T_O, S_C, T_C,\timeindex_O,\timeindex_C, R_{\mathcal E}
  \rangle,
\]
where \((S_O,S_C)\) are the two state domains (Axiom~\ref{ax:binary-attribute}),
\(\timeindex_O,\timeindex_C\) are the time-index maps (Axiom~\ref{ax:existence-duration}),
and \(R_{\mathcal E}\) is a computable enabling mechanism (Definition~\ref{def:enabling}).
\end{definition}

%-------------------------------------------------------------
%  II.D Synonymous state sets and ideal information
%-------------------------------------------------------------

\subsection{Synonymous states and ideal information}

Some distinctions between states are operational rather than semantic: different realizations may support the same relevant consequences for the queries under consideration. To factor out such variation, we allow state descriptions to be grouped into synonym classes that preserve the intended semantic content. This provides an idealized notion of information in which semantically equivalent realizations are identified up to the granularity relevant to inference.

This quotient viewpoint is conceptually useful, but most of the technical development below depends only on effective premise bases and derivability. Accordingly, synonymous-state machinery serves mainly as background motivation for separating semantic content from operational representation.
%-------------------------------------------------------------
%  II.E Noisy information
%-------------------------------------------------------------
\subsection{Noisy information}
\label{sec:noisy}
In practice, premise bases may be incomplete, contaminated, or otherwise imperfect. We model such situations by allowing noisy variants of an underlying premise base, with the understanding that some consequences may become harder to derive, unavailable, or derivable only through longer traces.

The detailed noisy model is used later when we analyze perturbation bounds and noisy storage--computation tradeoffs in Section~\ref{sec:noisy-extensions}. For the main development up to that point, it suffices to note that noise is treated as a controlled perturbation of the premise base rather than as a change to the underlying notion of derivation depth itself.

%-------------------------------------------------------------
%  II.F Dependencies, irredundant cores, and derivation depth
%-------------------------------------------------------------
\subsection{Semantic Dependencies, Irredundant Cores, and Derivation Depth}
\label{subsec:atomic-derivation}

This subsection introduces two logically distinct ingredients:
(i) an \emph{irredundant semantic core} \(\mathrm{Atom}(S_O)\) used to separate intrinsic vs.\ operational inference;
(ii) a computable predecessor operator \(P_O\) used to define and compute derivation depth \(\mathrm{Dd}(\cdot\mid B)\).
Crucially, \(\mathrm{Atom}(S_O)\) is defined via the deductive closure \(\mathrm{Cn}(\cdot)\) rather than via \(P_O\).

\begin{assumption}[Finite and effectively listable knowledge bases]
\label{assump:finite-so}
The knowledge bases \(S_O\) considered in this paper are finite and effectively listable under a fixed canonical order.
\end{assumption}

Because multiple irredundant generating sets may exist, we fix a deterministic elimination order to obtain a canonical representative.

\begin{definition}[Semantic atomic core (canonical order-dependent irredundant generating set)]
\label{def:atom-so}
Let \(S_O\) be a finite knowledge base.
Define \(\mathrm{Atom}(S_O)\) to be the output of the following fixed deterministic \emph{irredundantization} procedure:

\begin{enumerate}[label=\textup{(\roman*)}]
  \item Initialize \(A\gets S_O\).
  \item Scan elements of \(S_O\) in the fixed canonical order; for each \(s\in S_O\), if
  $ s\in \mathrm{Cn}(A\setminus\{s\})$, then set \(A\gets A\setminus\{s\}\).
  \item Output \(A\) and denote it by \(\mathrm{Atom}(S_O)\).
\end{enumerate}
\end{definition}

\begin{proposition}[Core correctness: equivalence and irredundancy]
\label{prop:atom-core-correct}
Under Assumptions~\ref{assump:finite-so} and~\ref{ass:decidability}, for every finite knowledge base \(S_O\),
the set \(A:=\mathrm{Atom}(S_O)\) satisfies:
\begin{enumerate}[label=\textup{(\roman*)}]
  \item \emph{Closure equivalence:} \(\mathrm{Cn}(A)=\mathrm{Cn}(S_O)\).
  \item \emph{Irredundancy:} for every \(a\in A\), \(a\notin \mathrm{Cn}(A\setminus\{a\})\).
  \item \emph{Canonicality:} \(\mathrm{Atom}(S_O)\) is uniquely determined by \(S_O\) and the fixed canonical order.
\end{enumerate}
\end{proposition}

\begin{proof}
Each removal step deletes some $s$ such that $s \in \mathrm{Cn}(A \setminus \{s\})$, and therefore preserves deductive closure:
\[
\mathrm{Cn}(A) = \mathrm{Cn}(A \setminus \{s\}).
\]
By induction over the finite scan, the final set $A$ satisfies
\[
\mathrm{Cn}(A) = \mathrm{Cn}(S_O).
\]

To prove irredundancy, let $a \in A$. When $a$ was scanned during the procedure, it was retained, so at that moment $a$ was not derivable from the current set with $a$ removed. After that point, the procedure only removes elements and never adds any. Hence the final set $A \setminus \{a\}$ is a subset of the set from which $a$ was tested for redundancy at its scan time. By monotonicity of $\mathrm{Cn}$, if
\[
a \in \mathrm{Cn}(A \setminus \{a\}),
\]
then $a$ would already have been derivable at the moment it was scanned, contradicting the fact that it was retained. Therefore
\[
a \notin \mathrm{Cn}(A \setminus \{a\}).
\]

Canonicality follows immediately from the determinism of the elimination procedure and the fixed canonical order.
\end{proof}

We emphasize that canonicality here is order-relative: the fixed canonical order removes ambiguity at the representation level, rather than asserting uniqueness of a minimum generating set in an order-free sense.
\begin{assumption}[Effective redundancy test (core extractability) ]
    \label{ass:decidability}
     For any finite $\Gamma\subseteq S_O$ and $s\in S_O$, the predicate $s\in\operatorname{Cn}(\Gamma)$ is decidable. In particular, the redundancy check $s\in\operatorname{Cn}(A\setminus\{s\})$ used in the construction of $\operatorname{Atom}(S_O)$ (Definition~\ref{def:atom-so}) is effective.
    
    \textbf{Why it is necessary.} Without decidability of $\operatorname{Cn}$, the core extraction algorithm would not be well‑defined as an effective procedure, and the separation between intrinsic premises and stored shortcuts (Definition~\ref{ass:decidability}) would lose its algorithmic foundation. Moreover, all subsequent definitions that rely on the core (e.g., intrinsic depth $d_{\mathrm{int}}(q)$) would not be computable in general.
    
    \textbf{Where it holds and its limitations.} This assumption is \emph{not} true for full first‑order logic, where logical entailment is undecidable. However, it is satisfied in many practically relevant decidable fragments, including:
    \begin{itemize}
        \item \textbf{Datalog / Horn clauses} with no function symbols (ground entailment is decidable in polynomial time)\cite{abiteboul1995foundations,dantsin2001complexity};
        \item \textbf{Propositional logic} (co‑NP complete but decidable);
        \item \textbf{Finite‑domain theories} (e.g., finite structures with a bounded domain, where one can exhaustively check models);
        \item \textbf{Description logics} such as $\mathcal{EL}$, $\mathcal{DL}$‑Lite, where subsumption is decidable.
    \end{itemize}
    If the knowledge base is manually curated and known to be free of redundancy, one may bypass this assumption by taking $\operatorname{Atom}(S_O)=S_O$; but in general, the framework requires that any redundancy checks needed for core extraction be decidable. For knowledge bases outside these decidable fragments, the core cannot be effectively identified, and our subsequent analysis does not apply.
\end{assumption}

\begin{definition}[Intrinsic vs.\ operational premise bases]
\label{def:intrinsic-operational-bases}
For a knowledge base \(S_O\), define:
\begin{align*}
 & A := \mathrm{Atom}({S_O}) \quad\text{(core premises)},\\
 & J := S_O\setminus A \quad\text{(stored shortcuts)}.
\end{align*}
We call \(A\) the \emph{intrinsic} premise base and \(A\cup J=S_O\) the \emph{operational} premise base.
\end{definition}

\paragraph{Predecessor structure for depth}
We model single-step inferential dependency by a predecessor operator \(P_O\).

\begin{axiom}[Finite and computable predecessors]
\label{ax:finite-predecessors}
For every \(s\in S_O\), the set \(P_O(s)\) is finite and effectively computable from \(\mathrm{enc}_O(s)\).
\end{axiom}

\begin{axiom}[Well-foundedness of predecessor unfolding]
\label{ax:wellfounded-predecessor}
For every \(s\in S_O\), the backward unfolding of \(s\) along \(P_O\) contains no infinite chain.
\end{axiom}

\begin{proposition}[Finite, acyclic, and computable predecessor unfolding]
\label{prop:no-cycle-dag}
Under Axioms~\ref{ax:finite-predecessors} and~\ref{ax:wellfounded-predecessor},
for any \(s\in S_O\), the backward unfolding of \(s\) along \(P_O\) is a finite directed acyclic graph (DAG),
and is effectively traversable from \(\mathrm{enc}_O(s)\).
\end{proposition}

\begin{proof}
Define the rank of \(s\) recursively by
\[
r(s):=\sup\{r(s')+1:s'\in P_O(s)\}.
\]
Axiom~\ref{ax:wellfounded-predecessor} implies that \(r(s)\) is a finite ordinal, hence a natural number, and Axiom~\ref{ax:finite-predecessors} gives finite branching.
Therefore the set of nodes reachable within rank at most \(r(s)\) is finite, so the backward unfolding is finite and acyclic.
\end{proof}

\begin{assumption}[Alignment of dependency and inference (relative to a premise base)]
\label{assump:alignment}
Fix a finite available premise base \(B\subseteq S_O\).
For every \(s\in S_O\setminus B\), a canonical normal-form one-step inference concluding \(s\) is fixed in the chosen proof system, and \(P_O(s)\) is defined to be exactly the set of immediate premises of that canonical step.
Moreover, for every \(s\in \mathrm{Cn}(B)\), repeated unfolding along \(P_O\) reaches only elements of \(B\) at the leaves.
\end{assumption}

This assumption is natural in proof systems equipped with an explicit normal form of one-step inferences, but may fail when dependency height and sequential proof length diverge substantially.

\begin{definition}[Base-relative derivation depth]
\label{def:derivation-depth}
Let \(B\subseteq S_O\) be a finite set of \emph{available premises}.
Under Assumption~\ref{assump:alignment}, the \emph{derivation depth} of \(s\in S_O\) from \(B\) is
\begin{equation}\label{eq:Dd-inductive}
  \mathrm{Dd}(s \mid B)
  :=
  \begin{cases}
    0, & \text{if } s\in B,\\[0.4ex]
    1+\displaystyle\max_{s'\in P_O(s)}\mathrm{Dd}(s' \mid B), & \text{otherwise},
  \end{cases}
\end{equation}
with the convention that the maximum over an empty set is \(0\). For \(s\notin \mathrm{Cn}(B)\), set \(\mathrm{Dd}(s\mid B):=\infty\).
\end{definition}

\begin{remark}[Nullary inference steps]
\label{rem:nullary}
The convention \(\max \varnothing = 0\) allows \(\mathrm{Dd}(s\mid B)=1\) when \(s\notin B\) and \(P_O(s)=\varnothing\),
which corresponds to admitting nullary (axiom) inference steps in the fixed normal form.
If a given proof system has no such steps, one may equivalently restrict attention to instances where
\(P_O(s)\neq\varnothing\) for all \(s\notin B\).
\end{remark}

A finite set \(B\subseteq S_O\) is called an \emph{available premise base} if its elements are treated as depth-\(0\) premises when measuring derivation depth from \(B\).

\begin{theorem}[Well-definedness and computability of derivation depth]
\label{thm:computability-Dd}
Let \(B\subseteq S_O\) be a finite, effectively decidable premise base.
Under Assumptions~\ref{assump:alignment} and Axioms~\ref{ax:finite-predecessors}--\ref{ax:wellfounded-predecessor},
for every \(s\in \mathrm{Cn}(B)\), \(\mathrm{Dd}(s\mid B)\) is a unique, finite, and computable non-negative integer.
If \(\mathrm{Dd}\) is extended by \(\mathrm{Dd}(s\mid B)=\infty\) for \(s\notin \mathrm{Cn}(B)\), then this extension is computable as an extended-value function whenever membership in \(\mathrm{Cn}(B)\) is decidable on the instance class considered.
\end{theorem}

\begin{proof}
By Proposition~\ref{prop:no-cycle-dag}, the predecessor unfolding is a finite computable DAG.
Thus the recursion~\eqref{eq:Dd-inductive} is well-founded. A bottom-up dynamic program computes \(\mathrm{Dd}(s\mid B)\)
using computability of \(P_O(\cdot)\) and decidability of membership in \(B\).
\end{proof}
\begin{remark}[Connection to logical depth]
\label{rem:bennett-preview}
The depth functional \(\mathrm{Dd}(\cdot)\) is related to Bennett's logical depth \cite{bennett1988logical}.
Later sections connect intrinsic/operational derivation depth to conditional logical depth under explicit simulation assumptions.
\end{remark}
\begin{definition}[Intrinsic and operational derivation depths]
\label{def:int-op-depth}
Let \(S_O\) be a knowledge base with core \(A:=\mathrm{Atom}(S_O)\).
For any query \(q\) derivable from \(A\) (hence from \(S_O\)), define
\[
  d_{\mathrm{int}}(q) := \mathrm{Dd}(q \mid A),
  \qquad
  d_{\mathrm{op}}(q) := \mathrm{Dd}(q \mid S_O).
\]
\end{definition}

This distinction will allow us to interpret stored shortcuts as reducing operational cost without altering intrinsic semantic content.

\begin{remark}[Monotonicity and interpretation]
\label{rem:int-op-monotone}
Since \(A\subseteq S_O\), enlarging the premise base cannot increase derivation depth; hence
\[
  d_{\mathrm{op}}(q) \le d_{\mathrm{int}}(q).
\]
The intrinsic depth \(d_{\mathrm{int}}\) captures inferential content relative to an irredundant core,
whereas the operational depth \(d_{\mathrm{op}}\) captures actual online cost under stored shortcuts.
This separation is the basis of the storage--computation optimization in later sections.
\end{remark}

\begin{definition}[Semantic atomicity measure]
\label{def:atomicity-measure}
The \emph{semantic atomicity} of an information instance \(\mathcal I\) with semantic space \(S_O\) is
\[
  \mathsf{A}(\mathcal I)\;:=\;\lvert \mathrm{Atom}(S_O)\rvert.
\]
\end{definition}

%%%%%%%%%%%%%%%%%%%%%%%%%%%%%%%%%%%%%%%%%%%%%%%%%%%%%%%%%%%%%%%%%%%%%%%%%%%%%%%%%%%%%%%%%%%%%%%%%%%%%%%%%
%%%   III. DERIVATION DEPTH AS AN INFORMATION METRIC
%%%%%%%%%%%%%%%%%%%%%%%%%%%%%%%%%%%%%%%%%%%%%%%%%%%%%%%%%%%%%%%%%%%%%%%%%%%%%%%%%%%%%%%%%%%%%%%%%%%%%%%%%

\section{Derivation Depth as an Information Metric}
\label{sec:derivation-entropy-ideal}

This section develops the ideal (noise-free) theory.
We show that, for generic information-rich queries,
base-relative derivation depth provides a quantitative
information-theoretic proxy for conditional algorithmic
complexity. We also introduce a BCQ-oriented \emph{compositional} (sum-rule)
depth functional \(\mathrm{Dd}_\Sigma(\cdot\mid B)\) for conjunctive/combined queries, while keeping
\(\mathrm{Dd}(\cdot\mid B)\) as the paper's primary depth notion.

A key point is that all depths are measured \emph{relative to a finite available premise base}
\(B\subseteq S_O\): elements of \(B\) are treated as depth-\(0\) premises
(Definition~\ref{def:derivation-depth}).

For readability, throughout this section we write \(\mathcal{P}\) for the fixed effective proof
system \(\mathsf{PS}\) in Assumption~\ref{assump:proof-system}.

\begin{remark}[Logarithm convention]
\label{rem:log-convention}
Throughout Sections~\ref{sec:derivation-entropy-ideal}--\ref{sec:tradeoff-ideal}, \(\log\) denotes \(\log_2\)
(bit-length scale), while \(\ln\) denotes the natural logarithm.
\end{remark}

%%%%%%%%%%%%%%%%%%%%%%%%%%%%%%%%%%%%%%%%%%%%%%%%%%%%%%%%%%%%%%%%%%%%%%%%%%%%%%%%%%%%%%%%%%%%%%%%%%%%%%%%%
%%%   III.A Derivation Entropy and the Information-Rich Regime
%%%%%%%%%%%%%%%%%%%%%%%%%%%%%%%%%%%%%%%%%%%%%%%%%%%%%%%%%%%%%%%%%%%%%%%%%%%%%%%%%%%%%%%%%%%%%%%%%%%%%%%%%

\subsection{Derivation Entropy and the Information-Rich Regime}
\label{subsec:derivation-entropy-def}

\begin{assumption}[Finite premise bases for asymptotic analysis]
\label{assump:finite-premise-base}
In Sections~\ref{sec:derivation-entropy-ideal}--\ref{sec:tradeoff-ideal}, we consider finite available
premise bases \(B\subseteq S_O\) that are effectively listable.
We write \(m:=|B|\) and fix a canonical computable encoding \(\langle B\rangle\in\{0,1\}^*\)
(e.g., a length-prefixed list of encodings in a fixed order).
\end{assumption}

Fix a premise base \(B\) satisfying Assumption~\ref{assump:finite-premise-base}, and consider queries
\(q \in \mathrm{Cn}(B)\) derivable in \(\mathcal{P}\) over \(\mathcal{L}\).

\begin{definition}[Derivation entropy]
\label{def:derivation-entropy}
Let \(B\) be a premise base and \(q \in \mathrm{Cn}(B)\) a derivable query.
The \emph{derivation entropy} of \(q\) given \(B\) is
\begin{equation}\label{eq:derive-entropy-def}
  H_{\mathrm{derive}}(q \mid B) := \mathrm{Dd}(q \mid B)\cdot \ln 2,
\end{equation}
where \(\mathrm{Dd}(q \mid B)\) is the base-relative derivation depth from Definition~\ref{def:derivation-depth}.
\end{definition}

\begin{remark}[Terminology]
\label{rem:derive-entropy-terminology}
The quantity \(H_{\mathrm{derive}}(q \mid B)\) is an \emph{entropy-like} rescaling of an online derivation-cost functional;
it is not a Shannon entropy unless a query distribution is specified (cf.\ Remark~\ref{rem:shannon-relationship}).
\end{remark}

\begin{definition}[Conditional algorithmic entropy]
\label{def:cond-algo-entropy}
Fix a universal Turing machine \(U\) and the encoding convention \(\langle\cdot\rangle\).
For a finite premise base \(B\) and a query \(q\), define
\[
  H_K(q \mid B) := K\bigl(q \mid \langle B \rangle\bigr)\cdot \ln 2.
\]
\end{definition}

\begin{remark}[Prefix-free complexity convention]
\label{rem:prefix-free-K}
Throughout, \(K(\cdot\mid\cdot)\) denotes \emph{prefix-free} (self-delimiting) conditional Kolmogorov complexity
with respect to the fixed universal machine \(U\) and encoding convention \(\langle\cdot\rangle\).
This convention aligns with the self-delimiting encodings used later for derivation traces.
\end{remark}

\begin{remark}[Relation to Shannon entropy (in expectation)]
\label{rem:shannon-relationship}
Let \(Q\) be a random query drawn from a computable distribution \(P\) over a finite query set.
Standard results (for a prefix-free variant of \(K\)) imply
\[
\mathbb{E}_P \left[K\left(Q \mid \langle B\rangle\right)\right]
\le H_P \left(Q \mid \langle B\rangle\right) + O(1),
\]
where \(H_P(\cdot\mid\cdot)\) is Shannon conditional entropy.
Thus, while our development is individual-sequence (distribution-free), it is compatible with the
Shannon-theoretic viewpoint in expectation under computable query sources.
\end{remark}

The rest of the section, however, remains entirely distribution-free and concerns individual queries.

\begin{assumption}[Computable query source]
\label{assump:computable-source}
Let \(Q\) be a random query supported on a finite set of encodings.
We assume the conditional distribution $P({Q\mid \langle B\rangle})$ is computable uniformly in \(\langle B\rangle\),
i.e., there exists an algorithm that, given \(\langle B\rangle\), outputs the probability mass function of \(Q\)
(or samples from it) to arbitrary precision.
\end{assumption}

\begin{proposition}[Expected Kolmogorov complexity matches Shannon entropy up to \(O(1)\)]
\label{prop:K-vs-Shannon}
Under Assumption~\ref{assump:computable-source}, let \(K(\cdot\mid\cdot)\) be prefix-free conditional Kolmogorov
complexity (Remark~\ref{rem:prefix-free-K}). Then
\[
\mathbb{E}[K(Q\mid \langle B\rangle)] = H_P(Q\mid \langle B\rangle) + O(1),
\]
where the additive constant is uniform over \(B\) within the family covered by the fixed algorithm witnessing uniform computability of $P({Q\mid \langle B\rangle})$, but is independent of the particular
realization of \(\langle B\rangle\) and of the sample \(Q\).
Standard references include, e.g.,~\cite{li2008introduction}.
\end{proposition}

\begin{proof}
Since $P({Q\mid \langle B\rangle})$ is computable uniformly from \(\langle B \rangle\) and the support is finite, there exists a prefix code \(C_{\langle B \rangle}\) (e.g., Shannon--Fano or arithmetic coding) computable from \(\langle B \rangle\) such that the associated code length \(l(q)\) satisfies
\[
l(q) \le -\log P(q \mid \langle B \rangle) + O(1).
\]
Hence
\[
K(q \mid \langle B \rangle) \le l(q) + O(1)
\le -\log P(q \mid \langle B \rangle) + O(1).
\]
Taking expectations gives
\[
\mathbb{E}[K(Q \mid \langle B \rangle)] \le H_p((Q \mid B)) + O(1). 
\]
For the reverse direction, prefix-free complexities induce (up to an additive constant) valid prefix code lengths,
so Kraft's inequality implies that any such lengths have expected value at least the Shannon entropy, yielding
\(\mathbb{E}[K(Q\mid \langle B\rangle)] \ge H_P(Q\mid \langle B\rangle)-O(1)\).
\end{proof}

\begin{remark}[Compatibility with the individual-sequence viewpoint]
\label{rem:entropy-compat}
Proposition~\ref{prop:K-vs-Shannon} is used only to translate our individual-query statements into
distributional (expected) ones. All main coding theorems below remain distribution-free.
\end{remark}

\begin{definition}[Complexity factor]
\label{def:complexity-factor}
For a premise base \(B\) with \(m=|B|\) and a query \(q\), for any nonnegative integer parameter \(n\), to be instantiated either as \(\mathrm{Dd}(q\mid B)\) or \(\mathrm{N}(q\mid B)\)(Definition~\ref{def:min-trace-length}), define 
\[
  L_B := \log(m+n).
\]
\end{definition}

\begin{definition}[Information-rich regime]
\label{def:info-rich}
Let \(B\) be a premise base with \(m:=|B|\).
We say that a family of instances \((q,B)\) lies in the \emph{information-rich regime} if
\[
  K\bigl(q \mid \langle B \rangle\bigr) = \omega(\log m),
\]
with asymptotics taken along the family with \(m\to\infty\).
\end{definition}

\begin{remark}[Meaning of the condition]
\label{rem:indexing-overhead}
The term \(\log m\) is the indexing overhead: the number of bits required merely to name
one premise among \(m\) candidates. The information-rich regime excludes queries whose conditional
description is essentially a constant number of premise indices (pure lookups), and focuses on queries
whose specification carries substantially more information than naming.
\end{remark}

%%%%%%%%%%%%%%%%%%%%%%%%%%%%%%%%%%%%%%%%%%%%%%%%%%%%%%%%%%%%%%%%%%%%%%%%%%%%%%%%%%%%%%%%%%%%%%%%%%%%%%%%%
%%%   III.B Traces, BCQs, and Encoding Length
%%%%%%%%%%%%%%%%%%%%%%%%%%%%%%%%%%%%%%%%%%%%%%%%%%%%%%%%%%%%%%%%%%%%%%%%%%%%%%%%%%%%%%%%%%%%%%%%%%%%%%%%%

\subsection{Derivation Traces, BCQs, and Encoding Length}
\label{subsec:encoding}

The main bridge to Kolmogorov complexity is explicitly constructed at the level of \emph{derivation traces}.
For conjunctive/combined queries (BCQs), we additionally introduce a compositional (sum-rule) depth
functional \(\mathrm{Dd}_\Sigma(\cdot\mid B)\) that directly models aggregation cost.

\subsubsection{Derivation traces and shortest-trace length}
\label{subsubsec:traces}

\begin{definition}[Derivation trace and length (base-parametric)]
\label{def:derivation-trace}
Fix a finite premise base \(B\) with \(m:=|B|\) and a proof system \(\mathcal P\) with a finite rule set~\cite{cook1979relative}.
A (sequential) \emph{derivation trace} of a query \(q\) from \(B\) is a pair
\[
  \pi = \bigl((r_1,\vec p_1),\ldots,(r_\ell,\vec p_\ell);\ o\bigr),
\]
where each \(r_i\) is a rule identifier, each \(\vec p_i\) is a vector of premise pointers (indices),
and \(o\) is an \emph{output pointer} (an index into the final pool), such that the following deterministic
replay procedure outputs \(q\):
\begin{enumerate}[label=\textup{(\roman*)}]
  \item Initialize the pool \(P_0\) as an ordered list consisting of the \(m\) premises \(B\) in the canonical
  order induced by \(\langle B\rangle\).
  \item For \(i=1,\ldots,\ell\), interpret \(\vec p_i\) as indices into the current pool \(P_{i-1}\), apply rule
  \(r_i\) of \(\mathcal P\) to the referenced premises (in the prescribed order), obtain a conclusion formula
  \(f_i\), and append it to obtain \(P_i := P_{i-1}\,\|\,\langle f_i\rangle\).
  \item Output the formula stored at index \(o\) in \(P_\ell\).
\end{enumerate}
The trace is \emph{valid} if each rule application is checkable in \(\mathcal P\) (hence decidable), and it
\emph{derives \(q\)} if the replay output equals \(q\). Its \emph{length} is \(|\pi|:=\ell\).
\end{definition}

\begin{definition}[Base premises referenced by a derivation trace]
\label{def:atoms-of-trace}
Let \(\pi=((r_1,\vec p_1),\ldots,(r_\ell,\vec p_\ell);\ o)\) be a derivation trace of \(q\) from \(B\)
(Definition~\ref{def:derivation-trace}).
Define \(\mathrm{Atoms}(\pi)\subseteq B\) as the set of base premises referenced by \(\pi\), namely those premises
whose indices in the initial pool \(P_0\) are referenced by some premise pointer in \(\pi\), or by the
output pointer when \(\ell=0\).
\end{definition}

\begin{definition}[Minimal derivation-trace length (base-parametric)]
\label{def:min-trace-length}
For a derivable query \(q\in\mathrm{Cn}(B)\), define the minimal derivation-trace length
\begin{align*}
  \mathrm{N}(q\mid B)
  :=
  \min\bigl\{&|\pi|:\ \pi \text{ is a valid derivation trace deriving } q \text{ from } B\bigr\}.
\end{align*}
(We allow \(|\pi|=0\), in which case the trace outputs \(q\) by selecting it from the initial pool \(P_0\).)
\end{definition}

\subsubsection{BCQs and compositional (sum-rule) depth}
\label{subsubsec:bcq-additive}

\begin{definition}[Boolean conjunctive query (BCQ)]
\label{def:bcq}
A \emph{Boolean conjunctive query (BCQ)} is a formula of the form
\[
  q \;=\; \bigwedge_{j=1}^{t} a_j,
\]
where each \(a_j\) is an atomic formula of \(\mathcal L\) (or a designated atomic query form) and \(t\ge 1\).
We write \(\mathrm{len}(q):=t\).
\end{definition}

\begin{definition}[Compositional (sum-rule) derivation depth for BCQs]
\label{def:additive-depth}
Fix a premise base \(B\) and the proof system \(\mathcal P\).
For a BCQ \(q\) derivable from \(B\), define its \emph{compositional} (sum-rule) depth by
\[
  \mathrm{Dd}_{\Sigma}(q\mid B) \;:=\; \mathrm{N}(q\mid B).
\]
\end{definition}

\begin{remark}[Role of \(\mathrm{Dd}_{\Sigma}\) and why Section IV still uses \(\mathrm{Dd}\)]
\label{rem:Dd-vs-DdSigma}
The base-relative depth \(\mathrm{Dd}(\cdot\mid B)\) from Definition~\ref{def:derivation-depth} is the paper's
primary inferential depth notion (height of a dependency unfolding).
For BCQs, the online computation cost is often governed by \emph{aggregation} across conjuncts; this motivates
the compositional depth \(\mathrm{Dd}_{\Sigma}(\cdot\mid B)\), which is step-accurate by construction.
In Section~\ref{sec:tradeoff-ideal}, we keep \(\mathrm{Dd}\) as the default cost symbol, and provide a BCQ
instantiation by replacing \(\mathrm{Dd}\) with \(\mathrm{Dd}_{\Sigma}\) when the query workload is BCQ-dominated.
Thus, \(\mathrm{Dd}_{\Sigma}\) should be interpreted as a workload-specific refinement of, rather than a competing replacement for, the primary depth notion \(\mathrm{Dd}\).
\end{remark}

\begin{lemma}[BCQ aggregation overhead is linear in the number of conjuncts]
\label{lem:bcq-aggregation}
Assume \(\mathcal P\) contains a conjunction-introduction rule allowing one to derive
\(\varphi\land \psi\) from \(\varphi\) and \(\psi\) in \(O(1)\) trace steps.
Let \(q=\bigwedge_{j=1}^{t} a_j\) be a BCQ derivable from \(B\). Then
\[
  \mathrm{Dd}_{\Sigma}(q\mid B)
  \;\le\;
  \sum_{j=1}^{t} \mathrm{Dd}_{\Sigma}(a_j\mid B) \;+\; O(t),
\]
where each \(a_j\) is treated as a one-conjunct BCQ.
\end{lemma}

\begin{proof}
Take shortest traces \(\pi_j\) deriving each \(a_j\) from \(B\), concatenate them (reusing the same base pool),
and then conjoin the obtained atoms using \(t-1\) conjunction-introduction steps under a fixed parenthesization.
This yields a valid trace for \(q\) of length \(\sum_j |\pi_j| + O(t)\), hence the bound.
\end{proof}

\subsubsection{Derivation encoding length}
\label{subsubsec:encoding-bound}

\begin{lemma}
\label{lem:encoding}
Let \(P\) be a proof system with finite rule set \(R\) and maximum rule arity \(k\), where \(|R|\) and \(k\) are constants independent of the premise-base size. Let \(B\) be a premise base with \(m=|B|\). Then every derivation trace \(\pi\) of length \(n:=|\pi|\) from \(B\) admits a self-delimiting encoding such that
\[
\ell'(\pi)
\le
n\bigl(k\lceil \log(m+n)\rceil + \lceil \log |R| \rceil + c_0\bigr)
+ \lceil \log(m+n+1)\rceil
+ c_1\bigl(\log m+\log n+1\bigr),
\]
for absolute constants \(c_0,c_1>0\) depending only on the fixed integer-coding convention. Here, $\ell'(\pi)$ is the encoding length of $\pi$ , i.e., the number of bits required to encode the derivation trace in a self-delimiting manner. In particular,
\[
\ell'(\pi)
\le
k n \log(m+n+1) + O(n) + O(\log(m+n+1)).
\]
Here the linear \(O(n)\) term depends only on \(|R|\), \(k\), and the fixed self-delimiting encoding convention.
    
\end{lemma}

\begin{proof}

Write
\[
\pi = ((r_1,\vec p_1),\dots,(r_n,\vec p_n);o).
\]
We encode \(\pi\) by concatenating:

(1) a self-delimiting header encoding \(m\) and \(n\);

(2) for each step \(i=1,\dots,n\):
  (a) the rule identifier \(r_i\in R\);
  (b) the premise-pointer vector \(\vec p_i\);

(3) a final output pointer \(o\).

Header.
Using any standard self-delimiting integer code, \(m\) and \(n\) can be encoded in
\[
c_1(\log m+\log n+1)
\]
bits for a constant \(c_1>0\).

Step \(i\).
At step \(i\), the current pool \(P_{i-1}\) contains the \(m\) initial premises together with the \(i-1\) previously derived formulas, hence
\[
|P_{i-1}| = m+i-1 \le m+n-1 < m+n.
\]
Therefore each premise pointer into \(P_{i-1}\) can be encoded using
\[
\lceil \log(m+n)\rceil
\]
bits. Since every rule has arity at most \(k\), the whole pointer vector \(\vec p_i\) costs at most
\[
k\lceil \log(m+n)\rceil + c_0
\]
bits, where the constant \(c_0\) covers any fixed overhead needed to identify the actual arity or delimit pointer fields. The rule identifier \(r_i\in R\) costs \(\lceil \log |R| \rceil\) bits.

Thus each step costs at most
\[
k\lceil \log(m+n)\rceil + \lceil \log |R| \rceil + c_0.
\]

Output pointer.
After \(n\) steps, the final pool has size \(m+n\), so the output pointer \(o\) costs at most
\[
\lceil \log(m+n)\rceil
\]
bits.

Summing these contributions yields
\[
\ell'(\pi)
\le
n\bigl(k\lceil \log(m+n)\rceil + \lceil \log |R| \rceil + c_0\bigr)
+\lceil \log(m+n)\rceil
+c_1(\log m+\log n+1).
\]
Since \(|R|\) and \(k\) are constants independent of \(m,n\), this implies
\[
\ell'(\pi)\le k n \log(m+n+1)+O(n)+O(\log(m+n+1)).
\]
    
\end{proof}
%%%%%%%%%%%%%%%%%%%%%%%%%%%%%%%%%%%%%%%%%%%%%%%%%%%%%%%%%%%%%%%%%%%%%%%%%%%%%%%%%%%%%%%%%%%%%%%%%%%%%%%%%
%%%   III.C Richness and Incompressibility (organized)
%%%%%%%%%%%%%%%%%%%%%%%%%%%%%%%%%%%%%%%%%%%%%%%%%%%%%%%%%%%%%%%%%%%%%%%%%%%%%%%%%%%%%%%%%%%%%%%%%%%%%%%%%

\subsection{Richness and Incompressibility}
\label{subsec:richness-incomp}

Lower bounds rely on a combinatorial condition ensuring that there are many distinct queries at a given
online cost scale. For our coding arguments, the natural cost scale is shortest-trace length (steps);
BCQs are covered by \(\mathrm{Dd}_\Sigma=\mathrm{N}\).

\begin{definition}[Richness condition relative to a base family]
\label{def:richness}
Let \(\mathcal B\) be a class of premise bases. We say that the pair \((P,\mathcal B)\) satisfies the richness condition with parameter \(\delta_0>0\) if there exists
\(m_0>0\) such that for all \(m\ge m_0\), all integers \(n\) with \(\log m \le n \le \poly(m)\), and all premise bases
\(B\) with \(|B| = m\),
\begin{equation}\label{eq:richness-cond}
  \bigl|\{\,\langle q\rangle : \mathrm{N}(q \mid B) = n\,\}\bigr|
  \ge
  2^{(1-\delta_0)\, n \log(m+n)}.
\end{equation}
The cardinality is taken over distinct encodings \(\langle q\rangle\) under the fixed encoding conventions.
\end{definition}

\begin{remark}[Scope of richness]
\label{rem:richness-scope}
Definition~\ref{def:richness} is stated uniformly over all premise bases \(B\) of size \(m\) for clarity.
For the main incompressibility arguments, it suffices that the lower bound \eqref{eq:richness-cond} holds
on the \emph{class of bases} under consideration (e.g., the atomic-fact bases in
Proposition~\ref{prop:bcq-richness-horn}), rather than for every possible \(B\).
\end{remark}

\begin{remark}[Compatibility with \(\mathrm{Dd}\) and with BCQs]
\label{rem:richness-compat}
Definition~\ref{def:richness} is stated in terms of step complexity \(\mathrm{N}(\cdot\mid B)\), which is the quantity
directly controlled by Lemma~\ref{lem:encoding}.
For general queries, Assumption~\ref{assump:serializable} (introduced below) allows transferring results to
\(\mathrm{Dd}(\cdot\mid B)\) up to constant factors.
For BCQs, \(\mathrm{Dd}_\Sigma(\cdot\mid B)=\mathrm{N}(\cdot\mid B)\) by definition, so richness immediately supports the BCQ
(sum-rule) instantiation.
\end{remark}

\subsubsection{A BCQ richness witness in a Horn/Datalog fragment}
\label{subsubsec:richness-bcq}

\begin{proposition}[Instance-wise richness for BCQs under compositional depth]
\label{prop:bcq-richness-horn}
Assume \(\mathcal P\) contains binary conjunction introduction (as in Lemma~\ref{lem:bcq-aggregation}), and suppose that in the chosen trace normal form each single trace step introduces at most one new occurrence of the connective \(\wedge\) in the derived formula. Let the premise base \(B\) consist of \(m\) pairwise distinct atomic formulas
\[
\{a_1,\dots,a_m\}\in \mathcal L .
\]

Fix \(n \ge 1\). For each tuple \(i=(i_0,\dots,i_n)\) of distinct elements of \([m]\), let
\[
q_i := a_{i_0}\wedge a_{i_1}\wedge \cdots \wedge a_{i_n}.
\]
Set
\[
Q_n^{\mathrm{BCQ}} := \{ q_i : i \text{ is as above} \}.
\]

Assume a fixed left-associative parenthesization throughout, so that distinct index sequences yield distinct encodings. Then:

\begin{enumerate}[label=\textup{(\roman*)}]
\item Every \(q \in Q_n^{\mathrm{BCQ}}\) satisfies
\begin{equation}
\mathrm{Dd}_\Sigma(q \mid B) = \mathrm{N}(q \mid B) = n .
\label{eq:bcq-depth}
\end{equation}

\item The family cardinality is
\begin{equation}
|Q_n^{\mathrm{BCQ}}| = m(m-1)\cdots(m-n),
\label{eq:bcq-cardinality}
\end{equation}
and therefore
\begin{equation}
\log |Q_n^{\mathrm{BCQ}}|
= (n+1)\log m - O(n).
\label{eq:bcq-log-cardinality}
\end{equation}
\end{enumerate}

Consequently, if there exists a fixed constant \(C \ge 1\) such that \(n \le m^C\) in the asymptotic regime, then the richness inequality~\eqref{eq:richness-cond} holds for these bases with any choice of
\begin{equation}
\delta_0 \ge 1 - \frac{1}{C+1},
\label{eq:bcq-delta0}
\end{equation}
for all sufficiently large \(m\).
\end{proposition}

\begin{proof}
(i) For any \(q_{\mathbf{i}}\in  Q_n^{\mathrm{BCQ}}\), each \(a_{i_j}\in B\) is available at depth \(0\).
A left-associative construction uses exactly \(n\) binary conjunction-introduction steps, hence
\(\mathrm{N}(q_{\mathbf{i}}\mid B)\le n\).
On the other hand, \(q_{\mathbf{i}}\) contains exactly \(n\) occurrences of \(\land\) under the fixed
parenthesization. By the normal-form assumption that each trace step can introduce at most one new \(\land\),
any derivation trace producing \(q_{\mathbf{i}}\) must have length at least \(n\).
Therefore \(\mathrm{N}(q_{\mathbf{i}}\mid B)=n\), and since \(\mathrm{Dd}_\Sigma=\mathrm{N}\) for BCQs, \(\mathrm{Dd}_\Sigma(q_{\mathbf{i}}\mid B)=n\).

(ii) The counting formula for \(| Q_n^{\mathrm{BCQ}}|\) is immediate.
For richness, when \(n\le m^C\), we have \(\log(m+n)\le (C+1)\log m\) for all sufficiently large \(m\), so
\begin{equation}
\begin{aligned}
\log \bigl|\{\, q : \mathrm{N}(q \mid B) = n \,\}\bigr|
  &\ge \log |Q_n^{\mathrm{BCQ}}| \\
  &= (n+1)\log m - O(n) \\
  &\ge (1-\delta_0)n \log(m+n).
\end{aligned}
\label{eq:bcq-richness-chain}
\end{equation}
for any \(\delta_0\ge 1-\frac{1}{C+1}\) and all sufficiently large \(m\), after absorbing the lower-order \(O(n)\) term.
\end{proof}

\subsubsection{Incompressibility at the \(n\log(m+n)\) scale}
\label{subsubsec:incompressibility}

\begin{lemma}[Derivation incompressibility (step-length form)]
\label{lem:incompressibility}
Let \(B\) be a premise base with \(m = |B|\), let \(q\in\mathrm{Cn}(B)\) be derivable, and set
\[
  n := \mathrm{N}(q\mid B)\ge 1,
  \;
  L_B := \log(m+n),
  \;
  L := n\cdot L_B.
\]
Assume \(\log m \le n \le \poly(m)\).

\begin{enumerate}[label=\textup{(\roman*)}]
  \item \emph{Upper bound (all queries):} There exists a constant \(c>0\) (depending only on \(\mathcal{P}\), \(U\),
    and the fixed encoding conventions) such that
    \begin{equation}\label{eq:K-upper}
       K(q\mid \langle B\rangle)=O(L).
    \end{equation}

  \item \emph{Counting bound:} For any \(\delta>0\), the set of \(\delta\)-compressible length-\(n\) queries
  \begin{align}\label{eq:compressible-def}
    \mathcal{C}_\delta := &\bigl\{q : \mathrm{N}(q \mid B) = n,\nonumber
     K(q \mid \langle B \rangle) < (1-\delta) L\bigr\},
  \end{align}
  satisfies
  \begin{equation}\label{eq:compressible-count}
    |\mathcal{C}_\delta| < 2^{(1-\delta)L}.
  \end{equation}

  \item \emph{Generic lower bound (fraction form, under richness):}
  Suppose \(\mathcal{P}\) satisfies the richness condition (Definition~\ref{def:richness}) with parameter \(\delta_0>0\),
  and let \(\delta>\delta_0\). Let \(\mathcal{Q}_n := \{q : \mathrm{N}(q \mid B) = n\}\). Then the fraction of
  \(\delta\)-compressible queries at length \(n\) satisfies
  \begin{equation}\label{eq:fraction-compressible}
    \frac{|\mathcal{C}_\delta|}{|\mathcal{Q}_n|}
    \le 2^{-(\delta-\delta_0)L}.
  \end{equation}
  Equivalently, at least a \(1-2^{-(\delta-\delta_0)L}\) fraction of queries in \(\mathcal{Q}_n\) satisfy
  \begin{equation}\label{eq:K-lower}
    K(q \mid \langle B \rangle) \ge (1-\delta) L.
  \end{equation}
\end{enumerate}
\end{lemma}

\begin{proof}
{(i)}
Let \(\pi\) be a shortest derivation trace of \(q\) from \(B\), so
\(|\pi|=n\). Consider a program that contains \(\operatorname{enc}(\pi)\),
reads \(\langle B\rangle\) as conditional input, replays \(\pi\) in
\(\mathcal P\), and outputs \(\langle q\rangle\). The simulation overhead is
constant. By Lemma~\ref{lem:encoding},
\[
\ell'(\operatorname{enc}(\pi))
\le kn\log(m+n)+O(n)=O(L),
\]
yielding~\eqref{eq:K-upper}.

\smallskip

{(ii)}
For any fixed conditional input \(y\), the number of strings \(x\)
such that
\[
K(x \mid y) < \ell
\]
is less than \(2^\ell\); see, e.g.,~\cite{li2008introduction}. Applying this with
\(y=\langle B\rangle\) and
\[
\ell = (1-\delta)L
\]
yields~\eqref{eq:compressible-count}.

\smallskip

{(iii)}
By richness,
\[
|\mathcal{Q}_n| \ge 2^{(1-\delta_0)L}.
\]
By part~(ii),
\[
|C_\delta| < 2^{(1-\delta)L}.
\]
Thus
\(|\mathcal{C}_\delta|/|\mathcal{Q}_n|\le 2^{-(\delta-\delta_0)L}\)
proving~\eqref{eq:fraction-compressible}. The complement yields~\eqref{eq:K-lower}.
\end{proof}
\begin{remark}[When richness fails]
Derivation depth alone does not force large conditional complexity. For example, in a linear-chain theory where each depth \(n\) yields essentially a unique query, one may have
\[
\mathrm{Dd}(q_n \mid B) = n
\quad\text{but}\quad
K(q_n \mid \langle B \rangle) = O(\log n),
\]
since specifying the query reduces to naming the position \(n\). This illustrates why the richness condition is needed for the lower-bound direction.
\end{remark}
%%%%%%%%%%%%%%%%%%%%%%%%%%%%%%%%%%%%%%%%%%%%%%%%%%%%%%%%%%%%%%%%%%%%%%%%%%%%%%%%%%%%%%%%%%%%%%%%%%%%%%%%%
%%%   III.D Main Theorem (organized with \mathrm{Dd} kept primary)
%%%%%%%%%%%%%%%%%%%%%%%%%%%%%%%%%%%%%%%%%%%%%%%%%%%%%%%%%%%%%%%%%%%%%%%%%%%%%%%%%%%%%%%%%%%%%%%%%%%%%%%%%

\subsection{Main Theorem: Depth as an Information Metric}
\label{subsec:main-theorem}

The coding argument is naturally stated for step complexity \(\mathrm{N}(q\mid B)\) (shortest trace length).
To keep \(\mathrm{Dd}(\cdot\mid B)\) as the paper's primary notion, we relate \(\mathrm{Dd}\) and \(\mathrm{N}\) via a
serializability assumption.

\begin{assumption}[Serializable depth witnesses]
\label{assump:serializable}
For the class of premise bases \(B\) and queries \(q\) considered, whenever \(q\in\mathrm{Cn}(B)\) and
\(\mathrm{Dd}(q\mid B) < \infty\), there exists a derivation trace \(\pi\) of \(q\) from \(B\) such that
\[
  |\pi| = \Theta \bigl(\mathrm{Dd}(q\mid B)\bigr).
\]
Equivalently, there exist constants \(c_{\min},c_{\max}>0\) such that for all such \((q,B)\),
\[
  c_{\min}\,\mathrm{Dd}(q\mid B)\le \mathrm{N}(q\mid B)\le c_{\max}\,\mathrm{Dd}(q\mid B).
\]
\end{assumption}

In particular, highly parallel or tree-like proof structures may violate this assumption by having small unfolding height but large shortest sequential trace length.

\begin{remark}[On the strength of Assumption~\ref{assump:serializable}]
\label{rem:serializable-strength}
Assumption~\ref{assump:serializable} is an \emph{instance-class} hypothesis: in general, the unfolding height
\(\mathrm{Dd}(q\mid B)\) and the shortest sequential trace length \(\mathrm{N}(q\mid B)\) need not be within constant factors.
When the assumption fails, all coding and incompressibility statements remain valid at the step scale \(\mathrm{N}(q\mid B)\),
and the BCQ instantiation via \(\mathrm{Dd}_\Sigma=\mathrm{N}\) (Corollary~\ref{cor:bcq-instantiation}) provides a step-accurate alternative.
\end{remark}

\begin{theorem}[Derivation depth as an information metric (base-parametric)]
\label{thm:derivation-depth-info-metric}
Let $B$ be a finite premise base with $m := |B|$, and let $q \in \mathrm{Cn}(B)$ be derivable with
\[
d := \mathrm{Dd}(q \mid B) \ge 1.
\]
Define
\[
L_B := \log(m+d), \qquad L := d\,L_B.
\]
Assume $\log m \le d \le \operatorname{poly}(m)$, Assumption~\ref{assump:finite-premise-base}, and Assumption~\ref{assump:serializable}.

\begin{enumerate}
\item[(i)] {Upper bound (all queries).}
There exists a constant $c_1>0$ (depending only on $P$, $U$, the encoding conventions, and the serializability constants) such that
\[
K(q \mid \langle B\rangle) \le c_1\, d \log(m+d).
\]
Equivalently,
\[
\mathrm{Dd}(q \mid B) \ge \frac{K(q \mid \langle B\rangle)}{c_1 \log(m + \mathrm{Dd}(q \mid B))}.
\]

\item[(ii)] {Generic lower bound (under richness).}
Suppose $P$ satisfies the richness condition (Definition~3.10) with parameter $\delta_0>0$. Let
\[
n := \mathrm{N}(q \mid B).
\]
Fix any constant $\delta \in (\delta_0,1)$, and assume that $q$ is $\delta$-incompressible at step scale, i.e.
\[
K(q \mid \langle B\rangle) \ge (1-\delta)\, n \log(m+n).
\]
Then there exists a constant $c_2>0$ such that
\[
\mathrm{Dd}(q \mid B) \le \frac{K(q \mid \langle B\rangle)}{c_2 (1-\delta)\log(m + \mathrm{Dd}(q \mid B))}.
\]
\end{enumerate}

Hence, for generic/incompressible queries in the above sense,
\[
\mathrm{Dd}(q \mid B) = \Theta \left(\frac{K(q \mid \langle B\rangle)}{\log(m + \mathrm{Dd}(q \mid B))}\right),
\]
equivalently,
\[
K(q \mid \langle B\rangle) = \Theta \bigl(\mathrm{Dd}(q \mid B)\log(m + \mathrm{Dd}(q \mid B))\bigr),
\]
and therefore
\[
H_{\mathrm{derive}}(q \mid B) = \Theta \left(\frac{H_K(q \mid B)}{\log(m + \mathrm{Dd}(q \mid B))}\right).
\]
\end{theorem}

\begin{proof}
Let
\[
d := \mathrm{Dd}(q \mid B), \qquad n := \mathrm{N}(q \mid B).
\]

{(i) Upper bound.}
By Assumption~\ref{assump:serializable}, we have
\[
n \le c_{\max} d
\]
for some constant $c_{\max}>0$. Let $\pi$ be a shortest derivation trace of $q$ from $B$, so $|\pi|=n$.
By Lemma~\ref{lem:encoding},
\[
\ell'(\operatorname{enc}(\pi))
 \le k n \log(m+n) + O(n)
 = O \bigl(n\log(m+n)\bigr).
\]
A program that contains $\operatorname{enc}(\pi)$ and, given $\langle B\rangle$, replays $\pi$ and outputs $q$ therefore has length
\[
O \bigl(n\log(m+n)\bigr).
\]
Hence
\[
K(q \mid \langle B\rangle) \le O \bigl(n\log(m+n)\bigr).
\]
Using $n \le c_{\max} d$ and $\log(m+n)=O(\log(m+d))$, we obtain
\[
K(q \mid \langle B\rangle) \le c_1\, d\log(m+d)
\]
for a suitable constant $c_1>0$.

\medskip
\noindent{(ii) Generic lower bound.}
Assume $q$ is $\delta$-incompressible at step scale:
\[
K(q \mid \langle B\rangle) \ge (1-\delta)\, n\log(m+n).
\]
By Assumption~\ref{assump:serializable}, we have
\[
n \ge c_{\min} d
\]
for some constant $c_{\min}>0$, and also
\[
\log(m+n)=\Theta(\log(m+d)).
\]
Therefore
\[
K(q \mid \langle B\rangle)
 \ge (1-\delta)\, c_{\min}\, d \cdot \Theta(\log(m+d)).
\]
Rearranging yields
\[
\mathrm{Dd}(q \mid B)
 \le
 \frac{K(q \mid \langle B\rangle)}
      {c_2(1-\delta)\log(m + \mathrm{Dd}(q \mid B))}
\]
for some constant $c_2>0$.

Combining (i) and (ii) gives the stated $\Theta$-relations. The entropy form follows immediately from
\[
H_{\mathrm{derive}}(q \mid B)=\mathrm{Dd}(q \mid B)\ln 2,
\qquad
H_K(q \mid B)=K(q \mid \langle B\rangle)\ln 2.
\]
\end{proof}

\begin{corollary}[BCQ instantiation (sum-rule depth without serializability)]
\label{cor:bcq-instantiation}
Let \(q\) be a BCQ and define its compositional (sum-rule) depth by
\(\mathrm{Dd}_\Sigma(q\mid B)\) (Definition~\ref{def:additive-depth}).
Then all statements of Theorem~\ref{thm:derivation-depth-info-metric} remain valid if we replace
\(\mathrm{Dd}(q\mid B)\) by \(\mathrm{Dd}_\Sigma(q\mid B)\) and drop Assumption~\ref{assump:serializable}.
\end{corollary}

\begin{proof}
For BCQs, Definition~\ref{def:additive-depth} gives
\[
\mathrm{Dd}_\Sigma(q \mid B) = \mathrm{N}(q \mid B)
\]
by definition. Hence the coding upper bound from Lemma~\ref{lem:encoding} and the incompressibility lower bound from Lemma~\ref{lem:incompressibility}
already apply directly to the quantity $\mathrm{Dd}_\Sigma(q \mid B)$ itself, without any need to compare dependency height
with shortest sequential trace length. Therefore the two-sided conclusions of Theorem~\ref{thm:derivation-depth-info-metric} remain valid after replacing
$\mathrm{Dd}(q \mid B)$ by $\mathrm{Dd}_\Sigma(q \mid B)$, and Assumption~\ref{assump:serializable} is unnecessary in this BCQ setting.
\end{proof}

\begin{corollary}[Intrinsic vs.\ operational depth as two instantiations]
\label{cor:intrinsic-operational}
Let \(S_O\) be a knowledge base and let \(A:=\mathrm{Atom}(S_O)\) be its irredundant semantic core
(Definition~\ref{def:atom-so}).
For any query \(q\) derivable from \(A\) (hence from \(S_O\)), define
\[
  d_{\mathrm{int}}(q) := \mathrm{Dd}\bigl(q \mid A\bigr),
  \qquad
  d_{\mathrm{op}}(q) := \mathrm{Dd}(q \mid S_O).
\]
Then \(d_{\mathrm{op}}(q)\le d_{\mathrm{int}}(q)\) (Remark~\ref{rem:int-op-monotone}).

Moreover, whenever the hypotheses of Theorem~\ref{thm:derivation-depth-info-metric} hold,
it applies both to \((q,B)=(q,A)\) (intrinsic analysis) and to \((q,B)=(q,S_O)\) (operational analysis),
with base sizes \(m_{\mathrm{int}}:=|A|\) and \(m_{\mathrm{op}}:=|S_O|\), respectively.
\end{corollary}

\begin{remark}[Interpretation]
\label{rem:interpretation}
Theorem~\ref{thm:derivation-depth-info-metric} connects descriptional complexity (\(K(\cdot\mid\cdot)\)) with inferential
depth (\(\mathrm{Dd}(\cdot\mid\cdot)\)). The logarithmic factor \(\log(m+\mathrm{Dd})\) represents the addressing cost of selecting
premises from a pool of size \(\lesssim m+\mathrm{Dd}\) in derivation encodings.
\end{remark}

%%%%%%%%%%%%%%%%%%%%%%%%%%%%%%%%%%%%%%%%%%%%%%%%%%%%%%%%%%%%%%%%%%%%%%%%%%%%%%%%%%%%%%%%%%%%%%%%%%%%%%%%%
%%%   III.E Connection to Bennett's Logical Depth
%%%%%%%%%%%%%%%%%%%%%%%%%%%%%%%%%%%%%%%%%%%%%%%%%%%%%%%%%%%%%%%%%%%%%%%%%%%%%%%%%%%%%%%%%%%%%%%%%%%%%%%%%

\subsection{Connection to Bennett's Logical Depth}
\label{subsec:bennett}

We now relate base-relative derivation depth (Definition~\ref{def:derivation-depth}) to Bennett's logical depth.
Since \(\mathrm{Dd}(\cdot\mid B)\) measures the height of the predecessor unfolding, while Bennett depth is defined via
running time of near-minimal programs, we connect the two by (i) serializability and (ii) simulation assumptions.

\begin{definition}[Bennett's logical depth~\cite{bennett1988logical,li2008introduction}]
\label{def:bennett-depth}
For strings \(x\) and \(y\) and a significance level \(t \ge 0\), the logical depth of \(x\) at significance \(t\)
conditional on \(y\) is
\begin{align}\label{eq:bennett-def}
  \mathrm{depth}_t(x \mid y)
  := \min\bigl\{T(p) :\;& |p| \le K(x \mid y) + t, \nonumber\\
  & U(p, y) = x\bigr\},
\end{align}
where \(U\) is a fixed universal Turing machine, \(T(p)\) is the running time of \(U(p,y)\), and \(K(\cdot\mid\cdot)\)
is (prefix-free) conditional Kolmogorov complexity under \(U\).
\end{definition}

\begin{lemma}[Depth is bounded by trace length]
\label{lem:depth-leq-trace}
Fix a finite premise base \(B\). Assume Assumption~\ref{assump:alignment} (so that the predecessor operator \(P_O\)
matches the fixed one-step normal form of the proof system \(\mathcal P\) used to define derivation traces).
Let \(\pi\) be any valid derivation trace (Definition~\ref{def:derivation-trace}) that derives \(q\) from \(B\).
If \(|\pi|=\ell\), then
\[
  \mathrm{Dd}(q \mid B) \le \ell.
\]
\end{lemma}

\begin{proof}
The claim is not purely syntactic in the trace formalism: it relies on Assumption \ref{assump:alignment}, which identifies the predecessor set \(P_O(s)\) appearing in Definition \ref{def:derivation-trace} with the immediate premises of the canonical one-step inference concluding \(s\). Under this alignment, every replay step in the trace witnesses one level in the recursive definition of \(\mathrm{Dd}\).

Let \(\pi=((r_1,\vec p_1),\ldots,(r_\ell,\vec p_\ell);\ o)\) and let \(f_i\) be the conclusion produced at step \(i\)
during replay, so that \(q\) is selected from the final pool \(P_\ell\).

We claim by induction on \(i\) that every formula in the pool \(P_i\) has derivation depth at most \(i\) from \(B\).
For \(i=0\), \(P_0\) consists exactly of \(B\), hence all elements have depth \(0\).
For the inductive step, \(P_i\) is obtained from \(P_{i-1}\) by appending \(f_i\).
By validity of the trace and Assumption~\ref{assump:alignment}, the immediate premises used to derive \(f_i\)
are exactly \(P_O(f_i)\), and each predecessor in \(P_O(f_i)\) is an element of \(P_{i-1}\).
By the induction hypothesis, each such predecessor has depth at most \(i-1\), hence
\[
\mathrm{Dd}(f_i\mid B)=1+\max_{s'\in P_O(f_i)}\mathrm{Dd}(s'\mid B)\le 1+(i-1)=i.
\]
Thus every element of \(P_i\) has depth at most \(i\). In particular, the output \(q\) selected from \(P_\ell\)
satisfies \(\mathrm{Dd}(q\mid B)\le \ell\).
\end{proof}

\begin{theorem}[Derivation depth as conditional logical depth]
\label{thm:bennett-equivalence}
Let \(S_O\) be a knowledge base with intrinsic base \(A:=\mathrm{Atom}(S_O)\), and assume
Assumption~\ref{assump:serializable}. Let \(U\) be a universal Turing machine and \(\langle\cdot\rangle\) the fixed encoding.

Assume the following two properties hold for the class of instances under consideration:
\begin{enumerate}[label=\textup{(\roman*)}, leftmargin=2.5em]
  \item \emph{Efficient proof simulation (EPS):}
  there exists a constant \(a_1>0\) such that, given \(\langle A\rangle\) and a derivation trace encoding \(\enc(\pi)\)
  of length \(|\pi|=\ell\), \(U\) can verify \(\pi\) and output \(\langle q\rangle\) in time at most \(a_1\cdot \ell\).
  \item \emph{No proof shortcut (NPS):}
  there exists a constant \(a_2>0\) such that if some program \(p\) outputs \(\langle q\rangle\) in time \(T(p)\) given \(\langle A\rangle\) as conditional input, then there exists a derivation trace \(\pi\) of \(q\) from \(A\)
  with \(|\pi|\le a_2\cdot T(p)\).
\end{enumerate}

Then for every derivable query \(q \in \mathrm{Cn}(A)\) with \(m=|A|\) and \(d=\mathrm{Dd}(q \mid A)\ge 1\),
letting \(L := d\log(m+d)\), there exists a constant \(c_t>0\) such that
\[
  \mathrm{depth}_{c_t L}\bigl(\langle q \rangle \mid \langle A\rangle\bigr)
  = \Theta \bigl(\mathrm{Dd}(q \mid A)\bigr).
\]
\end{theorem}

\begin{proof}
Let \(m=|A|\) and \(d=\mathrm{Dd}(q\mid A)\).

\emph{Upper bound.}
By Assumption~\ref{assump:serializable}, there exists a derivation trace \(\pi\) deriving \(q\) from \(A\)
with length \(|\pi|=\ell \le c_{\mathrm{ser}}\, d\) for some constant \(c_{\mathrm{ser}}>0\).
By Lemma~\ref{lem:encoding},
\begin{equation}
\begin{aligned}
\ell'({enc}(\pi))
  \le k\ell \log(m+\ell) + O(\ell) 
  = O\bigl(d\log(m+d)\bigr) 
  = O(L).
\end{aligned}
\end{equation}
Let \(p\) be a program that contains \(\enc(\pi)\), uses \(\langle A\rangle\) as conditional input, runs the EPS
verification/simulation, and outputs \(\langle q\rangle\).
Since \(|p|=O(L)\) and \(K(\langle q\rangle\mid \langle A\rangle)\le |p|+O(1)\), choosing \(c_t\) large enough ensures
\(|p| \le K(\langle q\rangle\mid \langle A\rangle) + c_t L\).
By (EPS), \(T(p)\le a_1\ell \le a_1c_{\mathrm{ser}}d\), hence
\(\mathrm{depth}_{c_t L}(\langle q\rangle\mid \langle A\rangle)=O(d)\).

\emph{Lower bound.}
Let \(T^* := \mathrm{depth}_{c_t L}(\langle q \rangle \mid \langle A\rangle)\).
By Definition~\ref{def:bennett-depth}, there exists a program \(p^*\) of admissible length at significance \(c_tL\)
that outputs \(\langle q\rangle\) in time \(T^*\) given \(\langle A\rangle\).
By (NPS), there exists a derivation trace \(\pi^*\) of \(q\) from \(A\) with \(|\pi^*|\le a_2T^*\).
By Lemma~\ref{lem:depth-leq-trace} (with base \(A\)),
\(
d=\mathrm{Dd}(q\mid A)\le |\pi^*|\le a_2T^*
\),
hence \(T^*\ge d/a_2\). Therefore
\(\mathrm{depth}_{c_t L}(\langle q\rangle\mid \langle A\rangle)=\Omega(d)\).

Combining the two bounds yields
\(\mathrm{depth}_{c_t L}(\langle q\rangle\mid \langle A\rangle)=\Theta(d)=\Theta(\mathrm{Dd}(q\mid A))\).
\end{proof}

\begin{remark}[On the time model in Theorem~\ref{thm:bennett-equivalence}]
\label{rem:bennett-time-model}
Property (EPS) treats one trace step (including addressing/pointer handling) as \(O(1)\) time on the underlying machine.
On standard bit-cost Turing machines, replaying a trace typically incurs an additional \(\Theta(\log(m+\ell))\) addressing
overhead per step, leading to a natural \(\Theta(\mathrm{Dd}(q\mid A)\log(m+\mathrm{Dd}(q\mid A)))\) time scale.
This matches the ``addressing cost'' factor that already appears in the coding bounds of
Theorem~\ref{thm:derivation-depth-info-metric}.
\end{remark}

%%%%%%%%%%%%%%%%%%%%%%%%%%%%%%%%%%%%%%%%%%%%%%%%%%%%%%%%%%%%%%%%%%%%%%%%%%%%%%%%%%%%%%%%%%%%%%%%%%%%%%%%%
%%%   III.F Tightness and Locality Refinements (kept mostly as-is)
%%%%%%%%%%%%%%%%%%%%%%%%%%%%%%%%%%%%%%%%%%%%%%%%%%%%%%%%%%%%%%%%%%%%%%%%%%%%%%%%%%%%%%%%%%%%%%%%%%%%%%%%%

\subsection{Tightness and Locality Refinements}
\label{subsec:tightness-extensions}

We first show the asymptotic characterization is tight up to constant factors.
Then we discuss refinements exploiting semantic locality.

%--------------------------------------------------------------
%  III.F.1 Tightness
%--------------------------------------------------------------
\subsubsection{Tightness analysis}
\label{subsubsec:tightness}

\begin{theorem}[Tightness of the asymptotic characterization]
\label{thm:tightness}
There exists a proof system \(\mathcal{P}\) and an infinite family of pairs \((q, B)\) with
\(m = |B|\) and \(d = \mathrm{Dd}(q \mid B)\) satisfying \(d \to \infty\), \(d = o(m)\), and the
information-rich condition (Definition~\ref{def:info-rich}), such that
\[
  \mathrm{Dd}(q \mid B)
  =
  \Theta \left(
    \frac{K\bigl(q \mid \langle B \rangle\bigr)}
         {\log\bigl(|B| + \mathrm{Dd}(q \mid B)\bigr)}
  \right).
\]
\end{theorem}

\begin{proof}
We give an explicit family over a \emph{fixed} finite relational signature in the standard descriptive-complexity setting.

\emph{Fixed signature and numerals.}
Fix a finite relational signature
\[
  \Sigma_0 := \{A(\cdot),\, <,\, \mathrm{BIT}(\cdot,\cdot)\},
\]
where \(<\) is a built-in linear order on the object sort and \(\mathrm{BIT}(x,y)\) is the standard bit predicate
available on ordered finite structures (cf.\ descriptive complexity references such as~\cite{immerman1999descriptive,ebbinghaus1995finite}).
For each integer \(i\ge 0\), we use the standard shorthand \(\underline{i}\) (``the numeral \(i\)'') to denote
a \(\Sigma_0\)-definable element with binary value \(i\). Formally, one may read an atom \(A(\underline{i})\) as the
closed formula
\[
  a_i \;:=\; \exists x\;\bigl(\mathrm{Num}_i(x)\land A(x)\bigr),
\]
where \(\mathrm{Num}_i(x)\) is a \(\Sigma_0\)-formula that holds exactly on the element encoding \(i\) in the
built-in number representation; \(\mathrm{Num}_i\) can be chosen with length \(O(\log i)\) under the usual BIT-based
encodings.

\emph{Proof system and one-step predecessor normal form.}
Let \(\mathcal P\) contain the binary conjunction-introduction rule: from \(\varphi\) and \(\psi\), derive
\(\varphi\land\psi\).
We work in a one-step normal form for this rule, so under Assumption~\ref{assump:alignment} we have
\[
  P_O(\varphi\land\psi)=\{\varphi,\psi\}.
\]

\emph{Premise bases.}
For each \(m\in\mathbb N\), define the premise base over the \emph{same} signature \(\Sigma_0\) by
\[
  B_m := \{a_0,a_1,\ldots,a_{m-1}\},
  \qquad\text{where } a_i := A(\underline{i}).
\]
Then \(|B_m|=m\), and all instances live in the fixed logical language \(\mathcal L=\mathrm{FO(LFP)}\) over the
fixed signature \(\Sigma_0\).

\emph{Parameters.}
Let \(n:=\lfloor\sqrt{m}\rfloor\) for \(m\ge 4\). Then \(n\to\infty\) and \(n=o(m)\).

\emph{Queries (canonical encoding).}
For each subset \(I\subseteq [m]\) with \(|I|=n+1\), write \(I=\{i_0<i_1<\cdots<i_n\}\) and define the left-associated
conjunction
\[
  q_I := (\cdots((a_{i_0}\land a_{i_1})\land a_{i_2})\cdots)\land a_{i_n}.
\]
Under the fixed parenthesization and increasing order, distinct \(I\)'s yield distinct encodings \(\langle q_I\rangle\).

\emph{Derivation depth.}
Since each \(a_{i_j}\in B_m\) is depth \(0\), there exists a derivation trace of \(q_I\) from \(B_m\) of length \(n\),
hence by Lemma~\ref{lem:depth-leq-trace} we have \(\mathrm{Dd}(q_I\mid B_m)\le n\).
Conversely, because \(P_O(\varphi\land\psi)=\{\varphi,\psi\}\), the predecessor unfolding of the left-associated
conjunction contains a chain of length \(n\) through the intermediate conjunctions, so \(\mathrm{Dd}(q_I\mid B_m)\ge n\).
Therefore \(\mathrm{Dd}(q_I\mid B_m)=n\).

\emph{Kolmogorov complexity bounds.}
Given \(\langle B_m\rangle\), specifying \(q_I\) amounts to specifying the subset \(I\).
There are \(\binom{m}{n+1}\) choices, hence by the standard counting bound for prefix-free Kolmogorov complexity
(see, e.g.,~\cite{li2008introduction}),
\[
  K\bigl(q_I \mid \langle B_m \rangle\bigr)\ge \log\binom{m}{n+1}-O(1).
\]
For \(n+1\le m/2\) (true for all sufficiently large \(m\)), \(\binom{m}{n+1}\ge \left(\frac{m}{n+1}\right)^{n+1}\),
so \(K(q_I\mid \langle B_m\rangle)=\Omega(n\log m)\).

On the other hand, Lemma~\ref{lem:encoding} (with \(k=2\)) yields an encoding of a length-\(n\) conjunction-building trace of
\(q_I\) from \(B_m\) of size \(\le 2n\log(m+n)+O(n)\), hence
\(K(q_I\mid \langle B_m\rangle)=O(n\log(m+n))\).
Since \(n=o(m)\), \(\log(m+n)=\Theta(\log m)\), so
\[
  K(q_I\mid \langle B_m\rangle)=\Theta(n\log m).
\]

Finally, \(\mathrm{Dd}(q_I\mid B_m)=n\) and \(\log(m+\mathrm{Dd}(q_I\mid B_m))=\Theta(\log m)\), giving
\[
  \mathrm{Dd}(q_I\mid B_m)
  =
  \Theta \left(
    \frac{K(q_I\mid \langle B_m\rangle)}{\log(m+\mathrm{Dd}(q_I\mid B_m))}
  \right).
\]
Moreover \(K(q_I\mid \langle B_m\rangle)=\Theta(n\log m)=\omega(\log m)\), so the family lies in the information-rich regime.
\end{proof}

%--------------------------------------------------------------
%  III.E.2 Locality
%--------------------------------------------------------------
\subsubsection{Semantic locality and effective complexity}
\label{subsubsec:locality}

The bounds in Theorem~\ref{thm:derivation-depth-info-metric} treat all base premises uniformly via
\(L_B=\log(m+n)\), where \(m=|B|\).
For semantic knowledge bases \(S_O\), the most relevant instantiation is the intrinsic base
\(A:=\mathrm{Atom}(S_O)\). For semantically local queries that effectively depend on only a small subset of \(A\),
one can tighten the logarithmic factor by replacing \(|A|\) with an \emph{effective} core size.

\begin{definition}[Essential core premises and effective parameters]
\label{def:essential}
Fix a knowledge base \(S_O\) with intrinsic base \(A:=\mathrm{Atom}(S_O)\).
For a derivable query \(q \in \mathrm{Cn}(A)\), let \(\Pi^*(q)\) denote the set of \emph{shortest}
derivation traces of \(q\) from \(A\), i.e.,
\[
  \Pi^*(q) := \{\pi:\ \pi \text{ derives } q \text{ from } A,\ |\pi| = \mathrm{N}(q\mid A)\},
\]
where \(\mathrm{N}(q\mid A)\) is defined in Definition~\ref{def:min-trace-length}. Define:
\begin{enumerate}[label=\textup{(\roman*)}]
  \item \emph{Essential atom set:}
  \[
    \mathrm{Ess}(q \mid S_O) := \bigcap_{\pi \in \Pi^*(q)} \Atoms(\pi)\ \subseteq\ A.
  \]

  \item \emph{Extended essential atom set:}
  \[
    \mathrm{Ess}^+(q \mid S_O) := \bigcup_{\pi \in \Pi^*(q)} \Atoms(\pi)\ \subseteq\ A.
  \]

  \item \emph{Effective knowledge-base size:}
    \[
    m_{\mathrm{eff}}(q):=|\mathrm{Ess}^+(q\mid S_O)|.
    \]

\item \emph{Effective local depth:}
\[
d_{\mathrm{eff}}(q):=\mathrm{Dd}(q\mid \mathrm{Ess}^+(q\mid S_O)).
\]

\item \emph{Effective complexity factor:}
\[
L_{\mathrm{eff}}(q):=\log(m_{\mathrm{eff}}(q)+d_{\mathrm{eff}}(q)).
\]

\item \emph{Semantic locality coefficient:}
\[
\lambda(q):=\frac{m_{\mathrm{eff}}(q)}{|A|}\in (0,1].
\]
\end{enumerate}
\end{definition}

\begin{remark}[Well-definedness and relation to \(\mathrm{Dd}\)]
\label{rem:essential-welldef}
Since \(q\) is derivable from \(A\), \(\Pi^*(q)\neq\varnothing\).
Moreover, Lemma~\ref{lem:depth-leq-trace} gives \(\mathrm{Dd}(q\mid A)\le \mathrm{N}(q\mid A)\), while
Assumption~\ref{assump:serializable} implies \(\mathrm{N}(q\mid A)=O(\mathrm{Dd}(q\mid A))\) in the regimes considered.
Hence \(\mathrm{N}(q\mid A)=\Theta(\mathrm{Dd}(q\mid A))\), so defining locality via shortest traces is compatible with using
\(\mathrm{Dd}(\cdot\mid A)\) as the main depth parameter.
\end{remark}

\begin{proposition}[Local-base characterization of effective derivation depth]
\label{prop:effective-bound}
Let \(S_O\) be a knowledge base, let \(A:=\mathrm{Atom}(S_O)\), and let \(q\in \mathrm{Cn}(A)\) be a derivable query.
Define
\[
\begin{aligned}
B_{\mathrm{eff}}(q) &:= \mathrm{Ess}^{+}(q\mid S_O),\\
m_{\mathrm{eff}}(q) &:= |B_{\mathrm{eff}}(q)|,\\
d_{\mathrm{eff}}(q) &:= \mathrm{Dd}(q\mid B_{\mathrm{eff}}(q)).
\end{aligned}
\]
Assume that \(d_{\mathrm{eff}}(q)\ge 1\), that
\[
\log m_{\mathrm{eff}}(q)\le d_{\mathrm{eff}}(q)\le \mathrm{poly}(m_{\mathrm{eff}}(q)),
\]
and that Assumption~\ref{assump:serializable} holds. Suppose moreover that the proof system satisfies the richness condition (Definition~\ref{def:richness}) relative to premise bases of size \(m_{\mathrm{eff}}(q)\). If \(q\) is generic/incompressible in the sense of Theorem~\ref{thm:derivation-depth-info-metric} with respect to \(\langle B_{\mathrm{eff}}(q)\rangle\), then
\[
d_{\mathrm{eff}}(q)
=
\Theta \left(
\frac{K \left(q \mid \langle B_{\mathrm{eff}}(q)\rangle\right)}
{L_{\mathrm{eff}}(q)}
\right),
\]
where \(L_{\mathrm{eff}}(q)=\log \bigl(m_{\mathrm{eff}}(q)+d_{\mathrm{eff}}(q)\bigr)\).
\end{proposition}

\begin{proof}
Let $B := B_{\mathrm{eff}}(q) = \mathrm{Ess}^+(q \mid S_O)$. By definition of $\mathrm{Ess}^+(q \mid S_O)$, there exists a shortest derivation trace
$\pi^* \in \Pi^*(q)$ such that $\mathrm{Atoms}(\pi^*) \subseteq B$. Since replay depends only on the base premises actually referenced
in the initial pool, restricting the initial pool to $B$ and relabeling those initial premise pointers yields a derivation trace
of the same length deriving $q$ from $B$. Hence $q \in \mathrm{Cn}(B)$.

By Definition~\ref{def:essential}, $|B| = m_{\mathrm{eff}}(q)$ and $\mathrm{Dd}(q \mid B) = d_{\mathrm{eff}}(q)$. Therefore all assumptions of Theorem~\ref{thm:derivation-depth-info-metric} are satisfied for the pair $(q,B)$. Applying Theorem~\ref{thm:derivation-depth-info-metric} yields
\[
d_{\mathrm{eff}}(q)
=
\Theta \left(
\frac{K(q \mid \langle B_{\mathrm{eff}}(q)\rangle)}{L_{\mathrm{eff}}(q)}
\right).
\]
\end{proof}

\begin{remark}[Relation to global intrinsic depth]
Since \(B_{\mathrm{eff}}(q)\subseteq A\), monotonicity of derivation depth with respect to the premise base gives
\[
\mathrm{Dd}(q\mid A)\le d_{\mathrm{eff}}(q).
\]
Thus the effective local depth provides a local semantic scale that upper-bounds the derivation depth measured relative to the full intrinsic base \(A\).
\end{remark}

\begin{corollary}[Conditional improvement from locality]
Let \(S_O\) be a knowledge base with intrinsic base \(A:=\mathrm{Atom}(S_O)\), and let \(q\in \mathrm{Cn}(A)\). If
\[
m_{\mathrm{eff}}(q)\ll |A|
\qquad\text{and}\qquad
d_{\mathrm{eff}}(q)=O \bigl(\mathrm{Dd}(q\mid A)\bigr),
\]
then
\[
L_{\mathrm{eff}}(q)
<
\log \bigl(|A|+\mathrm{Dd}(q\mid A)\bigr)
\]
for all sufficiently large instances. In particular, when
\[
d_{\mathrm{eff}}(q)\ll m_{\mathrm{eff}}(q)\ll |A|,
\]
one has
\[
\frac{\log \bigl(|A|+\mathrm{Dd}(q\mid A)\bigr)}{L_{\mathrm{eff}}(q)}
\approx
\frac{\log |A|}{\log m_{\mathrm{eff}}(q)}
=
\frac{\log |A|}{\log |A|+\log \lambda(q)}.
\]
\end{corollary}

\begin{proof}
Under the assumptions
\[
m_{\mathrm{eff}}(q) \ll |A|
\quad\text{and}\quad
d_{\mathrm{eff}}(q) = O \bigl(\mathrm{Dd}(q \mid A)\bigr),
\]
we have, for sufficiently large instances,
\[
L_{\mathrm{eff}}(q)
=
\log \bigl(m_{\mathrm{eff}}(q) + d_{\mathrm{eff}}(q)\bigr)
<
\log \bigl(|A| + \mathrm{Dd}(q \mid A)\bigr).
\]
In particular, when
\[
d_{\mathrm{eff}}(q) \ll m_{\mathrm{eff}}(q) \ll |A|,
\]
we have
\[
\log \bigl(|A| + \mathrm{Dd}(q \mid A)\bigr) \approx \log |A|,
\qquad
L_{\mathrm{eff}}(q) \approx \log m_{\mathrm{eff}}(q).
\]
Since
\[
m_{\mathrm{eff}}(q) = \lambda(q)\,|A|,
\]
it follows that
\[
\frac{\log \bigl(|A| + \mathrm{Dd}(q \mid A)\bigr)}{L_{\mathrm{eff}}(q)}
\approx
\frac{\log |A|}{\log m_{\mathrm{eff}}(q)}
=
\frac{\log |A|}{\log |A| + \log \lambda(q)}.
\]
\end{proof}

Exact computation of $\mathrm{Ess}^+(q \mid S_O)$
 may be difficult in practice; nevertheless, approximate locality estimates can still refine the effective logarithmic factor.
 
\begin{example}[Locality improvement]
\label{ex:locality}
Consider a knowledge base with \(|A| = 10^6\) atomic core facts and a query \(q\) depending on only
\(m_{\mathrm{eff}} = 10^3\) core facts, with intrinsic derivation depth \(n = 100\).
Then \(\log(|A| + n) \approx \log(10^6) \approx 19.9\) bits, while
\(L_{\mathrm{eff}} = \log(10^3 + 100) \approx 10.1\) bits, giving an improvement factor
\(\log(|A|+n)/L_{\mathrm{eff}} \approx 1.97\).
\end{example}

\begin{remark}[Computational considerations]
\label{rem:locality-computation}
Computing \(\mathrm{Ess}^+(q \mid S_O)\) exactly may be intractable if the set of shortest traces is large.
However, any single shortest trace \(\pi^*\) yields a certified lower bound
\(|\mathrm{Atoms}(\pi^*)| \le m_{\mathrm{eff}}(q)\).
In practice one may work with an efficiently computed upper bound \(\hat m \ge m_{\mathrm{eff}}(q)\),
e.g., by taking unions of premise sets from multiple near-shortest traces found by heuristic search.
\end{remark}

%%%%%%%%%%%%%%%%%%%%%%%%%%%%%%%%%%%%%%%%%%%%%%%%%%%%%%%%%%%%%%%%%%%%%%%%%%%%%%%%%%%%%%%%%%%%%%%%%%%%%%%%%
%%%   IV. STORAGE–COMPUTATION TRADEOFF AND OPTIMIZATION
%%%%%%%%%%%%%%%%%%%%%%%%%%%%%%%%%%%%%%%%%%%%%%%%%%%%%%%%%%%%%%%%%%%%%%%%%%%%%%%%%%%%%%%%%%%%%%%%%%%%%%%%%

\section{Storage--Computation Tradeoff and Optimization}
\label{sec:tradeoff-ideal}
Section~\ref{sec:derivation-entropy-ideal} links online inference cost to conditional algorithmic information: for generic information-rich queries,
the step cost (and, under Assumption~\ref{assump:serializable}, also the depth cost \(\mathrm{Dd}\)) satisfies
\[
\mathrm{Dd}(q \mid B)
=
\Theta \left(
\frac{K(q \mid \langle B\rangle)}
     {\log(|B| + \mathrm{Dd}(q \mid B))}
\right)
\]
up to constant factors. In this section we turn this correspondence into operational design principles: when to derive
on demand, when to cache, and how to allocate a finite storage budget across cache items. The point is not that
derivation depth and description length coincide exactly, but that the generic regime of Section~\ref{sec:derivation-entropy-ideal} yields a common
scale on which storage and online inference can be compared.

\paragraph{Two baselines: intrinsic vs.\ operational}
Let \(S_O\) be a finite knowledge base and let \(A := \mathrm{Atom}(S_O)\) be its irredundant semantic core (Definition~\ref{def:atom-so}).
Write \(S_O = A \cup J\), where \(J := S_O \setminus A\) are stored shortcuts. For any query
\(q \in \mathrm{Cn}(A) = \mathrm{Cn}(S_O)\), define
\[
d_{\mathrm{int}}(q) := \mathrm{Dd}(q \mid A),
\qquad
d_{\mathrm{op}}(q) := \mathrm{Dd}(q \mid S_O),
\]
with \(d_{\mathrm{op}}(q) \le d_{\mathrm{int}}(q)\) (Remark~\ref{rem:int-op-monotone}). Throughout Section~\ref{sec:tradeoff-ideal}, all tradeoff statements are
baseline-parametric: they apply relative to any chosen available premise base \(B_0\), with the intrinsic and operational
baselines as the two most natural instantiations.
%---------------------------------------------------------------------------------
%  IV.A  Frequency-Weighted Tradeoff Analysis
%---------------------------------------------------------------------------------

\subsection{Frequency-Weighted Tradeoff Analysis}
\label{subsec:freq-tradeoff}

Caching enlarges the available premise base, thereby turning expensive derivations into depth-\(0\) premises.

\subsubsection{Caching semantics as premise-base augmentation}
\label{subsubsec:augmented-bases}

\begin{definition}[Augmented-base derivation depth (baseline-parametric)]
\label{def:augmented-depth}
Fix a baseline available premise base \(B_0\subseteq S_O\) (finite and decidable).
For any finite cache set \(S\subseteq \mathrm{Cn}(B_0)\), define the augmented available premise base \(B_{0,S}:=B_0\cup S\).
For any query \(q\in \mathrm{Cn}(B_{0,S})\), define
\[
  \mathrm{Dd}(q\mid B_0\cup S)\;:=\;\mathrm{Dd}(q\mid B_{0,S}).
\]
\end{definition}

\begin{remark}[Monotonicity under caching]
\label{rem:monotone-caching}
If \(S_1\subseteq S_2\), then \(\mathrm{Dd}(q\mid B_0\cup S_2)\le \mathrm{Dd}(q\mid B_0\cup S_1)\).
\end{remark}

\subsubsection{Cost model and strategies}
\label{subsubsec:cost-model}

\begin{definition}[Cost model]
\label{def:cost-model}
Consider a system with two resources:
\begin{enumerate}[label=\textup{(\roman*)}]
  \item \emph{Storage cost} (bits): storing \(b\) bits incurs cost \(\alpha\cdot b\) for some \(\alpha>0\).
  \item \emph{Computation cost} (derivation steps): executing \(t\) steps incurs cost \(\beta\cdot t\) for some \(\beta>0\).
\end{enumerate}
Let \(\rho:=\alpha/\beta\) be the relative price of storage vs.\ computation. By choosing units we normalize \(\beta=1\).
We report online computation in the depth notation \(\mathrm{Dd}\), justified by Assumption~\ref{assump:serializable}.
For BCQs, one may use \(\mathrm{Dd}_\Sigma\) as a step-accurate replacement.
\end{definition}

\begin{definition}[Query distribution and access frequency]
\label{def:access-freq}
Let \(Q\) be a finite set of queries and \(P_Q\) a probability distribution over \(Q\).
For a horizon of \(N\) total queries, define the expected access count (frequency)
\[
  f_q := N\cdot P_Q(q).
\]
\end{definition}

\begin{definition}[On-demand derivation (baseline \(B_0\))]
\label{def:ondemand-strategy}
Fix a baseline premise base \(B_0\).
The on-demand strategy stores no additional query-specific information and answers \(q\) by deriving it from \(B_0\).
Its cost per access is
\[
  \mathrm{Cost}_{\mathrm{derive}}(q;B_0):=\mathrm{Dd}(q\mid B_0).
\]
\end{definition}

\begin{definition}[Caching a query as a premise (baseline \(B_0\))]
\label{def:caching-strategy}
Fix a baseline premise base \(B_0\).
The caching strategy stores a persistent description that makes \(q\) directly available as an additional premise,
i.e., it adds \(q\) to the cache set \(S\) so that \(q\in B_0\cup S\) and thus \(\mathrm{Dd}(q\mid B_0\cup S)=0\).

Let \(\ell_q(B_0)\) denote the number of stored bits required to cache \(q\) given that \(B_0\) is always available.
Information-theoretically, \(\ell_q(B_0)\) is lower bounded (up to \(O(1)\)) by \(K(q\mid \langle B_0\rangle)\) and can be
matched within constant slack under standard coding conventions. We use an information-theoretic benchmark: there exists a prefix-free code length
\(\mathrm{L}(q\mid \langle B_0\rangle)\) such that
\[
  K\bigl(q\mid \langle B_0\rangle\bigr)
  \;\le\;
  \mathrm{L}\bigl(q\mid \langle B_0\rangle\bigr)
  \;\le\;
  K\bigl(q\mid \langle B_0\rangle\bigr)+O(1),
\]
where the \(O(1)\) term depends only on the fixed universal machine and encoding conventions.
We take
\begin{equation}\label{eq:cache-length}
  \ell_q(B_0):=\mathrm{L}\bigl(q\mid \langle B_0\rangle\bigr).
\end{equation}
\end{definition}

\begin{definition}[Amortized cost per access (baseline \(B_0\))]
\label{def:amortized-cost}
Fix a baseline premise base \(B_0\) and a query \(q\) with frequency \(f_q>0\).
\begin{enumerate}[label=\textup{(\roman*)}]
  \item \emph{Caching:} caching \(q\) requires storing \(\ell_q(B_0)\) bits and (at most) one-time population cost.
  The amortized cost per access is modeled as
  \begin{equation}\label{eq:cost-cache}
  \mathrm{Cost}_{\mathrm{cache}}(q;B_0)
  :=
  \frac{\rho\cdot \ell_q(B_0)}{f_q}
  + \frac{\mathrm{Dd}(q \mid B_0)}{f_q}
  + c_{\mathrm{hit}},
  \end{equation}
  where \(c_{\mathrm{hit}}=O(1)\) accounts for the (model-dependent) per-access cache lookup overhead.
  \item \emph{On-demand:} the amortized cost equals \(\mathrm{Cost}_{\mathrm{derive}}(q;B_0)\) as defined in Definition~\ref{def:ondemand-strategy}.
\end{enumerate}
\end{definition}

\subsubsection{Main tradeoff theorem and critical frequency}
\label{subsubsec:main-tradeoff}

\begin{theorem}[Frequency-weighted storage--computation tradeoff (baseline-parametric)]
\label{thm:freq_tradeoff}
Fix a baseline premise base $B_0$ with $m := |B_0|$, and let $q \in \mathrm{Cn}(B_0)$ be derivable with
$d := \mathrm{Dd}(q \mid B_0) \ge 1$. Let $K_q := K(q \mid \langle B_0\rangle)$. Assume $\log m \le d \le \operatorname{poly}(m)$
and Assumption~\ref{assump:serializable}. Assume moreover that $(q,B_0)$ lies in the information-rich regime (Definition~\ref{def:info-rich}) and that
$q$ is generic/incompressible at its cost scale so that Theorem~\ref{thm:derivation-depth-info-metric} applies to $(q,B_0)$. Let $f_q \ge 1$ be the access
frequency.

Define the optimal amortized cost
\[
\mathrm{Cost}^*(q; B_0) := \min\{\mathrm{Cost}_{\mathrm{cache}}(q; B_0),\, \mathrm{Cost}_{\mathrm{derive}}(q; B_0)\}.
\]
Then:

\begin{enumerate}
\item[(i)] {Asymptotic characterization.}
\[
\mathrm{Cost}^*(q; B_0)
=
\min \left(
\frac{\rho K_q}{f_q} + O \left(\frac{d}{f_q}\right) + O(1),
\;
\Theta \left(\frac{K_q}{\log(m+d)}\right)
\right).
\]\label{eq:tradeoff-asymp}

\item[(ii)] {Critical frequency scale.}
The caching and on-demand strategies are asymptotically balanced when
\[
f_q = f_c
\quad\text{with}\quad
f_c = \Theta \bigl(\rho \log(m+d)\bigr).
\]

\item[(iii)] {Asymptotically preferable regime.}
\[
\begin{array}{ll}
\text{If } f_q \gg f_c, & \text{caching is asymptotically preferable;}\\[3pt]
\text{if } f_q \ll f_c, & \text{on-demand derivation is asymptotically preferable;}\\[3pt]
\text{if } f_q = \Theta(f_c), & \text{the two strategies have comparable amortized cost.}
\end{array}
\]
More explicitly, when $f_q \gg f_c$,
\[
\mathrm{Cost}^*(q; B_0)=\mathrm{Cost}_{\mathrm{cache}}(q; B_0)=O(1),
\]
whereas when $f_q \ll f_c$,
\[
\mathrm{Cost}^*(q; B_0)=\mathrm{Cost}_{\mathrm{derive}}(q; B_0)
=
\Theta \left(\frac{K_q}{\log(m+d)}\right)
=
\mathrm{Dd}(q \mid B_0)
\]
up to the generic equivalence from Theorem~\ref{thm:derivation-depth-info-metric}.
\end{enumerate}
\end{theorem}

\begin{proof}
By Definition~\ref{def:caching-strategy} and \eqref{eq:cache-length}, we have
\[
\ell_q(B_0) = K_q + O(1).
\]
Therefore, by Definition~\ref{def:amortized-cost},
\[
\mathrm{Cost}_{\mathrm{cache}}(q; B_0)
=
\frac{\rho\,(K_q + O(1))}{f_q}
+
\frac{d}{f_q}
+
O(1)
=
\frac{\rho K_q}{f_q}
+
O \left(\frac{d}{f_q}\right)
+
O(1).
\]
By Theorem~\ref{thm:derivation-depth-info-metric} applied to $(q,B_0)$ in the generic regime,
\[
\mathrm{Cost}_{\mathrm{derive}}(q; B_0)
=
d
=
\Theta \left(\frac{K_q}{\log(m+d)}\right).
\]
Taking the minimum yields part~(i).

For the break-even scale, compare the leading storage-amortization term
\[
\frac{\rho K_q}{f_q}
\]
with the on-demand term
\[
\Theta \left(\frac{K_q}{\log(m+d)}\right).
\]
Balancing these two quantities gives
\[
f_q = \Theta \bigl(\rho \log(m+d)\bigr),
\]
which is the claimed critical frequency scale $f_c$.

It remains to justify that the population term $d/f_q$ does not alter this scale. Using Theorem~\ref{thm:derivation-depth-info-metric} again,
\[
d=\Theta \left(\frac{K_q}{\log(m+d)}\right).
\]
Hence, at $f_q=\Theta(\rho \log(m+d))$,
\[
\frac{d}{f_q}
=
\Theta \left(
\frac{K_q/\log(m+d)}
     {\rho \log(m+d)}
\right)
=
\Theta \left(
\frac{K_q}{\rho \log^2(m+d)}
\right),
\]
which is smaller than each of the two balanced leading terms by a factor of order
\[
\frac{1}{\log(m+d)}=o(1).
\]
Thus the critical scale remains
\[
f_c=\Theta \bigl(\rho \log(m+d)\bigr).
\]

Finally, if $f_q \gg f_c$, then
\[
\frac{\rho K_q}{f_q}
=
o \left(\frac{K_q}{\log(m+d)}\right),
\qquad
\frac{d}{f_q}
=
o(d),
\]
so caching asymptotically dominates on-demand derivation, and
\[
\mathrm{Cost}^*(q; B_0)=\mathrm{Cost}_{\mathrm{cache}}(q; B_0)=O(1).
\]
If instead $f_q \ll f_c$, then
\[
\frac{\rho K_q}{f_q}
\gg
\frac{K_q}{\log(m+d)},
\]
so the on-demand strategy is asymptotically preferable. The case $f_q=\Theta(f_c)$ gives comparable amortized cost
for the two strategies.
\end{proof}
Thus, in the generic regime, the break-even rule depends only logarithmically on the effective premise-pool size and derivation depth.

\subsubsection{BCQ instantiation: constant-window \(f_c\) and entropy form}
\label{subsubsec:fc-entropy}

\begin{corollary}[Constant-window break-even frequency for generic BCQs]
\label{cor:fc-window}
Fix a baseline premise base $B_0$ with $m := |B_0|$ and let $q$ be a BCQ. Measure online computation by
compositional depth $d_\Sigma := \mathrm{Dd}_\Sigma(q \mid B_0)$. Let
\[
L_B := \log(m+d_\Sigma),
\qquad
K_q := K(q \mid \langle B_0\rangle).
\]
Assume the generic/incompressible condition at the relevant cost scale holds so that the two-sided bounds induced by
Theorem~\ref{thm:derivation-depth-info-metric} apply with $\mathrm{Dd}$ replaced by $\mathrm{Dd}_\Sigma$ (cf.\ Corollary~\ref{cor:bcq-instantiation}), i.e.,
\[
\frac{K_q}{c_1 L_B}
\le
d_\Sigma
\le
\frac{K_q}{(1-\delta)L_B}
\]
for some $\delta > \delta_0$.

Then the following are one-sided sufficient conditions for preferring caching or derivation, respectively:

\begin{enumerate}
\item[(i)] {Caching-favorable region.}
There exists a constant $C_{\mathrm{hit}}$ such that if
\[
f_q \ge \rho c_1 L_B + C_{\mathrm{hit}},
\]
then
\[
\mathrm{Cost}_{\mathrm{cache}}(q; B_0) \le \mathrm{Cost}_{\mathrm{derive}}(q; B_0) + O(1)
\]
when online cost is measured by $Dd_\Sigma$.

\item[(ii)] {Derivation-favorable region.}
If
\[
f_q \le \rho (1-\delta)L_B,
\]
then
\[
\mathrm{Cost}_{\mathrm{derive}}(q; B_0) \le \mathrm{Cost}_{\mathrm{cache}}(q; B_0) + O(1)
\]
when online cost is measured by $\mathrm{Dd}_\Sigma$.
\end{enumerate}

In particular, the BCQ critical frequency scale satisfies
\[
f_c = \Theta \bigl(\rho \log(m+d_\Sigma)\bigr).
\]
\end{corollary}

\begin{corollary}[Entropy-form expected on-demand cost for BCQs]
\label{cor:entropy-on-demand}
Assume there exist deterministic bounds \(1\le d_{\min}\le \mathrm{Dd}_\Sigma(Q\mid B_0)\le d_{\max}\) almost surely on the support of \(Q\), and that the instance-wise generic condition holds on that support. Then
\[
\mathbb{E}[\mathrm{Dd}_\Sigma(Q\mid B_0)]
=
\Theta \left(
\mathbb{E}[K(Q\mid \langle B_0\rangle)]
\right)
\]
up to a logarithmic factor lying between
\(\log(m+d_{\min})\) and \(\log(m+d_{\max})\); more precisely,
\[
c_1\frac{\mathbb{E}[K(Q\mid\langle B_0\rangle)]}{\log(m+d_{\max})}
\le
\mathbb{E}[\mathrm{Dd}_\Sigma(Q\mid B_0)]
\le
c_2\frac{\mathbb{E}[K(Q\mid\langle B_0\rangle)]}{\log(m+d_{\min})}.
\]
If, moreover, \(\mathrm{Dd}_\Sigma(Q\mid B_0)\) is concentrated so that
\(\log(m+\mathrm{Dd}_\Sigma(Q\mid B_0))=\Theta(\log(m+\bar d))\) uniformly on the support for some scale \(\bar d\), then
\[
\mathbb{E}[\mathrm{Dd}_\Sigma(Q\mid B_0)]
=
\Theta \left(
\frac{\mathbb{E}[K(Q\mid\langle B_0\rangle)]}{\log(m+\bar d)}
\right).
\]
\end{corollary}

\begin{remark}[Interpretation: bits-per-step and its role in \(f_c\)]
\label{rem:bits-per-step}
For generic BCQs, Theorem~\ref{thm:derivation-depth-info-metric} with \(n=\mathrm{Dd}_\Sigma\) implies that each unit of online computation
(one trace step) corresponds to \(\Theta(\log(m+\mathrm{Dd}_\Sigma))\) bits of conditional algorithmic information.
This explains why the break-even access frequency depends essentially only on \(\rho\log(m+\cdot)\).
\end{remark}

%---------------------------------------------------------------------------------
%  IV.B  System-Wide Allocation and Submodular Optimization
%---------------------------------------------------------------------------------

\subsection{System-Wide Allocation and Submodular Optimization}
\label{subsec:system-allocation}

We now move from the single-query tradeoff to a system-level allocation problem: given a finite storage budget,
which cache items should be materialized so as to reduce expected online derivation depth as much as possible?
By default in this subsection we take the operational baseline $B_0 := S_O$. The approximation guarantee stated
below is standard once the resulting objective is monotone submodular; the present contribution is therefore the
formulation of derivation-depth reduction as the objective, not a new approximation algorithm.

\begin{remark}[Notational warning]
\label{rem:K-ambiguity}
In this subsection, \(M\) denotes the number of queries in the finite query set \(Q=\{q_1,\ldots,q_M\}\).
\end{remark}

\begin{remark}[Notational warning: budget vs.\ premise bases]
\label{rem:B-ambiguity}
Earlier sections use \(B\) to denote an available premise base in \(\mathrm{Dd}(\cdot\mid B)\).
In this subsection we reserve \(\mathcal{B}>0\) for the storage \emph{budget} (in bits) to avoid overloading.
\end{remark}

\subsubsection{Problem formulation}
\label{sec:system-formulation}

Consider a knowledge system with:
\begin{itemize}
  \item a knowledge base \(S_O\) with intrinsic core \(A:=\mathrm{Atom}(S_O)\);
  \item a finite query set \(Q=\{q_1,\ldots,q_M\}\subseteq \mathrm{Cn}(A)\) with distribution \(P_Q\), where \(p_i:=P_Q(q_i)>0\);
  \item a candidate cache ground set \(\mathcal{U}\subseteq \mathrm{Cn}(A)\setminus S_O\);
  \item a storage cost function \(\sigma:\mathcal{U}\to\mathbb{R}_{>0}\) (bits);
  \item a storage budget \(\mathcal{B}>0\) (bits).
\end{itemize}
For any allocation \(S\subseteq \mathcal{U}\), define \(\sigma(S):=\sum_{u\in S}\sigma(u)\).

For each query \(q_i\) and allocation \(S\), define the operational depth under caching
\[
  n_i(S):=\mathrm{Dd}(q_i\mid S_O\cup S).
\]

\begin{definition}[Expected operational depth and reduction]
\label{def:expected-deriv-cost}
For \(S\subseteq\mathcal{U}\), define the expected online derivation cost
\[
  \bar n(S):=\E_{q\sim P_Q} \bigl[\mathrm{Dd}(q\mid S_O\cup S)\bigr]
  =
  \sum_{i=1}^{M} p_i\,n_i(S),
\]
and the expected cost reduction
\[
  \Delta(S):=\bar n(\varnothing)-\bar n(S)
  =\sum_{i=1}^{M} p_i\,[n_i(\varnothing)-n_i(S)]\ge 0.
\]
\end{definition}

The budgeted allocation problem is
\[
  \max_{S\subseteq \mathcal{U}}\ \Delta(S)
  \quad \text{s.t.} \quad
  \sigma(S)\le \mathcal{B}.
\]

\begin{remark}[Choosing storage costs \(\sigma(\cdot)\)]
A canonical choice consistent with Section~\ref{subsec:freq-tradeoff} is
\(\sigma(u)=\ell_u(S_O)\), where \(\ell_u(S_O)\) is the (prefix-free) code length used to store \(u\) given the always-available base
\(\langle S_O\rangle\) (cf.\ Definition~\ref{def:caching-strategy}).
In practice, \(\sigma(u)\) can be instantiated by a fixed compression pipeline that upper bounds this benchmark length.
\end{remark}

\subsubsection{Submodularity of depth reduction}
\label{sec:submodularity}

\begin{definition}[Marginal reduction]
\label{def:marginal-reduction}
For \(S\subseteq\mathcal{U}\) and \(u\in\mathcal{U}\setminus S\), define the marginal gain
\[
  \Delta_u(S):=\Delta(S\cup\{u\})-\Delta(S)=\bar n(S)-\bar n(S\cup\{u\}).
\]
\end{definition}

This system-level formulation treats cached items as storage investments whose value lies in their expected reduction of operational depth across the query distribution.

\begin{theorem}[A sufficient condition for submodularity of expected depth reduction]
\label{thm:submodularity-deriv}
Assume that for each query \(q_i\) and each \(u\in\mathcal{U}\), the depth reduction exhibits diminishing returns:
for all \(A\subseteq B\subseteq\mathcal{U}\) and \(u\in\mathcal{U}\setminus B\),
\[
  n_i(A)-n_i(A\cup\{u\})
  \ge
  n_i(B)-n_i(B\cup\{u\}).
\]
Then \(\Delta\) is normalized, monotone, and submodular.
\end{theorem}

\begin{proof}
Normalization and monotonicity are immediate.
For submodularity, for \(A\subseteq B\) and \(u\notin B\),
\begin{align*}
  \Delta_u(A)-\Delta_u(B) 
  =
  \sum_{i=1}^M p_i\Bigl( [n_i(A) - n_i( A \cup \{u\})]-[n_i(B)-n_i(B\cup\{u\})]\Bigr) 
  \ge &0.
\end{align*}
\end{proof}

\begin{remark}[On the diminishing-returns hypothesis]
The condition in Theorem~\ref{thm:submodularity-deriv} is an instance/workload assumption.
It is not implied by monotonicity of \(\mathrm{Dd}(\cdot\mid \cdot)\) alone, but can be validated empirically
(e.g., by sampling marginals \(\Delta_u(S)\)) and is often a reasonable approximation in cache/view selection workloads.
\end{remark}

\subsubsection{Greedy algorithm and approximation guarantee}
\label{sec:greedy-algorithm}

\begin{algorithm}[ht]
\caption{Greedy Storage Allocation for Submodular Knapsack}
\label{alg:greedy}
\begin{algorithmic}[1]
\Require Ground set \(\mathcal U\), cost function \(\sigma\), budget \(\mathcal B\), and value-oracle access to \(\Delta\)
\Ensure Selected set \(S_{\mathrm{out}}\)
\State \(S_{\mathrm{out}} \gets \varnothing\)
\ForAll{\(G \subseteq \mathcal U\) such that \(|G| \le 3\) and \(\sigma(G) \le \mathcal B\)}
    \State \(S \gets G\)
    \While{there exists \(u \in \mathcal U \setminus S\) with \(\sigma(u) \le \mathcal B - \sigma(S)\)}
        \State Choose \(u^\ast \in \mathcal U \setminus S\) maximizing
        \Statex \hspace{\algorithmicindent}
        \( \Delta_u(S)/\sigma(u) \)
        among all \(u\) such that \(\sigma(u) \le \mathcal B - \sigma(S)\)
        \State \(S \gets S \cup \{u^\ast\}\)
    \EndWhile
    \If{\(\Delta(S) > \Delta(S_{\mathrm{out}})\)}
        \State \(S_{\mathrm{out}} \gets S\)
    \EndIf
\EndFor
\State \Return \(S_{\mathrm{out}}\)
\end{algorithmic}
\end{algorithm}

\begin{theorem}[Approximation guarantee under a knapsack constraint]
\label{thm:greedy-guarantee}
Let \(U\) be a finite ground set, and let
\(\Delta:2^U\to \mathbb{R}_{\ge 0}\) be a normalized, monotone, and submodular set function.
Let \(\sigma:U\to \mathbb{R}_{\ge 0}\) be a nonnegative modular cost function, extended to
subsets by
\[
  \sigma(S):=\sum_{u\in S}\sigma(u), \qquad S\subseteq U.
\]
Let \(\mathcal{B}\ge 0\) be a budget, and let
\[
  S^* \in \arg\max\{\Delta(S): S\subseteq U,\ \sigma(S)\le \mathcal{B}\}.
\]
Then Algorithm~\ref{alg:greedy} returns a set \(S_{\mathrm{out}}\subseteq U\) satisfying
\[
  \Delta(S_{\mathrm{out}})
  \ge
  \left(1-\frac{1}{e}\right)\Delta(S^*).
\]
\end{theorem}

\begin{proof}
This is the standard approximation guarantee for monotone submodular maximization under a knapsack constraint
using density-greedy together with constant-size partial enumeration; see Sviridenko~\cite{sviridenko2004note}. We include it here
only to record the guarantee inherited by the depth-reduction objective once the assumptions of Theorem~\ref{thm:submodularity-deriv} hold.
\end{proof}

For large-scale deployments, semantic clustering can be used as a heuristic preprocessing step to reduce the candidate universe before solving the budgeted allocation problem; a representative formulation is deferred to the appendix.
%---------------------------------------------------------------------------------
%  IV.C  Semantic Clustering for Practical Deployment
%---------------------------------------------------------------------------------

%%%%%%%%%%%%%%%%%%%%%%%%%%%%%%%%%%%%%%%%%%%%%%%%%%%%%%%%%%%%%%%%%%%%%%%%%%%%%%%%%%%%%%%%%%%%%%%%%%%%%%%%%
%   V.  EXTENSIONS TO NOISY INFORMATION
%%%%%%%%%%%%%%%%%%%%%%%%%%%%%%%%%%%%%%%%%%%%%%%%%%%%%%%%%%%%%%%%%%%%%%%%%%%%%%%%%%%%%%%%%%%%%%%%%%%%%%%%%
\section{Extensions to Noisy Information}
\label{sec:noisy-extensions}
Sections~\ref{sec:model}--\ref{sec:tradeoff-ideal} develop an ideal (noise-free) theory in which derivational cost is measured relative to a finite available premise base \(B\) via \(\mathrm{Dd}(\cdot \mid B)\) (Definition\ref{def:derivation-depth}). Section~\ref{sec:noisy} introduced the noisy setting abstractly as perturbation of the available premise base. We now specialize that setup and study how loss and pollution affect three quantities: derivation depth, conditional description length, and the resulting storage--computation tradeoffs.

Throughout, noise acts on the available premise base itself rather than on the definition of derivation depth. Thus the central object remains \(\mathrm{Dd}(\cdot \mid \tilde{B})\), where \(\tilde{B}\) is a finite noisy variant of a baseline base \(B_0\). We maintain the standing requirement that any such noisy base have decidable membership, so that the computability guarantee of Theorem~\ref{thm:computability-Dd} continues to apply.

Our emphasis is cost robustness rather than a full semantic-correctness theory under corruption. In particular, we quantify how premise loss and pollution perturb derivability cost, description length, and caching decisions, while treating stronger semantic-validity questions only indirectly through robust penalties and compensation mechanisms. This is also consistent with the operational viewpoint of Section~\ref{sec:tradeoff-ideal}: caching enlarges the available premise base, whereas noise may either shrink it (loss) or enlarge it (pollution). In deployment, both effects may coexist, so the relevant available base takes the form \(\tilde{B}_0 \cup S\).

%---------------------------------------------------------------------------------
%  V.A Setup: premise perturbations and base-relative depth
%---------------------------------------------------------------------------------
\subsection{Setup: premise perturbations and base-relative depth}
\label{subsec:noisy-setup}

\paragraph{Baseline premise bases}
Recall the two natural baselines from Section~\ref{sec:tradeoff-ideal}:
the intrinsic baseline \(B_0=A:=\mathrm{Atom}(S_O)\) and the operational baseline \(B_0=S_O\).
Under Assumption~\ref{assump:finite-so}, both are finite available premise bases with decidable membership.
Section~\ref{sec:tradeoff-ideal} is operational by default (baseline \(B_0=S_O\)), while still comparing to the intrinsic baseline \(A\).
Operationally, noise acts by perturbing the available premise base itself, either by removing usable premises or by introducing spurious ones.

\begin{definition}[Noisy premise base (baseline-parametric)]
\label{def:noisy-base}
Fix a finite baseline available premise base \(B_0\subseteq S_O\).
A \emph{noisy premise base} associated with \(B_0\) is any finite set of the form
\begin{equation}\label{eq:noisy-base}
  \tilde B_0 := (B_0 \setminus B_0^-) \cup B_0^+,
\end{equation}
where \(B_0^- \subseteq B_0\) is a finite set of \emph{lost premises} and
\(B_0^+ \subseteq S_O\setminus B_0\) is a finite set of \emph{spurious added premises}.
\end{definition}

\begin{remark}[Effectivity requirements for noisy bases]
\label{rem:noisy-effectivity}
When we write \(\mathrm{Dd}(q\mid \tilde B_0)\), we implicitly assume membership in \(\tilde B_0\) is decidable
(e.g., \(\tilde B_0\) is given by an explicit list or by a decidable index set).
Under this effectivity assumption and the finiteness in Definition~\ref{def:noisy-base},
Theorem~\ref{thm:computability-Dd} applies with premise base \(\tilde B_0\).
\end{remark}

\begin{remark}[Noisy derivation depth]
Given a noisy base \(\tilde B_0\), the noisy depth of a query \(q\) is \(\mathrm{Dd}(q\mid \tilde B_0)\) in the sense of
Definition~\ref{def:derivation-depth}. This is the only notion of derivation cost used in this section.
(Noise may also affect semantic correctness; here we focus on cost, while robustness mechanisms are introduced later.)
\end{remark}

\begin{remark}[Consistency with Sections~\ref{sec:derivation-entropy-ideal}--\ref{sec:tradeoff-ideal}]
In the ideal regime, the premise base is simply the chosen baseline \(B_0\) (intrinsic \(A\) or operational \(S_O\)).
Section~\ref{sec:tradeoff-ideal} already treats caching as base augmentation \(B_0\mapsto B_0\cup S\)
(Definition~\ref{def:augmented-depth}); Section~\ref{sec:noisy-extensions} keeps exactly the same semantics
but allows \(B_0\) itself to be perturbed into \(\tilde B_0\).
\end{remark}

%---------------------------------------------------------------------------------
%  V.B Derivation entropy under noise
%---------------------------------------------------------------------------------
\subsection{Derivation Entropy under Noise}
\label{subsec:noisy-entropy}

\begin{assumption}[Bounded and effectively describable noise]
\label{assump:bounded-noise}
Let \(B_0\) be a baseline premise base with \(m:=|B_0|\), and let \(\tilde B_0\) be as in
Definition~\ref{def:noisy-base}. Write \(\mathcal N:=|B_0^-|+|B_0^+|\).
We assume:

\begin{enumerate}[label=\textup{(\roman*)}]
  \item (\emph{Polynomial naming universe}) There exists a computable family of universes
  \(\{\mathcal V_s\}_{s\in\mathbb N}\subseteq S_O\) with canonical indexing such that
  \(|\mathcal V_s|\le \poly(s)\) and
  \[
    B_0 \cup B_0^+ \subseteq \mathcal V_{m+\mathcal N}.
  \]
  Consequently, every element of \(B_0\cup B_0^+\) can be named by an index using
  \(O(\log(m+\mathcal N))\) bits.

  \item (\emph{Two-way effective noise description}) Given either \(\langle B_0\rangle\) or \(\langle \tilde B_0\rangle\),
  the finite sets \(B_0^-\) and \(B_0^+\) are effectively describable (and membership is decidable) by
  index lists over \(\mathcal V_{m+\mathcal N}\) of total length
  \(O(\mathcal N\log(m+\mathcal N))\) under the fixed encoding conventions.
\end{enumerate}
\end{assumption}

\begin{definition}[Derivation entropy under a premise base]
\label{def:derive-entropy-base}
For a finite premise base \(B\) and a query \(q\in\mathrm{Cn}(B)\),  define 
\[
  H_{\mathrm{derive}}(q\mid \tilde B_0) := \mathrm{Dd}(q\mid \tilde B_0)\cdot \ln 2.
\]
\end{definition}

\begin{remark}[Conditional algorithmic entropy under noise]
\label{rem:hk-noisy}
For any finite premise base \(B\), define
\[
  H_K(q\mid B) := K\left(q \mid \langle B\rangle\right)\cdot \ln 2.
\]
In particular, \(H_K(q\mid \tilde B_0)=K(q\mid \langle \tilde B_0\rangle)\cdot \ln 2\).
\end{remark}

\begin{lemma}[Base-conversion description length]
\label{lem:base-conversion}
Let \(B_0\) be a baseline base and \(\tilde B_0=(B_0\setminus B_0^-)\cup B_0^+\) a noisy base.
Let \(m:=|B_0|\) and \(\mathcal N:=|B_0^-|+|B_0^+|\).
Under Assumption~\ref{assump:bounded-noise},
\begin{align*}
K\left(\langle\tilde B_0\rangle \mid \langle B_0\rangle\right)
=
O\bigl(\mathcal N\log(m+\mathcal N)\bigr),
\\
K\left(\langle B_0\rangle \mid \langle\tilde B_0\rangle\right)
=
O\bigl(\mathcal N\log(m+\mathcal N)\bigr).
\end{align*}
\end{lemma}

\begin{proof}
We describe each base from the other by specifying the symmetric difference using index lists over the
polynomial naming universe \(\mathcal V_{m+\mathcal N}\).

\emph{From \(\langle B_0\rangle\) to \(\langle \tilde B_0\rangle\).}
Given \(\langle B_0\rangle\), encode the finite sets \(B_0^-\) and \(B_0^+\) as index lists in
\(\mathcal V_{m+\mathcal N}\). By Assumption~\ref{assump:bounded-noise}(i), each index costs
\(O(\log(m+\mathcal N))\) bits, and by (ii) the total description length is
\(O(\mathcal N\log(m+\mathcal N))\).
From \(\langle B_0\rangle\) and these lists we can compute \(\langle \tilde B_0\rangle\) effectively.

\emph{From \(\langle \tilde B_0\rangle\) to \(\langle B_0\rangle\).}
Similarly, given \(\langle \tilde B_0\rangle\), provide the same two index lists for \(B_0^-\) and \(B_0^+\)
(over \(\mathcal V_{m+\mathcal N}\)); Assumption~\ref{assump:bounded-noise}(ii) guarantees such a two-way effective description.
Then \(B_0\) is recovered by adding \(B_0^-\) and removing \(B_0^+\), and \(\langle B_0\rangle\) is computable
from \(\langle \tilde B_0\rangle\) plus this description.

In both directions, the program overhead is \(O(1)\), so the conditional Kolmogorov complexities satisfy the stated bounds.
\end{proof}

\begin{proposition}[Noisy derivation entropy and conditional-description comparison]
\label{prop:noisy-bound}
Let $B_0$ be a baseline base, $\tilde B_0$ a noisy base as in Definition~\ref{def:noisy-base}, $m := |B_0|$, $\tilde m := |\tilde B_0|$,
and $\mathcal{N} := |B_0^-| + |B_0^+|$. Assume Assumption~\ref{assump:bounded-noise} and Assumption~\ref{assump:serializable}.

Let $q$ be a query derivable from $\tilde B_0$, with noisy depth
$\tilde n := Dd(q \mid \tilde B_0)$ satisfying
$\log \tilde m \le \tilde n \le \operatorname{poly}(\tilde m)$.
Assume $(q,\tilde B_0)$ lies in the information-rich regime and is generic/incompressible at its depth scale so that
Theorem~\ref{thm:derivation-depth-info-metric} applies with premise base $\tilde B_0$. Then:

\begin{enumerate}
\item[(i)] {Noisy derivation-entropy characterization.}
\[
H_{\mathrm{derive}}(q \mid \tilde B_0)
=
\Theta \left(
\frac{H_K(q \mid \tilde B_0)}
     {\log(\tilde m + \tilde n)}
\right).
\]

\item[(ii)] {Conditional-description comparison under base conversion.}
\[
\bigl|H_K(q \mid \tilde B_0) - H_K(q \mid B_0)\bigr|
\le
O\bigl(\mathcal{N} \log(m+\mathcal{N})\bigr).
\]
\end{enumerate}
\end{proposition}

\begin{proof}
(i) Apply Theorem~\ref{thm:derivation-depth-info-metric} with premise base \(\tilde B_0\).

(ii) Using the conditioning inequality
\(K(x\mid y)\le K(x\mid z)+K(z\mid y)+O(1)\) and Lemma~\ref{lem:base-conversion}, we have
\[
K(q\mid \langle B_0\rangle)
\le
K(q\mid \langle \tilde B_0\rangle)+O(\mathcal N\log(m+\mathcal N)),
\]
and symmetrically with \(B_0\) and \(\tilde B_0\) swapped. Multiply by \(\ln 2\).
\end{proof}

%---------------------------------------------------------------------------------
%  V.C Noise effects on tradeoff parameters
%---------------------------------------------------------------------------------
\subsection{Noise Effects on Tradeoff Parameters}
\label{subsec:noise-perturb}

We quantify how premise loss/pollution perturbs the parameters that govern the tradeoffs and optimization
problems in Section~\ref{sec:tradeoff-ideal}. The key asymmetry is monotonic:
\emph{adding premises cannot increase depth}, while \emph{removing premises cannot decrease depth}.

\begin{lemma}[Monotonicity in the premise base]
\label{lem:premise-monotonicity}
If \(B_1\subseteq B_2\) are finite available premise bases, then for every \(s\in S_O\),
\[
  \mathrm{Dd}(s\mid B_2)\le \mathrm{Dd}(s\mid B_1).
\]
\end{lemma}

\begin{proof}
This follows by well-founded recursion on the predecessor unfolding using \eqref{eq:Dd-inductive}:
the base case set \(\{s:s\in B\}\) only grows with \(B\), and the recursive clause takes a maximum over predecessors.
\end{proof}

\begin{definition}[Reconstruction depth from preserved premises]
\label{def:rec-depth-preserved}
Fix a baseline \(B_0\) and loss set \(B_0^-\subseteq B_0\). Let \(B_\cap:=B_0\setminus B_0^-\) be the preserved base.
Define the (worst-case) reconstruction depth of lost premises from preserved premises as
\[
d_{\mathrm{rec}}
:=
\max_{b \in B_0^-} \mathrm{Dd}\bigl(b \mid B_\cap\bigr),
\]
with the convention \(d_{\mathrm{rec}}=\infty\) if some lost premise \(b\in B_0^-\) is not derivable from \(B_\cap\),
i.e., \(b\notin \mathrm{Cn}(B_\cap)\).
\end{definition}

\begin{remark}[Why reconstruction depth is meaningful]
An element is depth-\(0\) only relative to the chosen premise base (Definition~\ref{def:derivation-depth}).
Thus a premise \(b\in B_0\) can be atomic/available under \(B_0\) (depth \(0\)) yet have positive reconstruction
depth when removed from the available base.
\end{remark}

\begin{lemma}[Noise perturbation of depth and description length]
\label{lem:noise-perturbation}
Let \(B_0\) be a baseline base, \(\tilde B_0=(B_0\setminus B_0^-)\cup B_0^+\) a noisy base, and let \(q\) be derivable
from both \(B_0\) and \(\tilde B_0\).
Let \(m:=|B_0|\), \(\tilde m:=|\tilde B_0|\), \(n:=\mathrm{Dd}(q\mid B_0)\), and \(\tilde n:=\mathrm{Dd}(q\mid \tilde B_0)\).
Under Assumption~\ref{assump:bounded-noise}:
\begin{enumerate}[label=\textup{(\roman*)}]
  \item \emph{Base-size perturbation:} \(|\tilde m-m|\le \mathcal N\).

  \item \emph{Loss-driven depth degradation:}
  letting \(B_\cap:=B_0\setminus B_0^-\) and \(d_{\mathrm{rec}}\) be as in Definition~\ref{def:rec-depth-preserved},
  \begin{equation}\label{eq:depth-degrade}
    \mathrm{Dd}(q\mid B_\cap)\le n + d_{\mathrm{rec}}.
  \end{equation}
  Moreover, since \(B_\cap \subseteq \tilde B_0\), pollution cannot increase depth:
  \[
    \tilde n=\mathrm{Dd}(q\mid \tilde B_0)\le \mathrm{Dd}(q\mid B_\cap).
  \]
  In particular, \(\tilde n \le n+d_{\mathrm{rec}}\).

  \item \emph{Algorithmic-entropy perturbation:}
  \begin{equation}\label{eq:entropy-perturb-C}
    \bigl|H_K(q \mid \tilde B_0) - H_K(q \mid B_0)\bigr|
    \le O \bigl(\mathcal N \log(m+\mathcal N)\bigr).
  \end{equation}
\end{enumerate}
\end{lemma}

\begin{proof}
(i) Immediate from \(\tilde B_0=(B_0\setminus B_0^-)\cup B_0^+\).

(ii) Let \(B_\cap:=B_0\setminus B_0^-\). Assume \(d_{\mathrm{rec}}<\infty\); otherwise the bound is trivial.
We prove the stronger claim that for all \(s\in S_O\),
\begin{equation}\label{eq:depth-shift}
  \mathrm{Dd}(s\mid B_\cap)\le \mathrm{Dd}(s\mid B_0)+d_{\mathrm{rec}}.
\end{equation}
Consider the backward predecessor unfolding DAG of \(s\) (Proposition~\ref{prop:no-cycle-dag}) and proceed in a
topological order (equivalently, by well-founded recursion on the unfolding height).

If \(s\in B_\cap\), then \(\mathrm{Dd}(s\mid B_\cap)=0\le \mathrm{Dd}(s\mid B_0)+d_{\mathrm{rec}}\).

If \(s\in B_0^-\), then \(\mathrm{Dd}(s\mid B_0)=0\) and by definition of \(d_{\mathrm{rec}}\) we have
\(\mathrm{Dd}(s\mid B_\cap)\le d_{\mathrm{rec}}\), so \eqref{eq:depth-shift} holds.

Otherwise \(s\notin B_0\), and by \eqref{eq:Dd-inductive},
\[
  \mathrm{Dd}(s\mid B_\cap)=1+\max_{s'\in P_O(s)}\mathrm{Dd}(s'\mid B_\cap).
\]
By the induction hypothesis, for every \(s'\in P_O(s)\),
\(\mathrm{Dd}(s'\mid B_\cap)\le \mathrm{Dd}(s'\mid B_0)+d_{\mathrm{rec}}\), hence

\begin{equation}
\begin{aligned}
\mathrm{Dd}(s \mid B_\cap)
  &\le 1 + \max_{s' \in P_O(s)}
      \bigl(\mathrm{Dd}(s' \mid B_0) + d_{\mathrm{rec}}\bigr) \\
  &= \mathrm{Dd}(s \mid B_0) + d_{\mathrm{rec}} .
\end{aligned}
\end{equation}
which is \eqref{eq:depth-shift}. Taking \(s=q\) yields \eqref{eq:depth-degrade}.

Finally, since \(B_\cap\subseteq \tilde B_0\), Lemma~\ref{lem:premise-monotonicity} gives
\(\mathrm{Dd}(q\mid \tilde B_0)\le \mathrm{Dd}(q\mid B_\cap)\), and thus \(\tilde n\le n+d_{\mathrm{rec}}\).

(iii) This is Proposition~\ref{prop:noisy-bound}(ii).
\end{proof}

\begin{remark}[No two-sided bound on \(|\tilde n-n|\) without restricting pollution]
Without structural restrictions on \(B_0^+\), pollution can arbitrarily shorten derivations (e.g., by adding \(q\) itself),
so \(|\tilde n-n|\) cannot be controlled solely by \(\mathcal N\) or loss-only parameters. This is why
Lemma~\ref{lem:noise-perturbation} gives robust one-sided ``no-worse-than'' bounds for loss and monotonicity for pollution.
\end{remark}

\begin{corollary}[Loss-only noise tolerance (depth stability)]
\label{cor:noise-tolerance}
Assume loss-only noise \(B_0^+=\varnothing\), so \(\tilde B_0=B_0\setminus B_0^-\).
Assume \(q\) remains derivable and \(d_{\mathrm{rec}}<\infty\) satisfies \(d_{\mathrm{rec}}=o(n)\) where \(n=\mathrm{Dd}(q\mid B_0)\).
Then
\[
  H_{\mathrm{derive}}(q\mid \tilde B_0)=H_{\mathrm{derive}}(q\mid B_0)\cdot (1+o(1)).
\]
\end{corollary}

\begin{proof}
Loss-only implies \(\mathrm{Dd}(q\mid \tilde B_0)\ge \mathrm{Dd}(q\mid B_0)=n\) (removing premises cannot decrease depth).
Lemma~\ref{lem:noise-perturbation}(ii) gives \(\mathrm{Dd}(q\mid \tilde B_0)\le n+d_{\mathrm{rec}}\).
Thus \(\mathrm{Dd}(q\mid \tilde B_0)=n(1+o(1))\), and multiply by \(\ln 2\).
\end{proof}

%---------------------------------------------------------------------------------
%  V.D Noisy tradeoff and noisy optimization
%---------------------------------------------------------------------------------
\subsection{Noisy Tradeoffs and Noise-Aware Optimization}
\label{subsec:noise-strategies}

We now extend the single-query tradeoff and system-wide allocation of Section~\ref{sec:tradeoff-ideal}
to noisy bases. As in Section~\ref{sec:tradeoff-ideal}, caching is modeled as premise-base augmentation.

%-----------------------------------------
%  V.D.1 Noisy single-query tradeoff (baseline-parametric)
%-----------------------------------------
\subsubsection{Noisy storage--computation tradeoff}
\label{subsubsec:noisy-tradeoff}

\begin{proposition}[Noisy storage--computation tradeoff (baseline-parametric)]
\label{prop:noisy-tradeoff}
Fix a baseline base \(B_0\) and its noisy version \(\tilde B_0=(B_0\setminus B_0^-)\cup B_0^+\).
Let \(q\in \mathrm{Cn}(\tilde B_0)\cap \mathrm{Cn}(B_0)\) be a query that remains derivable from both bases, let \(\rho>0\) be the
storage-to-computation cost ratio, and let \(f_q\ge 1\) be the access frequency.

Let \(\tilde m:=|\tilde B_0|\) and \(\tilde n:=\mathrm{Dd}(q\mid \tilde B_0)\).
Assume \((q,\tilde B_0)\) lies in the information-rich regime and is generic/incompressible at its depth scale so that
Theorem~\ref{thm:derivation-depth-info-metric} applies to \((q,\tilde B_0)\).

Define the noisy cache-length benchmark
\[
  \tilde \ell_q := K \left(q\mid \langle \tilde B_0\rangle\right)+O(1),
\]
and define noisy amortized costs (cf.\ Definition~\ref{def:amortized-cost}) by
\begin{align*}
  &\tilde{\mathrm{Cost}}_{\mathrm{cache}}(q;\tilde B_0)
  :=
  \frac{\rho\cdot \tilde \ell_q}{f_q}
  +\frac{\mathrm{Dd}(q\mid \tilde B_0)}{f_q}
  + c_{\mathrm{hit}},
  \\
  &\tilde{\mathrm{Cost}}_{\mathrm{derive}}(q;\tilde B_0)
  :=
  \mathrm{Dd}(q\mid \tilde B_0).
\end{align*}
Let
\[
\tilde{\mathrm{Cost}}^*(q;\tilde B_0)
:=
\min \left\{
\tilde{\mathrm{Cost}}_{\mathrm{cache}}(q;\tilde B_0),
\tilde{\mathrm{Cost}}_{\mathrm{derive}}(q;\tilde B_0)\
\right\}.
\]
Then:
\begin{enumerate}[label=\textup{(\roman*)}]
  \item \emph{Noisy tradeoff characterization:}
\begin{equation}
\begin{aligned}
\tilde{\mathrm{Cost}}^*(q;\tilde B_0)
  = \min\Biggl\{
      \frac{\rho\,\tilde\ell_q}{f_q}
      + O \left(\frac{\tilde n}{f_q}\right)
      + O(1), 
    \Theta \left(\frac{\tilde\ell_q}{\log(\tilde m+\tilde n)}\right)
    \Biggr\}.
\end{aligned}
\label{eq:noisy-tradeoff-char}
\end{equation}
  \item \emph{Noisy critical frequency:}
  \begin{equation}\label{eq:noisy-critical-freq}
    \tilde f_c=\Theta\bigl(\rho\cdot \log(\tilde m+\tilde n)\bigr).
  \end{equation}

  \item \emph{Robust comparison to the baseline optimum (no-worse-than bound):}
  let \(\mathrm{Cost}^*(q;B_0)\) be the ideal optimal amortized cost under baseline \(B_0\) as in Theorem~\ref{thm:freq_tradeoff}.
  Under Assumption~\ref{assump:bounded-noise},
  \begin{align}\label{eq:noisy-no-worse}
    \tilde{\mathrm{Cost}}^*(q;\tilde B_0) 
    \le
    \mathrm{Cost}^*(q;B_0)
    +  O \left( \frac{\rho \cdot \mathcal N\log(m+\mathcal N)}{f_q}\right)
    + O(d_{\mathrm{rec}}),
  \end{align}
  where \(m:=|B_0|\), \(\mathcal N:=|B_0^-|+|B_0^+|\), and \(d_{\mathrm{rec}}\) is from Definition~\ref{def:rec-depth-preserved}.
\end{enumerate}
\end{proposition}

\begin{proof}
(i) By \(\tilde \ell_q=K(q\mid\langle \tilde B_0\rangle)+O(1)\) and the definition of \(\tilde{\mathrm{Cost}}_{\mathrm{cache}}\),
\[
  \tilde{\mathrm{Cost}}_{\mathrm{cache}}(q;\tilde B_0)
  =
  \frac{\rho\cdot \tilde \ell_q}{f_q}
  + O \left(\frac{\tilde n}{f_q}\right)
  + O(1).
\]
Moreover, Theorem~\ref{thm:derivation-depth-info-metric} applied to \((q,\tilde B_0)\) yields
\(
\tilde n=\Theta \bigl(\tilde \ell_q/\log(\tilde m+\tilde n)\bigr)
\),
so \eqref{eq:noisy-tradeoff-char} follows by taking the minimum with \(\tilde{\mathrm{Cost}}_{\mathrm{derive}}(q;\tilde B_0)=\tilde n\).

(ii) Compare the leading terms \(\rho\tilde \ell_q/f_q\) and \(\Theta(\tilde \ell_q/\log(\tilde m+\tilde n))\) to obtain
\(\tilde f_c=\Theta(\rho\log(\tilde m+\tilde n))\).

(iii) By Lemma~\ref{lem:noise-perturbation}(iii),
\(
\tilde \ell_q
=
\ell_q(B_0) + O(\mathcal N\log(m+\mathcal N))
\)
up to additive constants, and by Lemma~\ref{lem:noise-perturbation}(ii),
\(
\mathrm{Dd}(q\mid \tilde B_0)\le \mathrm{Dd}(q\mid B_0)+d_{\mathrm{rec}}.
\)
Therefore, the noisy caching cost is at most the baseline caching cost plus
\(O(\rho\,\mathcal N\log(m+\mathcal N)/f_q)\) (and lower-order population differences), while the noisy on-demand cost is at most
the baseline on-demand cost plus \(O(d_{\mathrm{rec}})\).
Taking minima over the two strategies and using the elementary inequality
\[
\min\{a+\varepsilon_1,\; b+\varepsilon_2\}
\le
\min\{a,b\}+\max\{\varepsilon_1,\varepsilon_2\},
\]
we obtain
\[
\tilde{\mathrm{Cost}}^{\,*}(q;\tilde{B}_0)
\le
\mathrm{Cost}^{*}(q;B_0)
+
O \left(\frac{\rho\,N\log(m+N)}{f_q}\right)
+
O(d_{\mathrm{rec}}),
\]
which is exactly \eqref{eq:noisy-no-worse}..
\end{proof}

%-----------------------------------------
%  V.D.2 Noisy system-wide allocation and submodularity
%-----------------------------------------
\subsubsection{Noisy system-wide allocation and submodularity}
\label{subsubsec:noisy-submodular}

Fix an operational candidate set \(\mathcal U\subseteq \mathrm{Cn}(\mathrm{Atom}(S_O))\setminus S_O\) and a query set
\(Q=\{q_1,\ldots,q_K\}\) with distribution \(P_Q\) as in Section~\ref{subsec:system-allocation}.
In noise, the baseline premise base becomes \(\tilde B_0\) (typically \(\tilde B_0=\tilde S_O\) operationally,
or \(\tilde B_0=\tilde A\) intrinsically), and caching enlarges \(\tilde B_0\) to \(\tilde B_0\cup S\).

\begin{definition}[Noisy expected derivation cost and reduction]
\label{def:noisy-cost-reduction}
Given a noisy baseline \(\tilde B_0\), define for \(S\subseteq \mathcal U\):
\begin{align*}
  &\tilde{\bar n}(S)
  :=
  \sum_{i=1}^{K} p_i\,\mathrm{Dd}(q_i\mid \tilde B_0\cup S),
  \\
  &\tilde{\Delta}(S):=\tilde{\bar n}(\varnothing)-\tilde{\bar n}(S).
\end{align*}
\end{definition}

\begin{proposition}[Submodularity under noisy diminishing returns]
\label{prop:submod-preservation}
Assume diminishing returns on the noisy base: for all \(A_1\subseteq A_2\subseteq \mathcal U\) and
\(u\in \mathcal U\setminus A_2\),
\begin{align*}
\mathrm{Dd}(q_i\mid \tilde B_0\cup A_1)-\mathrm{Dd}(q_i\mid \tilde B_0\cup A_1\cup\{u\}) 
\ge
\mathrm{Dd}(q_i\mid \tilde B_0\cup A_2)-\mathrm{Dd}(q_i\mid \tilde B_0\cup A_2\cup\{u\})
\end{align*}
for each \(i\).
Then \(\tilde\Delta\) is normalized, monotone, and submodular on \(2^{\mathcal U}\).
\end{proposition}

\begin{proof}
Identical to Theorem~\ref{thm:submodularity-deriv}, with \(\bar n\) replaced by \(\tilde{\bar n}\).
\end{proof}

%-----------------------------------------
%  V.D.3 Pollution exposure penalty (robust objective)
%-----------------------------------------
\subsubsection{Pollution exposure and a robust objective}
\label{subsubsec:pollution-robust}

In pollution-dominated regimes it can be desirable, from a robustness perspective, to avoid cache items whose
usefulness relies heavily on spurious premises $B_0^+$, even if they reduce derivation depth relative to the noisy base.
We encode this preference as a modular penalty added to the noisy depth-reduction objective.

\begin{definition}[Pollution exposure score]
\label{def:pollution-exposure}
Fix the noisy base \(\tilde B_0=(B_0\setminus B_0^-)\cup B_0^+\).
For a candidate cache item \(u\in \mathcal U\), define its pollution exposure as
\[
  \pi_{\mathrm{poll}}(u)
  :=
  \frac{\bigl|\mathrm{Ess}^+(u\mid \tilde B_0)\cap B_0^+\bigr|}
       {\max\{1,\ |\mathrm{Ess}^+(u\mid \tilde B_0)|\}}
  \in [0,1],
\]
where \(\mathrm{Ess}^+(\cdot\mid \tilde B_0)\) is defined analogously to Definition~\ref{def:essential} with baseline \(\tilde B_0\).
\end{definition}

\begin{definition}[Robustified noisy objective]
\label{def:robust-objective}
For penalty weight \(\lambda\ge 0\), define for \(S\subseteq \mathcal U\)
\[
  \tilde{\Delta}_\lambda(S)
  :=
  \tilde{\Delta}(S)
  -
  \lambda\sum_{u\in S}\sigma(u)\,\pi_{\mathrm{poll}}(u).
\]
\end{definition}

\begin{proposition}[Submodularity is preserved under pollution penalties]
\label{prop:penalty-submodular}
If \(\tilde{\Delta}\) is normalized and submodular, then \(\tilde{\Delta}_\lambda\) is normalized and submodular for every \(\lambda\ge 0\).
(Monotonicity may be lost.)
\end{proposition}

\begin{proof}
The penalty term is modular in \(S\), hence submodular; subtracting a modular function preserves submodularity.
\end{proof}

\begin{remark}[Monotonicity and algorithms]
\label{rem:penalty-monotonicity}
Because \(\tilde{\Delta}_\lambda\) may be non-monotone, Algorithm~\ref{alg:greedy} does not retain its \((1-1/e)\) guarantee
when applied directly. One may instead use approximation algorithms for \emph{non-monotone} submodular maximization;
see, e.g.,~\cite{buchbinder2014submodular}.
\end{remark}

%-----------------------------------------
%  V.D.4 Loss-dominated regimes: compensation-then-optimize
%-----------------------------------------
\subsubsection{Loss-dominated regimes: compensation-then-optimize}
\label{subsubsec:loss-compensation}

When loss dominates, a small number of missing premises can dramatically increase depth.
A natural response is to first allocate storage to \emph{restore} (cache) critical lost premises, and then
optimize remaining storage for depth reduction.

\begin{definition}[Critical and reconstructible lost premises (baseline-parametric)]
\label{def:critical-atoms}
Fix a baseline \(B_0\), its noisy version \(\tilde B_0\), and a finite query set \(Q\).
Define:
\begin{enumerate}[label=\textup{(\roman*)}]
  \item \emph{Critical lost premises:}
  \[
    \mathcal{B}_{\mathrm{crit}}
    :=
    B_0^- \cap \bigcup_{q\in Q}\mathrm{Ess}^+(q\mid B_0).
  \]
  \item \emph{Reconstructible critical premises:}
  \[
    \mathcal{B}_{\mathrm{rec}}
    :=
    \{b\in \mathcal{B}_{\mathrm{crit}}:\ b\in \mathrm{Cn}(\tilde B_0)\}.
  \]
  \item \emph{Irrecoverable critical premises:}
  \[
    \mathcal{B}_{\mathrm{irr}}:=\mathcal{B}_{\mathrm{crit}}\setminus \mathcal{B}_{\mathrm{rec}}.
  \]
\end{enumerate}
\end{definition}

\begin{remark}[Caching lost premises during compensation]
Although Section~\ref{subsec:system-allocation} typically uses \(\mathcal U\subseteq \mathrm{Cn}(\mathrm{\Atom}(S_O))\setminus S_O\),
in the compensation phase we allow caching elements \(b\in \mathcal{B}_{\mathrm{rec}}\subseteq B_0\) as well:
adding \(b\) to the available premise base simply makes it depth-\(0\) again under Definition~\ref{def:derivation-depth}.
Formally, define the extended ground set
\[
  \mathcal U_{\mathrm{ext}}:=\mathcal U\cup \mathcal{B}_{\mathrm{rec}}.
\]
\end{remark}

\begin{definition}[Depth-threshold service constraint (\(\mathrm{SLA}_h\), noisy baseline)]
\label{def:sla-depth}
Fix a depth threshold \(h\in\mathbb N\) and a storage budget \(\mathcal{B}>0\) (bits).
An allocation \(S\subseteq \mathcal U_{\mathrm{ext}}\) is \(\mathrm{SLA}_h\)-feasible under noisy baseline \(\tilde B_0\) if
\begin{equation}\label{eq:sla-feasible}
  \sigma(S)\le \mathcal{B}
  \quad\text{and}\quad
  \max_{1\le i\le K}\mathrm{Dd}(q_i\mid \tilde B_0\cup S)\le h.
\end{equation}
Denote the feasible family by \(\mathcal F_h(\mathcal{B})\).
\end{definition}

\begin{assumption}[SLA repair by compensation]
\label{assump:comp-feasible}
Let \(S_{\mathrm{comp}}:=\mathcal{B}_{\mathrm{rec}}\) and \(\mathcal{B}_{\mathrm{comp}}:=\sigma(S_{\mathrm{comp}})\).
Assume \(\mathcal{B}\ge \mathcal{B}_{\mathrm{comp}}\) and the compensation cache alone meets the depth SLA:
\[
  \max_{1\le i\le K}\mathrm{Dd}(q_i\mid \tilde B_0\cup S_{\mathrm{comp}})\le h.
\]
\end{assumption}

\begin{assumption}[Mandatory compensation under \(\mathrm{SLA}_h\)]
\label{assump:mandatory-comp}
Assume that every \(\mathrm{SLA}_h\)-feasible allocation must contain the compensation set:
\[
  \forall S\in \mathcal F_h(\mathcal{B}),\qquad S_{\mathrm{comp}}\subseteq S.
\]
\end{assumption}

\begin{lemma}[Compensation cache size bound]
\label{lem:comp-cost}
Assume \(\mathcal{B}_{\mathrm{irr}}=\varnothing\).
Then the total description size needed to store \(\mathcal{B}_{\mathrm{rec}}\) as cached premises satisfies

\begin{align*}
\sigma(S_{\mathrm{comp}})
  \le |B_{\mathrm{rec}}| \cdot c \log(m+N) 
  \le cN \log(m+N) 
  = \mathcal{O}(N\log(m+N)).
\end{align*}
where \(m:=|B_0|\) and \(\mathcal N:=|B_0^-|+|B_0^+|\).
\end{lemma}

\begin{proof}
Each element of \(B_0\) can be named by an index in a pool of size \(O(m+\mathcal N)\), costing \(O(\log(m+\mathcal N))\) bits.
Moreover, under Assumption~\ref{assump:bounded-noise}(i) and the fixed representation conventions
(Axiom~\ref{ax:state-representation}), such an index effectively specifies the corresponding encoded premise,
so it can be stored and re-introduced as an available premise when cached.
Since \(|\mathcal{B}_{\mathrm{rec}}|\le |B_0^-|\le \mathcal N\), the bound follows.
\end{proof}

\begin{proposition}[Two-phase strategy: \((1-1/e)\)-guarantee under \(SLA_h\)]
\label{prop:two-phase-opt}
Assume \(B_{\mathrm{irr}}=\varnothing\) and \(\sigma(S_{\mathrm{comp}})\le B\). Let
\[
B' := B-\sigma(S_{\mathrm{comp}})
\]
be the residual budget after compensation.

Define the post-compensation objective on the free extension set \(U\) by
\[
\widetilde{\Delta}_{\mathrm{add}}(T)
:=
\widetilde{\Delta}(S_{\mathrm{comp}}\cup T)
-
\widetilde{\Delta}(S_{\mathrm{comp}}),
\qquad T\subseteq U.
\]
Assume that \(\widetilde{\Delta}_{\mathrm{add}}\) is normalized, monotone, and submodular on \(2^U\). Let
\[
T_h^*
\in
\arg\max\{\widetilde{\Delta}_{\mathrm{add}}(T): \sigma(T)\le B'\}.
\]
Let \(T_{\mathrm{greedy}}\) be the set returned by Algorithm~\ref{alg:greedy} when run on ground set \(U\), with objective \(\widetilde{\Delta}_{\mathrm{add}}\) and budget \(B'\). Define
\[
S_{2\mathrm{ph}} := S_{\mathrm{comp}}\cup T_{\mathrm{greedy}},
\qquad
S_h^* := S_{\mathrm{comp}}\cup T_h^*.
\]
Assume moreover that every \(T\subseteq U\) with \(\sigma(T)\le B'\) yields a feasible completion \(S_{\mathrm{comp}}\cup T\in \mathcal F_h(B)\). Then
\[
S_{2\mathrm{ph}},\, S_h^* \in \mathcal F_h(B),
\]
and
\[
\widetilde{\Delta}_{\mathrm{add}}(T_{\mathrm{greedy}})
\ge
\left(1-\frac{1}{e}\right)
\widetilde{\Delta}_{\mathrm{add}}(T_h^*).
\]
Equivalently,
\[
\widetilde{\Delta}(S_{2\mathrm{ph}})
-
\widetilde{\Delta}(S_{\mathrm{comp}})
\ge
\left(1-\frac{1}{e}\right)
\Bigl(
\widetilde{\Delta}(S_h^*)
-
\widetilde{\Delta}(S_{\mathrm{comp}})
\Bigr).
\]
\end{proposition}

\begin{proof}

Consider the set function
\[
\Delta'(T) := \widetilde{\Delta}_{\mathrm{add}}(T), \qquad T\subseteq U.
\]
By assumption, \(\Delta'\) is normalized, monotone, and submodular on \(2^U\). Applying Theorem~\ref{thm:greedy-guarantee}  to the knapsack instance on ground set \(U\) with objective \(\Delta'\), cost function \(\sigma\), and budget \(B'\), we obtain
\[
\Delta'(T_{\mathrm{greedy}})
\ge
\left(1-\frac{1}{e}\right)\Delta'(T_h^*),
\]
where
\[
T_h^*
\in
\arg\max\{\Delta'(T): \sigma(T)\le B'\}.
\]
Since \(\Delta'=\widetilde{\Delta}_{\mathrm{add}}\), this gives
\[
\widetilde{\Delta}_{\mathrm{add}}(T_{\mathrm{greedy}})
\ge
\left(1-\frac{1}{e}\right)
\widetilde{\Delta}_{\mathrm{add}}(T_h^*).
\]

By the feasibility assumption, every \(T\subseteq U\) with \(\sigma(T)\le B'\) yields a feasible full solution \(S_{\mathrm{comp}}\cup T\in\mathcal F_h(B)\). Since both \(T_{\mathrm{greedy}}\) and \(T_h^*\) satisfy the residual budget constraint, it follows that
\[
\begin{aligned}
S_{2\mathrm{ph}} &:= S_{\mathrm{comp}} \cup T_{\mathrm{greedy}} \in \mathcal{F}_h(B),\\
S_h^* &:= S_{\mathrm{comp}} \cup T_h^* \in \mathcal{F}_h(B).
\end{aligned}
\]

Finally, by definition of \(\widetilde{\Delta}_{\mathrm{add}}\),
\[
\widetilde{\Delta}_{\mathrm{add}}(T)
=
\widetilde{\Delta}(S_{\mathrm{comp}}\cup T)
-
\widetilde{\Delta}(S_{\mathrm{comp}})
\]
for every \(T\subseteq U\). Substituting \(T=T_{\mathrm{greedy}}\) and \(T=T_h^*\) yields
\[
\widetilde{\Delta}(S_{2\mathrm{ph}})
-
\widetilde{\Delta}(S_{\mathrm{comp}})
\ge
\left(1-\frac{1}{e}\right)
\Bigl(
\widetilde{\Delta}(S_h^*)
-
\widetilde{\Delta}(S_{\mathrm{comp}})
\Bigr).
\]

This proves the claim.
    
\end{proof}
%=============================================================
%  VI.  CONCLUSION
%=============================================================
\section{Conclusion}
\label{sec:conclusion}

This paper develops an axiomatic and information-theoretic framework for analyzing
\emph{inference cost} and \emph{storage--computation tradeoffs} in knowledge-based query answering systems.
Our starting point is a two-domain view of information \((S_O,S_C)\) together with a computable realization mechanism
(Section~\ref{sec:model}), and a canonical separation between \emph{intrinsic} semantic content and
\emph{operational} shortcut effects via the irredundant core \(\mathrm{Atom}(S_O)\)
(Section~\ref{subsec:atomic-derivation}).

\paragraph{Main results}
We introduce \emph{base-relative derivation depth} \(\mathrm{Dd}(\cdot\mid B)\) as a computable, premise-parametric cost functional.
By encoding derivation traces (Lemma~\ref{lem:encoding}) and combining richness with incompressibility
(Definition~\ref{def:richness} and Lemma~\ref{lem:incompressibility}), we show that, for generic information-rich queries,
\(\mathrm{Dd}(\cdot\mid B)\) is tightly coupled to conditional algorithmic information up to an unavoidable logarithmic addressing factor
(Theorem~\ref{thm:derivation-depth-info-metric}).
For conjunctive workloads, the compositional depth \(\mathrm{Dd}_\Sigma\) provides a step-accurate instantiation
(Definition~\ref{def:additive-depth} and Corollary~\ref{cor:bcq-instantiation}).
Building on the depth--information correspondence, we derive a frequency-weighted storage--computation tradeoff with a break-even
frequency scale \(f_c=\Theta(\rho\log(m+d))\) in the generic regime (Theorem~\ref{thm:freq_tradeoff} and
Corollary~\ref{cor:fc-window}), and formulate system-wide caching as a budgeted optimization problem that admits standard
approximation guarantees under submodularity assumptions (Section~\ref{subsec:system-allocation}).
Finally, we extend the framework to noisy premise bases (Section~\ref{sec:noisy-extensions}) and derive baseline-parametric,
noise-aware tradeoffs and allocation objectives.
We complement the theoretical analysis with a small controlled synthetic validation designed to isolate the effects of premise loss, memory pollution, and cache-budget variation. The full setup and all plots are deferred to Appendix D.

\paragraph{Limitations}
Several limitations reflect the gap between idealized information-theoretic models and deployed systems.
First, the depth-as-information theorems rely on explicit structural assumptions, including alignment between the predecessor
operator and single-step inference (Assumption~\ref{assump:alignment}), serializability relating \(\mathrm{Dd}\) to shortest trace length
\(\mathrm{N}\) (Assumption~\ref{assump:serializable}), and a richness condition guaranteeing sufficiently many distinct queries at each cost
scale (Definition~\ref{def:richness}).
Characterizing when these assumptions hold for concrete logics, proof systems, and workloads remains an open task.
Second, our semantic-core extraction \(\mathrm{Atom}(S_O)\) relies on a decidable redundancy test
(Assumption~\ref{ass:decidability}) and assumes finite, effectively listable knowledge bases
(Assumption~\ref{assump:finite-so}).
Third, the noise model in Section~\ref{sec:noisy-extensions} is premise-set based (loss/pollution) and does not yet capture richer
corruption modes (e.g., structured drift in rules or semantics-changing perturbations).
Finally, while our results use conditional Kolmogorov complexity as an information-theoretic benchmark, it is not computable; any
operational deployment must instantiate the description-length terms via concrete encodings or computable proxies.

\paragraph{Future directions}
The framework suggests several directions for further study:
(i) develop testable sufficient conditions (or diagnostics) for richness and serializability in common proof fragments, and quantify
the constant factors relating \(\mathrm{Dd}\), \(N\), and execution time;
(ii) design scalable methods to approximate semantic locality (e.g., \(\mathrm{Ess}^+(q\mid S_O)\)) and exploit locality-refined parameters
in caching decisions;
(iii) extend the optimization layer to nonstationary query distributions and online adaptation;
(iv) move beyond set perturbations to richer noise and robustness models, including adversarial pollution; and
(v) integrate the proposed metrics into concrete reasoning/caching engines to study end-to-end behavior and system interactions.

In summary, derivation depth provides a principled bridge between inference cost and information, supporting both individual-query
tradeoffs and system-wide storage allocation.

%%%%%%%%%%%%%%%%%%%%%%%%%%%%%%%%%%%%%%%%%%%%%%%%%%%%%%%%%%%%%%%%%%%%%%%%%%%%
%%%%%% Acknowledgements
%%%%%%%%%%%%%%%%%%%%%%%%%%%%%%%%%%%%%%%%%%%%%%%%%%%%%%%%%%%%%%%%%%%%%%%%%%%%
\section*{CRediT authorship contribution statement}
S.Q.: Conceptualization, Methodology, Writing -- original draft, Writing -- review \& editing.
J.X.: Conceptualization, Methodology, Writing -- review \& editing, Project administration.
R.W.: Conceptualization, Methodology, Writing -- review \& editing, Project administration.
\section*{Data availability}
No data was used for the research described in the article.

\section*{Declaration of competing interest}
The authors declare that they have no known competing financial interests or personal relationships that could have appeared to 
influence the work reported in this paper. 

\section*{Acknowledgments}
The authors thank all colleagues and students who contributed useful discussions related to this work.

\appendix

\section{Abbreviations and Terminology}

\begin{description}
  \item[\textbf{FO(LFP)}] First-order logic with least fixed-point operators.
  \item[\textbf{BCQ}] Boolean conjunctive query.
  \item[\textbf{MDL}] Minimum Description Length.
  \item[\textbf{OIT}] Objective Information Theory.
  \item[\textbf{EPS}] Efficient proof simulation assumption in Theorem~\ref{thm:bennett-equivalence}.
  \item[\textbf{NPS}] No proof shortcut assumption in Theorem~\ref{thm:bennett-equivalence}.
  \item[\textbf{SLA}] Service-level agreement; here used for depth-threshold constraints in Section~\ref{sec:noisy-extensions}.
  \item[\textbf{Intrinsic base}] The canonical irredundant semantic core \(A := \mathrm{Atom}(S_O)\).
  \item[\textbf{Operational base}] The full stored knowledge base \(S_O\), including shortcuts.
  \item[\textbf{Noisy base}] A perturbed available premise base of the form \(\tilde{B}_0 = (B_0 \setminus B_0^{-}) \cup B_0^{+}\).
  \item[\textbf{Semantic core}] The canonical irredundant generating set \(\mathrm{Atom}(S_O)\).
  \item[\textbf{Stored shortcuts}] Formulas in \(J := S_O \setminus \mathrm{Atom}(S_O)\) that reduce operational depth without changing closure.
  \item[\textbf{Compensation set}] The set \(S_{\mathrm{comp}} := B_{\mathrm{rec}}\) used in Section~\ref{subsubsec:loss-compensation} to restore critical lost premises before optimization.
  \item[\textbf{Pollution exposure}] The quantity \(\pi_{\mathrm{poll}}(u)\), measuring how strongly a candidate relies on spurious premises.
\end{description}

\section{Notation Summary}

\begingroup
\setlength{\LTleft}{0pt}
\setlength{\LTright}{0pt}
\renewcommand{\arraystretch}{1.08}
\small

\begin{longtable}{@{}>{\raggedright\arraybackslash}p{0.24\textwidth}>{\raggedright\arraybackslash}p{0.72\textwidth}@{}}
\toprule
\textbf{Symbol} & \textbf{Meaning} \\
\midrule
\endfirsthead

\toprule
\textbf{Symbol} & \textbf{Meaning} \\
\midrule
\endhead

\midrule
\multicolumn{2}{r}{\emph{Continued on next page}} \\
\endfoot

\bottomrule
\endlastfoot

\(\mathcal{L}\) & Fixed underlying logical language (instantiated as FO(LFP) in this paper). \\
\(\mathsf{PS}\) & Fixed effective proof system for \(L\). \\
\(S_O, S_C\) & Semantic/ontological state set and carrier/physical state set. \\
\(R_{\mathcal E}\) & Enabling (realization) relation between \(S_O\) and \(S_C\). \\
\({\mathcal E}\) & Induced enabling map from \(S_O\) to subsets of \(S_C\). \\
\(\mathrm{Cn}(\Gamma)\) & Deductive closure of a premise set \(\Gamma\). \\
\(P_O(s)\) & Computable predecessor set of state/formula \(s\). \\
\(\operatorname{Atom}(S_O)\) & Canonical irredundant semantic core of \(S_O\). \\
\(A\) & Intrinsic premise base, \(A = \operatorname{Atom}(S_O)\). \\
\(J\) & Stored shortcuts, \(J = S_O \setminus A\). \\
\(B\) & Generic finite available premise base. \\
\(B_0\) & Baseline available premise base. \\
\(\mathrm{Dd}(q \mid B)\) & Base-relative derivation depth of query \(q\) from premise base \(B\). \\
\(\mathrm{N}(q \mid B)\) & Shortest derivation-trace length of \(q\) from \(B\). \\
\(\mathrm{Dd}_{\Sigma}(q \mid B)\) & Compositional (sum-rule) depth for BCQs. \\
\(d_{\mathrm{int}}(q)\) & Intrinsic derivation depth, \(\mathrm{Dd}(q \mid A)\). \\
\(d_{\mathrm{op}}(q)\) & Operational derivation depth, \(\mathrm{Dd}(q \mid S_O)\). \\
\(H_{\mathrm{derive}}(q \mid B)\) & Derivation entropy, \(\mathrm{Dd}(q \mid B)\ln 2\). \\
\(K(q \mid \langle B \rangle)\) & Conditional prefix-free Kolmogorov complexity of \(q\) given \(\langle B \rangle\). \\
\(H_K(q \mid B)\) & Algorithmic entropy, \(K(q \mid \langle B \rangle)\ln 2\). \\
\(\pi\) & A derivation trace. \\
\(\Pi^{*}(q)\) & Set of shortest derivation traces of \(q\). \\
\(\operatorname{Atoms}(\pi)\) & Base premises referenced by derivation trace \(\pi\). \\
\(\mathrm{Ess}(q \mid S_O)\) & Essential atom set: intersection of referenced base premises over shortest traces. \\
\(\mathrm{Ess}^{+}(q \mid S_O)\) & Extended essential atom set: union of referenced base premises over shortest traces. \\
\(m_{\mathrm{eff}}(q)\) & Effective local base size, \(|\mathrm{Ess}^{+}(q \mid S_O)|\). \\
\(d_{\mathrm{eff}}(q)\) & Effective local depth. \\
\(L_{\mathrm{eff}}(q)\) & Effective logarithmic factor, \(\log(m_{\mathrm{eff}}(q) + d_{\mathrm{eff}}(q))\). \\
\(\lambda(q)\) & Semantic locality coefficient, \(m_{\mathrm{eff}}(q)/|A|\). \\
\(\ell_q(B_0)\) & Stored code length required to cache query \(q\) relative to baseline \(B_0\). \\
\(\rho\) & Relative unit price of storage versus computation, \(\rho = \alpha/\beta\). \\
\(f_q\) & Access frequency of query \(q\). \\
\(f_c\) & Break-even query frequency for caching. \\
\(Q = \{q_1, \ldots, q_M\}\) & Finite query set used in workload-level optimization. \\
\(p_i\) & Query probability \(P_Q(q_i)\). \\
\(\bar{n}(S)\) & Expected operational derivation depth under allocation \(S\). \\
\(\Delta(S)\) & Expected reduction in derivation depth under allocation \(S\). \\
\(\Delta_u(S)\) & Marginal gain obtained by adding candidate \(u\) to \(S\). \\
\(U\) & Candidate cache ground set. \\
\(\sigma(u)\) & Storage cost of candidate \(u\). \\
\(\mathcal{B}\) & Storage budget. \\
\(d_{\mathrm{sem}}(q_i, q_j \mid S_O)\) & Semantic distance between queries \(q_i\) and \(q_j\). \\
\(\delta_{\mathrm{clust}}\) & Cohesion threshold used in semantic clustering. \\
\(\mathcal{C} = \{Q_1, \ldots, Q_r\}\) & Query clustering. \\
\(A_k^{\mathrm{core}}\) & Cluster core atom set for cluster \(Q_k\). \\
\(A_k^{\mathrm{ext}}\) & Cluster extended atom set for cluster \(Q_k\). \\
\(A_k^{\mathrm{supp}}\) & Cluster support-only atom set, \(A_k^{\mathrm{ext}} \setminus A_k^{\mathrm{core}}\). \\
\(\kappa_k\) & Cluster centrality/cohesion score. \\
\(U_k^{\mathrm{core}}\) & Core-restricted candidate set for cluster \(Q_k\). \\
\(U_{\mathrm{reduced}}\) & Reduced candidate set obtained from cluster-level construction. \\
\(w_k\) & Total query probability mass of cluster \(Q_k\). \\
\(\bar{d}_k\) & Average baseline operational depth within cluster \(Q_k\). \\
\(\tilde B_0\) & Noisy baseline premise base. \\
\(B_0^{-}, B_0^{+}\) & Lost and spurious premises in the noisy baseline. \\
\(\mathcal{N}\) & Total noise magnitude, \(\mathcal{N} = |B_0^{-}| + |B_0^{+}|\). \\
\(d_{\mathrm{rec}}\) & Reconstruction depth of lost premises from preserved premises. \\
\(\widetilde{\ell}_q\) & Noisy cache-length benchmark for query \(q\). \\
\(\widetilde{\bar{n}}(S)\) & Expected derivation depth under noisy baseline and allocation \(S\). \\
\(\widetilde{\Delta}(S)\) & Expected noisy depth reduction under allocation \(S\). \\
\(\pi_{\mathrm{poll}}(u)\) & Pollution exposure score of candidate cache item \(u\). \\
\(\widetilde{\Delta}_{\lambda}(S)\) & Robustified noisy objective with pollution penalty weight \(\lambda\). \\
\(B_{\mathrm{crit}}\) & Critical lost premises. \\
\(B_{\mathrm{rec}}\) & Reconstructible critical lost premises. \\
\(B_{\mathrm{irr}}\) & Irrecoverable critical lost premises. \\
\(S_{\mathrm{comp}}\) & Compensation set used in the loss-dominated regime. \\
\(U_{\mathrm{ext}}\) & Extended candidate set, typically \(U \cup B_{\mathrm{rec}}\). \\
\(\mathcal{F}_h(\mathcal{B})\) & Family of allocations satisfying the depth-threshold SLA \(h\) under budget \(\mathcal{B}\). \\
\(T_{\mathrm{greedy}}\) & Greedy post-compensation extension set. \\
\(T_h^{*}\) & Optimal post-compensation extension set under SLA level \(h\). \\
\(S_{\mathrm{2ph}}\) & Two-phase allocation returned by compensation followed by optimization. \\
\(S_h^{*}\) & Optimal feasible two-phase allocation under SLA level \(h\). \\
\(\widetilde{\Delta}_{\mathrm{add}}(T)\) & Post-compensation incremental objective, \(\widetilde{\Delta}(S_{\mathrm{comp}} \cup T) - \widetilde{\Delta}(S_{\mathrm{comp}})\). \\

\end{longtable}
\endgroup

\section{Semantic Clustering for Practical Deployment}
\label{subsec:clustering}

This appendix is explicitly heuristic and deployment-oriented. Its purpose is not to strengthen the approximation
guarantees of Section~\ref{subsec:system-allocation}, but to reduce the effective candidate space in large systems where the ground set $U$ is too
large to handle directly and evaluating marginal gains $\Delta_u(S)$ is expensive. Throughout, semantic clustering is
used only for candidate generation; the final storage decision is still evaluated against the operational objective
$\Delta(\cdot)$ and implemented via Algorithm~\ref{alg:greedy} on the reduced ground set.

\paragraph{Intrinsic structure, operational objective}
We cluster queries using intrinsic dependence on core premises, via $\mathrm{Ess}^+(q \mid S_O)$ from Definition~\ref{def:essential}, which
is defined using shortest traces from $A = \mathrm{Atom}(S_O)$. However, the optimization objective remains operational:
it reduces depths $Dd(q \mid S_O \cup S)$. This distinction is deliberate. Clustering serves only as a scalability
heuristic for proposing candidates, whereas the objective being optimized remains the expected reduction in actual
online derivation cost under the operational baseline.

%-----------------------------------------
%  IV.C.1 Semantic distance
%-----------------------------------------

\subsection{Semantic distance and cohesive clustering}
\label{sec:semantic-distance}

\begin{definition}[Semantic distance]
\label{def:semantic-distance}
For queries \(q_i,q_j \in \mathrm{Cn}(\mathrm{Atom}(S_O))\), let
\(A_i:=\mathrm{Ess}^+(q_i\mid S_O)\) and \(A_j:=\Ess^+(q_j\mid S_O)\).
Define the semantic distance (Jaccard distance) as
\begin{equation}\label{eq:semantic-distance}
  d_{\mathrm{sem}}(q_i,q_j\mid S_O)
  :=
  \begin{cases}
    0, & \text{if } A_i=A_j=\varnothing,\\[0.3ex]
    1, & \text{if exactly one} 
        \text{of } A_i,A_j \text{ is } 
        \text{empty},\\[0.3ex]
    1-\dfrac{|A_i\cap A_j|}{|A_i\cup A_j|}, & \text{otherwise}.
  \end{cases}
\end{equation}
\end{definition}

\begin{remark}[Approximation of \(\mathrm{Ess}^+\) in deployment]
\label{rem:essplus-approx}
Computing \(\mathrm{Ess}^+(q\mid S_O)\) exactly may be hard if shortest derivations are not unique or hard to find.
In deployment one may replace it by \(\widehat{\mathrm{Ess}}^+(q\mid S_O)\) obtained from one or several certified (near-)short traces.
All definitions below remain meaningful under this replacement by interpreting \(A_i\) and \(A_j\) as such approximations.
\end{remark}

\begin{lemma}[Metric properties]
\label{lem:metric-properties}
The function \(d_{\mathrm{sem}}(\cdot,\cdot\mid S_O)\) is a pseudometric on \(\mathrm{Cn}(\mathrm{Atom}(S_O))\).
\end{lemma}

\begin{proof}
Jaccard distance is a metric on finite sets (with the empty-set conventions above).
Composing a metric with a possibly non-injective map \(q\mapsto \mathrm{Ess}^+(q\mid S_O)\) yields a pseudometric.
\end{proof}

\begin{definition}[\(\delta_{\mathrm{clust}}\)-cohesive clustering]
\label{def:query-clustering}
Let \(Q \subseteq \mathrm{Cn}(\mathrm{Atom}(S_O))\) be finite.
A \(\delta_{\mathrm{clust}}\)-cohesive clustering is a partition \(\mathcal{C}=\{Q_1,\ldots,Q_r\}\) such that
for every cluster \(Q_k\) and all \(q_i,q_j\in Q_k\),
\[
  d_{\mathrm{sem}}(q_i,q_j\mid S_O)\le \delta_{\mathrm{clust}}.
\]
\end{definition}

\begin{remark}[Constructing cohesive clusters is heuristic]
Definition~\ref{def:query-clustering} specifies a target property; computing such a partition exactly is not required
for our guarantees (which are stated relative to the reduced ground set).
In practice, one can form approximate cohesive clusters using similarity search pipelines such as MinHash/LSH
to bucket queries by approximate Jaccard similarity~\cite{broder1997resemblance,indyk1998approximate}, followed by
lightweight post-processing (e.g., splitting large buckets).
\end{remark}

%-----------------------------------------
%  IV.C.2 Core sets and centrality
%-----------------------------------------

\subsection{Core sets: what is guaranteed, and what is not}
\label{sec:core-sets}

\begin{definition}[Cluster core and extended atom sets]
\label{def:cluster-atom-sets}
For a cluster \(Q_k\subseteq Q\), define
\begin{align}
  &A_k^{\mathrm{core}} := \bigcap_{q\in Q_k}\mathrm{Ess}^+(q\mid S_O),\\
  &A_k^{\mathrm{ext}}  := \bigcup_{q\in Q_k}\mathrm{Ess}^+(q\mid S_O),\\
  &A_k^{\mathrm{supp}} := A_k^{\mathrm{ext}}\setminus A_k^{\mathrm{core}}.
\end{align}
\end{definition}

\begin{definition}[Cluster centrality]
\label{def:cluster-centrality}
For a cluster \(Q_k\) define
\[
  \kappa_k
  :=
  \frac{|A_k^{\mathrm{core}}|}
       {\max\{1,\ \min_{q\in Q_k}|\mathrm{Ess}^+(q\mid S_O)|\}}
  \in [0,1].
\]
\end{definition}

\begin{lemma}[Two-query core coverage]
\label{lem:pair-core-coverage}
If \(Q_k=\{q_1,q_2\}\) is \(\delta_{\mathrm{clust}}\)-cohesive and both \(\mathrm{Ess}^+(q_1\mid S_O)\) and \(\mathrm{Ess}^+(q_2\mid S_O)\) are nonempty,
then for \(i\in\{1,2\}\),
\[
  \frac{|A_k^{\mathrm{core}}|}{|\mathrm{Ess}^+(q_i\mid S_O)|}\ge 1-\delta_{\mathrm{clust}},
\]
and hence \(\kappa_k\ge 1-\delta_{\mathrm{clust}}\).
\end{lemma}

\begin{proof}
Let \(A=\mathrm{Ess}^+(q_1\mid S_O)\) and \(B=\mathrm{Ess}^+(q_2\mid S_O)\).
\(\delta_{\mathrm{clust}}\)-cohesiveness gives \(|A\cap B|/|A\cup B|\ge 1-\delta_{\mathrm{clust}}\).
Since \(|A|\le |A\cup B|\), \(|A\cap B|/|A|\ge 1-\delta_{\mathrm{clust}}\). Symmetry yields the claim for \(B\).
\end{proof}

\begin{remark}[Practical implication]
\label{rem:cluster-practical}
Large clusters can be pairwise coherent yet have nearly empty cores (cf.\ the standard counterexample for Jaccard cohesion).
Thus, in deployment one should either prefer small clusters (e.g., pairs) or explicitly verify \(\kappa_k\)
before treating \(A_k^{\mathrm{core}}\) as a shared basis for candidate generation.
\end{remark}

%-----------------------------------------
%  IV.C.3 Cluster-aware candidates and allocation
%-----------------------------------------

\subsection{Cluster-aware candidate construction and allocation}
\label{sec:cluster-allocation}

\begin{definition}[Cluster core candidate set]
\label{def:cluster-core-candidates}
Let \(\mathcal{U}\subseteq \mathrm{Cn}(A)\setminus S_O\) be the global candidate ground set.
Given a clustering \(\mathcal{C}=\{Q_1,\ldots,Q_r\}\), define for each cluster \(Q_k\) a core-restricted candidate subset
\[
  \mathcal{U}_k^{\mathrm{core}}
  :=
  \{u\in \mathcal{U}:\ \mathrm{Ess}^+(u\mid S_O)\subseteq A_k^{\mathrm{core}}\}.
\]
In deployment one may replace \(\mathrm{Ess}^+\) by \(\widehat{\mathrm{Ess}}^+\) (Remark~\ref{rem:essplus-approx}).
\end{definition}

\begin{definition}[Cluster benefit proxy (operational)]
\label{def:cluster-benefit-proxy}
For each cluster \(Q_k\), define its total probability mass \(w_k:=\sum_{q\in Q_k}P_Q(q)\) and its average
\emph{operational} baseline depth
\[
  \bar d_k
  :=
  \frac{1}{w_k}\sum_{q\in Q_k} P_Q(q)\,\mathrm{Dd}(q\mid S_O).
\]
High \(w_k\) and large \(\bar d_k\) indicate high leverage; low \(\kappa_k\) warns that core-based sharing may be weak.
\end{definition}

\begin{algorithm}[t]
\caption{Cluster-Aware Storage Allocation (Deployment-Oriented)}
\label{alg:cluster-aware}
\begin{algorithmic}[1]
\Require Query set \(Q\), distribution \(P_Q\), knowledge base \(S_O\), budget \(\mathcal{B}\), threshold \(\delta_{\mathrm{clust}}\)
\Require Candidate set \(\mathcal{U}\), cost function \(\sigma(\cdot)\), oracle access to \(\Delta(\cdot)\)
\Ensure Storage allocation \(S\)
\State Compute \(\widehat{\mathrm{Ess}}^+(q\mid S_O)\) for all \(q\in Q\) \Comment{Approximate if needed; see Remark~\ref{rem:essplus-approx}}
\State Build a \(\delta_{\mathrm{clust}}\)-cohesive clustering \(\mathcal{C}=\{Q_1,\ldots,Q_r\}\) using \(d_{\mathrm{sem}}\)
\Comment{Approx.\ Jaccard via MinHash/LSH is common~\cite{broder1997resemblance,indyk1998approximate}}
\For{\(k=1,\ldots,r\)}
  \State Compute \(A_k^{\mathrm{core}}\) and \(\kappa_k\) (Definitions~\ref{def:cluster-atom-sets},~\ref{def:cluster-centrality})
  \If{\(|A_k^{\mathrm{core}}|=0\) \textbf{or} (\(|Q_k|>2\) \textbf{and} \(\kappa_k\) is below a chosen threshold)}
    \State Split \(Q_k\) (e.g., into pairs) or skip core-based candidates
    \State \textbf{continue}
  \EndIf
  \State Construct \(\mathcal{U}_k^{\mathrm{core}}\) (Definition~\ref{def:cluster-core-candidates})
\EndFor
\State \(\mathcal{U}_{\mathrm{reduced}} \gets \bigcup_{k=1}^r \mathcal{U}_k^{\mathrm{core}}\)
\State Run Algorithm~\ref{alg:greedy} on \(\mathcal{U}_{\mathrm{reduced}}\) with budget \(\mathcal{B}\) and objective \(\Delta(\cdot)\)
\State \Return \(S\)
\end{algorithmic}
\end{algorithm}

\begin{remark}[Guarantee vs.\ heuristic]
\label{rem:cluster-guarantee}
Algorithm~\ref{alg:cluster-aware} inherits the \((1-1/e)\)-approximation guarantee of
Theorem~\ref{thm:greedy-guarantee} \emph{with respect to the reduced ground set} \(\mathcal{U}_{\mathrm{reduced}}\),
provided \(\Delta\) remains monotone submodular on \(2^{\mathcal{U}_{\mathrm{reduced}}}\).
It does not guarantee the same factor against the optimum over the full \(\mathcal{U}\); the clustering layer is
a scalability heuristic that trades candidate coverage for tractability.
\end{remark}

\section{Synthetic Validation of Memory Trade-offs}

\subsection{Experimental Setup}
A reference implementation of the synthetic validation, including environment generation, cache-selection procedures, and plotting scripts, is available at:
https://github.com/caimiaomiao6/Derivation-Depth

The repository contains the code required to reproduce the synthetic derivation environment, the premise-loss and pollution experiments, and the cache-budget results reported in Appendix D. To keep the appendix concise, we include in the paper only the experimental setup needed to interpret the figures and qualitative trends, while full implementation details, parameter settings, and figure-generation scripts are provided in the repository for transparency and reproducibility.

To provide a controlled validation of the mechanisms discussed in the main text, we construct a synthetic layered derivation environment. The environment is represented as a directed acyclic graph, where layer-0 nodes correspond to base premises, higher-layer nodes correspond to intermediate derived formulas, and top-layer nodes are treated as queries. A query is considered derivable if its dependency chain can be reconstructed from the currently available premises and cached intermediates.

Premise loss is simulated by randomly removing a fraction of layer-0 premises before derivation. Memory pollution is simulated by injecting spurious items into the available memory pool, thereby increasing retrieval burden and pollution-sensitive inference cost. Caching is modeled as storing a budgeted set of intermediate nodes that can later be reused as immediately available facts.

We compare two cache selection strategies. The first is a frequency-based selector that prioritizes intermediate nodes appearing frequently in query dependency closures. The second is a structurally informed robust selector that additionally favors higher-layer nodes, broader query coverage, and larger base-support compression. Intuitively, the robust selector is designed to store intermediates that summarize larger portions of the derivation graph and are therefore less sensitive to local corruption or retrieval noise.

We report three primary metrics: (i) the fraction of derivable queries, (ii) the average inference cost among derivable queries, and (iii) the net gain relative to a no-cache baseline after accounting for storage cost. All results are averaged over multiple random trials, and error bars indicate one standard deviation.

\begin{figure*}[t]
    \centering
    \begin{subfigure}[t]{0.48\textwidth}
        \centering
        \includegraphics[width=\linewidth]{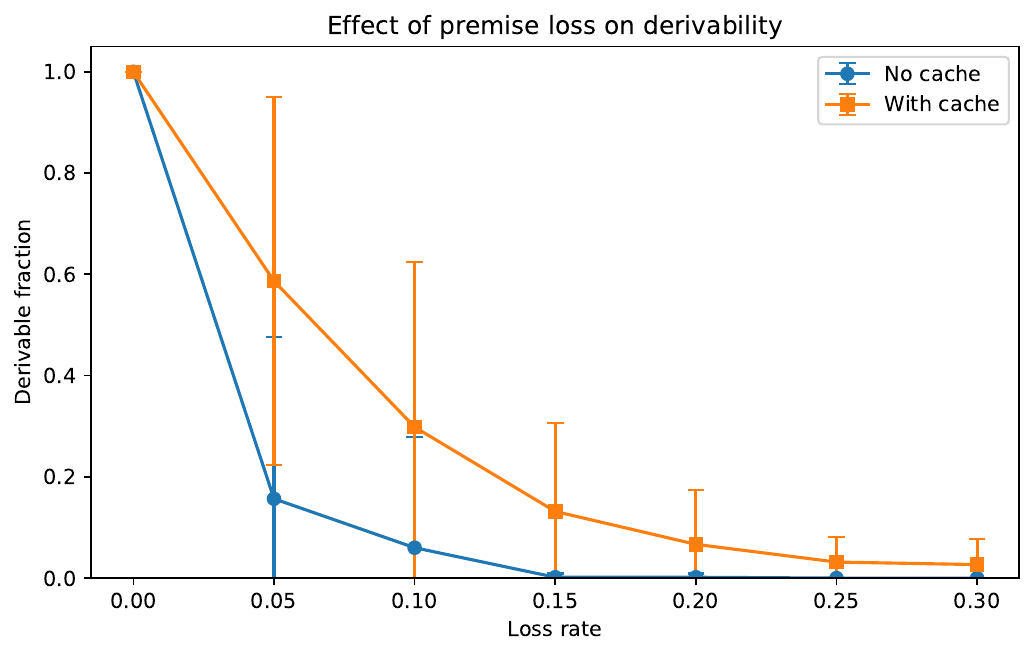}
        \caption{Derivable fraction under premise loss.}
        \label{fig:loss-derivability}
    \end{subfigure}
    \hfill
    \begin{subfigure}[t]{0.48\textwidth}
        \centering
        \includegraphics[width=\linewidth]{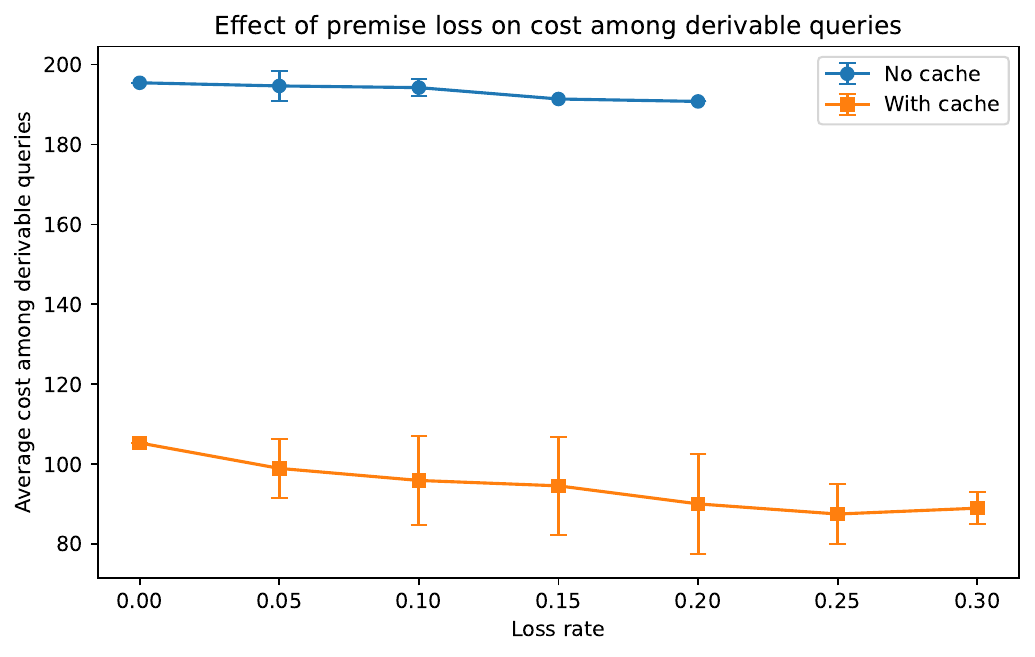}
        \caption{Average cost among derivable queries under premise loss.}
        \label{fig:loss-cost}
    \end{subfigure}
    \caption{Effects of premise loss in the synthetic derivation environment. 
    \textbf{Left:} As the premise loss rate increases, the fraction of derivable queries decreases sharply, while caching preserves a larger derivable set than the no-cache baseline. 
    \textbf{Right:} Among derivable queries, caching substantially reduces inference cost across the tested loss range.}
    \label{fig:appendix-loss}
\end{figure*}

\subsection{Effect of Premise Loss}
Figure \ref{fig:loss-derivability} shows the effect of premise loss on derivability. As the loss rate increases, the fraction of derivable queries drops sharply. However, across all nonzero loss settings, caching consistently preserves a larger fraction of derivable queries than the no-cache baseline. This is consistent with the view that cached intermediate structure can partially compensate for missing premises by shortening or bypassing portions of the original derivation chain.

Figure \ref{fig:loss-cost} reports the average inference cost among derivable queries. For all comparable loss settings, the cached condition yields substantially lower cost than the no-cache baseline. The slight downward trend of conditional cost at higher loss rates is expected: harder queries become underivable earlier, leaving a surviving subset of relatively easier derivations. Accordingly, this figure should be interpreted as a conditional cost analysis over derivable queries rather than as an overall expected cost over all queries.

\begin{figure*}[t]
    \centering
    \begin{subfigure}[t]{0.48\textwidth}
        \centering
        \includegraphics[width=\linewidth]{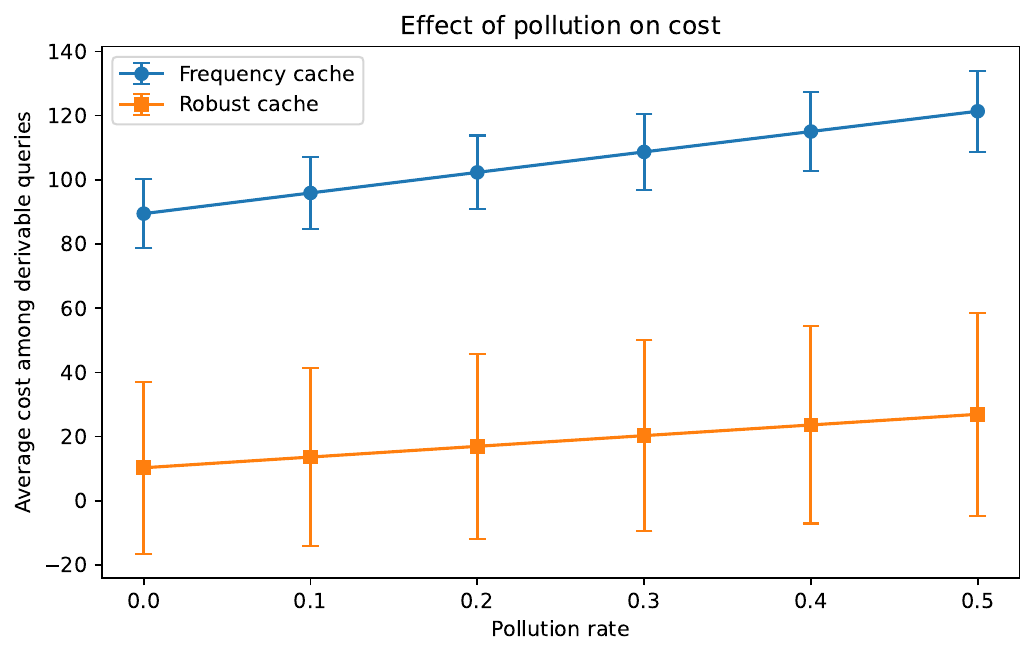}
        \caption{Average cost under increasing pollution.}
        \label{fig:pollution-cost}
    \end{subfigure}
    \hfill
    \begin{subfigure}[t]{0.48\textwidth}
        \centering
        \includegraphics[width=\linewidth]{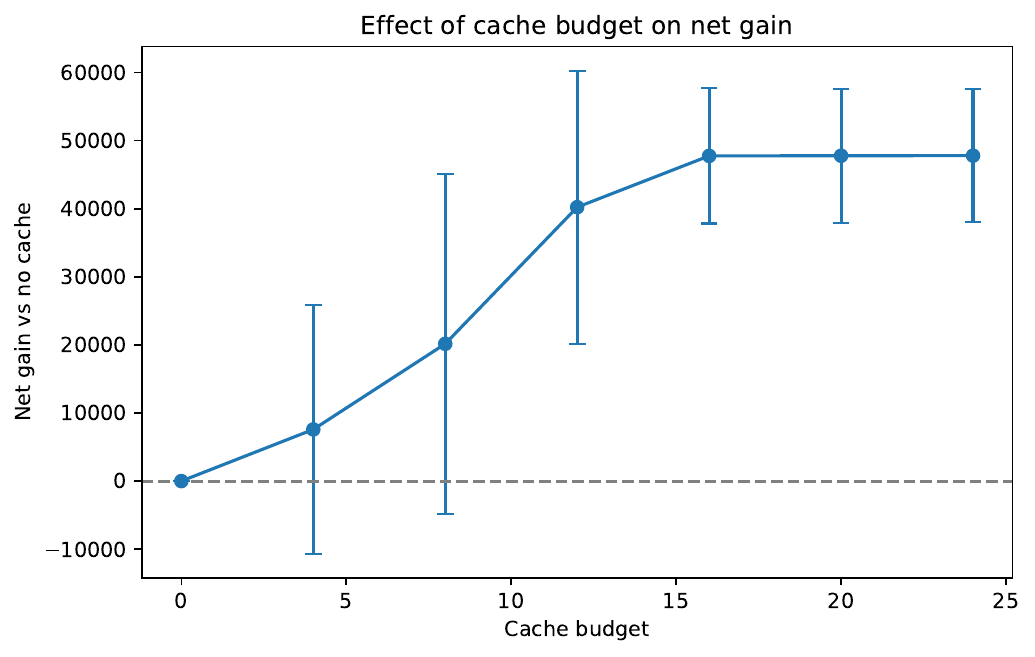}
        \caption{Net gain as a function of cache budget.}
        \label{fig:cache-gain}
    \end{subfigure}
    \caption{Effects of pollution and cache budget in the synthetic derivation environment. 
    \textbf{Left:} Increasing pollution raises inference cost under both cache-selection strategies, while the structurally informed robust selector consistently outperforms the frequency-based selector. 
    \textbf{Right:} The net gain of caching increases with cache budget, becomes reliably positive beyond a moderate threshold, and then exhibits diminishing returns.}
    \label{fig:appendix-pollution-gain}
\end{figure*}
\subsection{Effect of Pollution and Cache Strategy}
Figure \ref{fig:pollution-cost} shows the effect of memory pollution on the average cost among derivable queries. As pollution increases, inference cost rises under both cache-selection strategies. At the same time, the structurally informed robust selector consistently achieves lower mean cost than the frequency-based selector across the entire tested range. This suggests that memory quality depends not only on reuse frequency, but also on the structural role played by cached intermediates in supporting downstream derivations.

Although the robust selector exhibits larger variance across trials, its average performance remains substantially better than that of the frequency-based selector. This indicates that cache selection based on structural compression and query coverage can be more resilient to polluted memory than selection based on raw reuse counts alone.
\subsection{Cache Budget Trade-off}
Figure \ref{fig:cache-gain} shows the net gain of caching as a function of cache budget. At very small budgets, the gain is positive on average but remains unstable across trials. Once the cache budget exceeds a moderate threshold, caching reliably outperforms the no-cache baseline. At larger budgets, the gain begins to plateau, indicating diminishing returns once the most valuable intermediate nodes have already been stored. This behavior is consistent with the expected trade-off between memory investment and amortized reasoning benefit.

\subsection{Discussion}
Taken together, the synthetic experiments provide qualitative support for the framework developed in the main text. They show that: (i) premise loss degrades derivability; (ii) caching improves resilience and lowers conditional inference cost; (iii) memory pollution raises the effective burden of inference; and (iv) cache selection strategy and budget both materially affect the resulting trade-offs.

These experiments are intentionally synthetic and are not intended to serve as a task-level benchmark. Their role is instead to isolate the structural mechanisms of interest in a controllable setting. Therefore, the results should be interpreted as mechanistic validation and qualitative evidence, rather than as direct claims of downstream task superiority.

\bibliographystyle{IEEEtran}
\bibliography{ref}

%% Displayed equations can be tagged using various environments. 
%% Single line equations can be tagged using the equation environment.

%% Unnumbered equations are tagged using starred versions of the environment.
%% amsmath package needs to be loaded for the starred version of equation environment.

%% align or eqnarray environments can be used for multi line equations.
%% & is used to mark alignment points in equations.
%% \\ is used to end a row in a multiline equation.

%% Refer following link for more details.
%% https://en.wikibooks.org/wiki/LaTeX/Mathematics
%% https://en.wikibooks.org/wiki/LaTeX/Advanced_Mathematics

%% Use a table environment to create tables.
%% Refer following link for more details.
%% https://en.wikibooks.org/wiki/LaTeX/Tables

%% Use figure environment to create figures
%% Refer following link for more details.
%% https://en.wikibooks.org/wiki/LaTeX/Floats,_Figures_and_Captions


%%       \cite{<label>} ==> [1]




\end{document}